\documentclass[12pt,a4paper]{article}

\usepackage[T1]{fontenc}
\usepackage{lmodern}[lmr]

\usepackage{bm, amsfonts, amsmath, amssymb, amsthm, enumitem, float, mathrsfs, mathtools, multirow, nccmath, rotating, tocloft, setspace, subfig}

\usepackage[mathcal]{euscript}
\usepackage[dvipsnames]{xcolor}
\usepackage{tikz}
\usetikzlibrary{arrows, shapes}
\usepackage{hyperref}
\hypersetup{colorlinks,
linkcolor={MidnightBlue},
citecolor={RoyalBlue},
urlcolor={Blue}}
\usepackage[title]{appendix}

\usepackage[nameinlink, noabbrev, capitalize]{cleveref}
\usepackage[hmargin=0.9in, vmargin=0.9in]{geometry}
\usepackage[round]{natbib}
\def\be{\begin{equation}}
\def\ee{\end{equation}}
\def\bea{\begin{eqnarray}}
\def\eea{\end{eqnarray}}

\author{}
\title{}
 
\DeclareMathOperator*{\argmin}{\arg\!\min}

\DeclareMathOperator*{\tr}{\normalfont\textrm{trace}}

\newtheorem{assumption}{Assumption}
\newtheorem{corollary}{Corollary}

\newtheorem{example}{Example}
\newtheorem{lemma}{Lemma}
\newtheorem{proposition}{Proposition}

\newtheorem{theorem}{Theorem}

\renewcommand{\arraystretch}{1.20}

\newcommand\blfootnote[1]{
\begingroup
\renewcommand\thefootnote{}\footnote{#1}
\addtocounter{footnote}{-1}
\endgroup
}

\allowdisplaybreaks[4]

\begin{document}

\setlength{\abovedisplayskip}{8pt}
\setlength{\belowdisplayskip}{8pt}

\begin{titlepage}

\begin{center}

\begin{spacing}{1.5}
{\Large  \textbf{Panel Data Estimation and Inference: \\Homogeneity versus Heterogeneity}}
\end{spacing}

\bigskip

$^{\ast}${\sc Jiti Gao}, $^{\dag}${\sc Fei Liu}, $^{\ast}${\sc Bin Peng} and $^{\ddag}${\sc Yayi Yan}\blfootnote{Gao and Peng would like to acknowledge the Australian Research Council Discovery Projects Program for its financial support under Grant Numbers: DP250100063.  Liu's research was financially supported by National Natural Science Foundation of China under Grant Number 72203114. Yan acknowledges the financial support by the NSFC under the grant number 72303142 and the Fundamental Research Funds for the Central Universities under grant numbers 2022110877 and 2023110099. The authors contributed equally to this paper and are credited in alphabetical order.}   

\bigskip

$^{\ast}$Monash University

\medskip

$^{\dag}$Nankai University
 
\medskip

$^\ddag$Shanghai University of Finance and Economics

\bigskip

\today

\end{center}

\begin{abstract}
 
In this paper, we define an underlying data generating process that allows for different magnitudes of cross-sectional dependence, along with time series autocorrelation. This is achieved via high-dimensional moving average processes of infinite order (HDMA($\infty$)). Our setup and investigation integrates and enhances homogenous and heterogeneous panel data estimation and testing in a unified way.  To study HDMA($\infty$), we extend the Beveridge-Nelson decomposition to a high-dimensional time series setting, and derive a complete toolkit set. We exam homogeneity versus heterogeneity using Gaussian approximation, a prevalent technique for establishing uniform inference. For post-testing inference, we derive central limit theorems through Edgeworth expansions for both homogenous and heterogeneous settings. Additionally, we showcase the practical relevance of the established asymptotic theory by (1). connecting our results with the literature on grouping structure analysis, (2). examining a nonstationary panel data generating process, and (3). revisiting the common correlated effects (CCE) estimators. Finally, we verify our theoretical findings via extensive numerical studies using both simulated and real datasets.

\medskip
	

\medskip
	
	
\end{abstract}

\end{titlepage}

\section{Introduction}\label{Sec1}

Panel data analysis has seen its popularity in the past thirty years or so. Comprehensive reviews have been conducted at different stages, while the literature evolves. See, for example, \cite{ARELLANO20013229}, \cite{Petersen2009}, \cite{CP2015}, \cite{hsiao2022analysis}, etc. Among all challenges raised in different surveys, this article aims to offer a unified framework and a set of toolkit to 

\begin{enumerate}[leftmargin=24pt, parsep=2pt, topsep=2pt]
    \item simultaneously test homogeneity (i.e., $\mathbb{H}_0$) vs. heterogeneity (i.e., $\mathbb{H}_1$);
    \item develop valid inference under either $\mathbb{H}_0$ or $\mathbb{H}_1$, and account for both weak and strong cross-sectional dependence (WCD and SCD), as well as the dependence along the time dimension.
\end{enumerate}
In what follows, we review the relevant literature, point out the challenges, and then highlight our contributions.

We start with the hypothesis testing about homogeneity against heterogeneity, which has always been a central topic in empirical studies. Without loss of generality, we consider a simple setup as follows:

\begin{eqnarray}\label{def.xit}
    x_{it}= \mu_i + \epsilon_{it},
\end{eqnarray}
where $x_{it}$, $\mu_i$, and $\epsilon_{it}$ are all scalars, $(i,t)\in [N]\times [T]$, $[L]\coloneqq \{1,\ldots,L \}$ for any given positive integer $L$, $N$ stands for the number of individuals, and $T$ stands for the total number of periods.
More often than not, the key hypotheses are 

\begin{eqnarray}\label{def.test}
&& \mathbb{H}_0:  \mu_i = \mu \text{ for all } i; \quad\quad\mathbb{H}_1:  \mu_i\neq \mu \text{ for some } i.
\end{eqnarray}
Sometimes, the hypotheses are even more straightforward, e.g.,  

\begin{eqnarray*}
&& \mathbb{H}_0:  \mu_i = 0 \text{ for all } i; \quad\quad \mathbb{H}_1:  \mu_i\neq 0 \text{ for some } i.
\end{eqnarray*}
The model \eqref{def.xit} extends the location model of \cite{Lazarus} to panel data settings. To settle \eqref{def.test}, a large literature (\citealp{PY2008,GAO2020329}; and references therein) adopts the quadratic test statistics. However, as explained in \cite{fan2015power}, the tests based on quadratic forms often suffer from low powers, and fail to detect the sparse alternatives. Therefore, \cite{fan2015power} and \cite{YU2024105458} provide power enhanced test statistics to tackle this issue. 

To the best of our knowledge, the above literature largely (if not all) ignores the dependence along both dimensions. As time series autocorrelation (TSA) has been well discussed in the literature (see \citealp{FanYao}; \citealp{Gao2007}), we justify the necessity of accounting for the cross-sectional dependence (CD) here. As surveyed by \cite{CP2015}, CD is likely to be the rule rather than the exception, and it sometimes goes beyond WCD due to omitted variables as documented in \cite{GX2021}. It is then reasonable to call for a complete toolkit set that is robust to the presence of dependence. 

To better present our motivations, we start with four datasets, of which $N$ and $T$ stand for the number of individuals and the number of time periods respectively. 

\begin{itemize}[leftmargin=12pt, parsep=2pt, topsep=2pt]
    \item[] \textbf{Dataset 1}: the U.S. macroeconomic dataset assembled by \cite{MN2016}.
    
    \item[] \textbf{Dataset 2}: the climate data of 37 stations from  the U.K. Meteorological Office.  
    
    \item[] \textbf{Dataset 3}: the bank equity return data constructed by \cite{Baron2021}.
    
    \item[] \textbf{Dataset 4}: the realized volatility data of 16 international stock markets.
\end{itemize}
We shall provide numerical evidences demonstrating that the four datasets exemplify both WCD and SCD, highlighting the need for a unified framework to model these varying dependencies. 

For each dataset, we observe $\{x_{it}\}$ as defined in \eqref{def.xit}, which yields the pairwise correlation $r_{ij}$ for $\forall i,j\in [N]$. For $\forall\tau \in[0,1]$, we can calculate:

\begin{eqnarray*}
    p(\tau) =\frac{1}{(N-1)N/2}\sum_{i = 2}^{N}\sum_{j=1}^{i-1} I(|r_{ij}| > \tau),
\end{eqnarray*}
where  $I(\cdot)$ stands for the indicator function. This allows us to plot $p(\tau)$ against $\tau$, as shown in Figure \ref{FG1} below. Here, $\tau$ represents a specific correlation threshold, while $p(\tau)$ measures the percentage of absolute correlation values that exceed the threshold. We also calculate the following measure:

\begin{eqnarray}\label{def.rho}
    \overline{\rho} =\frac{1}{N}\sum_{i=1}^N\sum_{j=1}^N |r_{ij}|,
\end{eqnarray}
which is adopted from Assumption C of \cite{BN2002}, and represents the magnitude of cross-sectional dependence of a panel dataset. Having presented these measures, we proceed.

\textbf{Dataset 1} --- We examine a time period spanning from October 2003 to September 2023, resulting in a total of $T = 240$ observations along the time dimension. After removing variables with missing values, we are left with $N = 127$ macro variables. In the first sub-figure of Figure \ref{FG1}, we observe that around 20\% of the absolute correlations $\{|r_{ij}|\mid i>j\}$  are greater than 0.8, and roughly more than 50\% of these correlations have absolute values exceeding 0.5. Additionally, $\overline{\rho}=66$ is approximately $N/2$. Thus, we find a high degree of correlation among these macro variables. An intuitive thought is that many of these macro variables are generated by the same set of unobservable shocks. 

\textbf{Dataset 2} --- We analyze temperature and sunshine data from the U.K. Meteorological office. There are 37 stations in total, widely distributed across the U.K. Each station reports both temperature and sunshine monthly, resulting in $N = 74$ individual time series. After handling missing values, we focus on the period from January 1950 to February 2023, resulting in $T = 878$ observations. In the second sub-figure of Figure \ref{FG1}, over 80\% of the absolute correlations $\{|r_{ij}|\mid i>j\}$  are greater than 0.5, and approximately 30\% of these correlations exceed 0.8. We have $\overline{\rho}=52$, which is almost the same as the sample size of the individual dimension. Thus, it is evident that these climate data exhibit a high degree of correlation.

\textbf{Dataset 3} --- We study a dataset of real bank equity returns for 46 advanced and emerging economies ($N=46$) in the period from 1870 to 2016 ($T=147$). This dataset, constructed by \cite{Baron2021}, illustrates the influence of banking crises on subsequent output gaps and credit contractions.  The estimated values for $p(\tau)$ are presented in the third sub-figure of Figure \ref{FG1}. Obviously, this dataset has the weakest cross-sectional dependence among the four datasets, and most of $r_{ij}$'s are less than 0.6, which might be a signal of WCD. The value $\overline{\rho}$ of $\eqref{def.rho}$ is 13 which is relatively small. 

\textbf{Dataset 4} --- A final example is a realized volatility dataset, which exhibits strong connections across different equity markets. This dataset consists of realized volatility data for 16 international stock market indices ($N=16$) and is computed using tick-by-tick stock index data from Refinitiv DataScope Select. The dataset covers 3809 common trading days ($T=3809$) over a 16-year period from January 4, 2005, to February 26, 2021. For this dataset, we compute the pairwise correlations between realized volatilities  and present the estimated values of $p(\tau)$ in the fourth sub-figure of Figure \ref{FG1}. More than 50\% of the correlations are above 0.5, indicating a strong connectedness in realized volatility across markets. The value $\overline{\rho}$ of $\eqref{def.rho}$ is 8.

\begin{figure}[htbp!]
\centering\includegraphics[scale = 0.25]{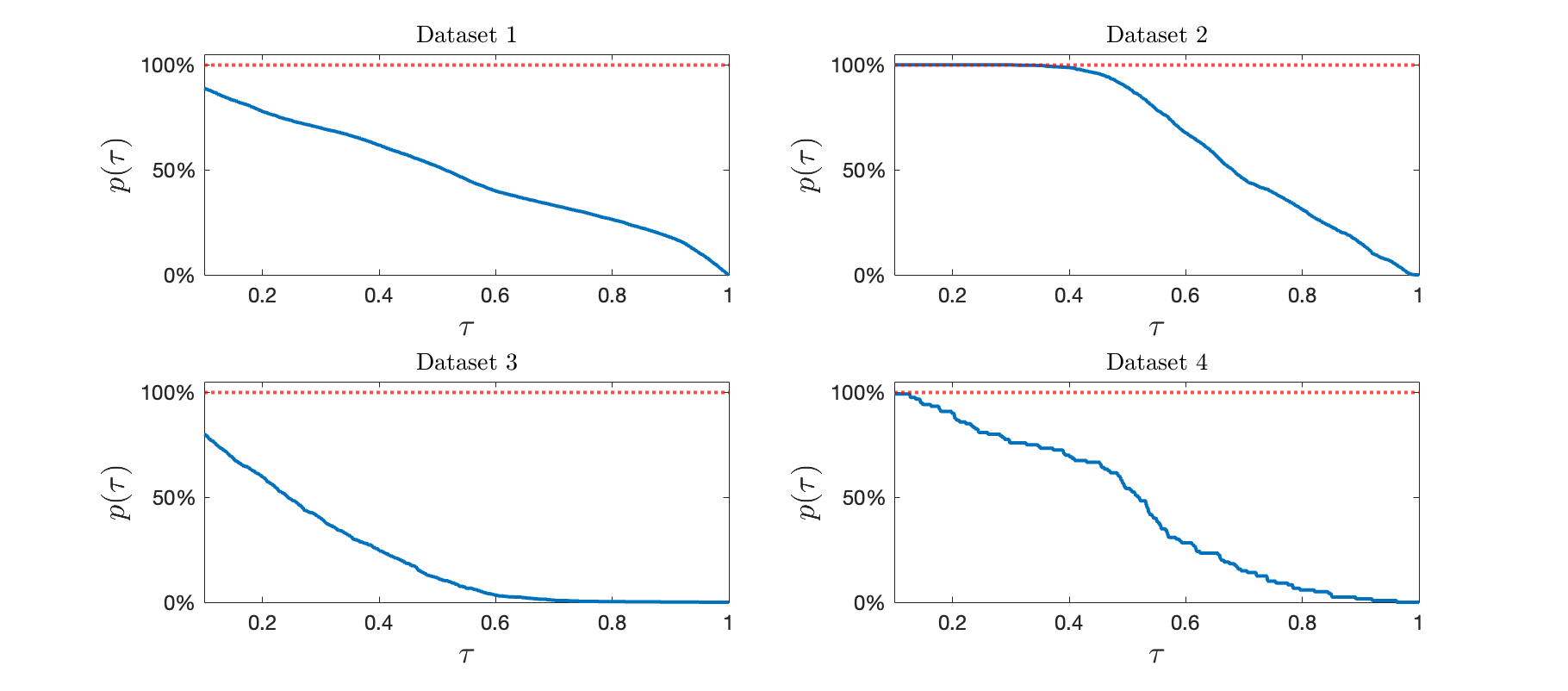}
\caption{$p(\tau)$'s of Datasets 1-4}\label{FG1}
\end{figure}
   
Four datasets offer examples of CD with different magnitude. The issue at hand is certainly a cause for concern, as inferring $\mu_i$ of \eqref{def.xit} is a cornerstone of data analysis (e.g., \citealp[Chapter 3]{Hamilton1994}; \citealp{Lazarus}; \citealp{CP2024}). Without a grasp of the dependence magnitude, the existing literature offers scant guidance on how to infer the homogenous/heterogeneous means, let alone more intricate scenarios including homogenous/heterogeneous trends (\citealp{ROBINSON20124,WSX2023}). See Example \ref{EX1} and Example \ref{EX2} of Section \ref{Sec2} for the purpose of demonstration.  To our knowledge, although \cite{CP2015} formalize the definitions of WCD and SCD, only Assumption 3 of \cite{goncalves_2011} addresses this issue using a set of high-level conditions. Yet, the underlying data generating mechanism remains underexplored.  In related research, \cite{ROBINSON2011}, \cite{ROBINSON2012} and \cite{LEE2016}  impose a linear system structure to represent cross-sectional dependence. This approach, while convenient in theory, is infeasible in practice due to the fact that there is no natural ordering in the cross-sectional dimension. When solely WCD is of interest, Assumption C of \cite{BN2002}, which regulates dependence across cross-sections and time using various moments, is often cited. Nonetheless, the underlying data generating process remains somewhat obscure.

Considering the aforementioned points, our contributions in this paper are as follows:

\begin{enumerate}[leftmargin=24pt, parsep=2pt, topsep=2pt]
    \item First, we define an underlying data generating process that allows for different magnitude of CD, along with TSA. This is achieved via high-dimensional moving average processes of infinite order (HDMA($\infty$)), which automatically generalizes the spatial structure introduced by Robinson and his co-authors in recent years. The framework is important in the sense that as noted by \citet[p. 187]{brockwell1991time} and \citet[pp. 33 \& 190]{FanYao}, the Wold decomposition theorem  ensures a formal linear representation exists for any stationary time series with no deterministic components, and HDMA($\infty$) naturally incorporates this result into a panel data framework. 
    
    \item To the best of our knowledge, HDMA($\infty$) has not been carefully explored in the literature of panel data analysis. Our setup and investigation significantly integrates and enhances both homogenous and heterogeneous panel data modelling and testing (such as \citealp{Pesaran2006, PY2008, fan2015power, YU2024105458}).  To study HDMA($\infty$), we extend the BN decomposition (e.g., \citealp{BN1981,PS1992}) to a high-dimensional time series setting, and derive a complete set of toolkit. 
    
    Our development offers theoretical justification for some high level assumptions of the literature (e.g., \citealp[Assumption 3]{goncalves_2011}, \citealp[Assumption C]{BN2002}), and also generalizes Assumption 2 of \cite{Pesaran2006} and many extensions since then. Additionally, our investigation complements the work of \cite{fan2015power}, who specifically study cases where $\frac{T}{\sqrt{N}}\to 0$, by considering a broader range of scenarios and relaxing the independence assumptions employed in \cite{PY2008} and \cite{YU2024105458}.
    
    \item We exam homogeneity against heterogeneity using Gaussian approximation, a prevalent technique for establishing uniform inference (e.g., \citealp{chernozhuokov2022improved}, and references therein). For post-testing inference, we derive Central Limit theorems through Edgeworth expansions for both homogenous and heterogeneous settings. Notably, the demand for Gaussian approximation in panel data analysis has been increasing recently, as exemplified in Section 4 of \cite{SJW_2024} and Section 4 of \cite{LLS2024}. Our study also contributes to this research direction by providing a set of foundational conditions and deriving a set of useful basic results.
    
    \item We showcase the practical relevance of the established asymptotic theory by (1). connecting our results with the literature on grouping structure analysis such as those surveyed in \cite{BM2015} and \cite{SSP2016}, (2). examining a nonstationary panel data generating process presented in \cite{PM1999}, and (3). revisiting the common correlated effects (CCE) estimators of \cite{Pesaran2006}. Typically, when investigating nonstationary panel data, one has to impose cross-sectional independence such as \cite{PM1999}, \cite{DGP2021} and \cite{HUANG2021198} due to technical constraints. Our study offers a set of complete toolkit to account for the dependence of unit root precesses.
    
    \item Finally, we evaluate our theoretical findings via extensive numerical studies using both simulated and real datasets.
\end{enumerate}

The remainder of this paper is structured as follows. In Section \ref{Sec2}, we present the underlying data generating process in detail, and show its practical relevance. The corresponding asymptotic properties under homogenous (i.e., $\mathbb{H}_0$ of \eqref{def.test}) and heterogeneous (i.e., $\mathbb{H}_1$ of \eqref{def.test}) settings are given in Sections \ref{Sec2.1} and \ref{Sec2.2} respectively. Building on Sections \ref{Sec2.1} and \ref{Sec2.2}, we provide the test statistic to exam \eqref{def.test} in Section \ref{Sec2.3}. In Section \ref{Sec3}, we revisit the CCE estimators of \cite{Pesaran2006}, and some results presented in \cite{PM1999} to showcase the practical relevance of the results of Section \ref{Sec2}. Section \ref{Sec4} conducts extensive simulation studies to exam our theoretical results. An empirical study is given in Section \ref{Sec5} to exam whether the rational expectations of financial markets are homogeneous or heterogeneous. Section \ref{Sec6} concludes with a few remarks. The preliminary lemmas and proofs are regulated to the online appendices.

\textbf{Notations} --- Before proceeding, we introduce some notations and present a few useful facts to facilitate development. Throughout, vectors and matrices are always in bold font. We let $\mathsf{i}$ be the imaginary unit; for a matrix $\mathbf{A}=\{a_{ij}\}_{m\times n}$, let $\mathbf{A}^+$ define Moore-Penrose inverse, and let 

\begin{eqnarray*}
&&\|\mathbf{A}\|_2=(\lambda_{\max}\{\mathbf{A}^\top \mathbf{A}\} )^{1/2},\quad\|\mathbf{A}\|_1 =\max_{j\in [n]}\sum_{i=1}^m |a_{ij}|,\notag \\
&&\|\mathbf{A}\|_\infty =\max_{i\in [m]}\sum_{j=1}^n |a_{ij}|,\quad |\mathbf{A}|_p =\left( \sum_{i=1}^m\sum_{j=1}^n a_{ij}^p\right)^{1/p} \quad\text{for}\quad p\ge 1,
\end{eqnarray*}
define respectively its Spectral norm, column norm, row norm, and entry wise norm; moreover, we always use $^\dag$ and $^\sharp$ to represent the column and row of a matrix, e.g.,

\[
\mathbf{A} =(\mathbf{a}_1^\dag,\ldots,\mathbf{a}_n^\dag)=(\mathbf{a}_1^\sharp,\ldots,\mathbf{a}_m^\sharp)^\top .
\]
Given two conformable matrices $\mathbf{A}$ and $\mathbf{B}$, we let $\mathbf{A}\circ \mathbf{B}$ denote its Hadamard product. For a vector $\mathbf{v} =(v_1,\ldots, v_p)^\top$, we let $|\mathbf{v}|_\infty\coloneqq \max_{i}v_i$. $\mathbf{I}_p$ stands for a $p\times p$ identify matrix, and when no misunderstanding arises, we write $\mathbf{I}$. $\mathbf{1}_p$ stands for a $p\times 1$ vector of ones. $\mathbf{e}_j$ always stands for a selection column vector with the $j^{th}$ individual being 1 and others being 0. For two positive constants $a$ and $b$, $a\asymp b$ stands for $a=O(b)$ and $b=O(a)$; for $a,b\in \mathbb{R}$, $a\wedge b=\min \{a,b\}$ and $a\vee b=\max \{a,b\}$. We always let $\Phi(x)$ and $\phi(x)$ be the CDF and PDF of the standard normal distribution, and let $\widetilde{\phi}(x)\coloneqq \exp(-x^2/2)$ for notational simplicity. Thus, $\sqrt{2\pi}\phi(x)=\widetilde{\phi}(x)$. The cumulant generating function of a random variable $x$ is defined by $C(u)=\log E[\exp(u x)],$ and we have

\begin{eqnarray*}
\kappa_r = C^{(r)}(0) =(-\mathsf{i})^r\frac{\mathrm{d}^r}{\mathrm{d} u^r}\log\psi(u)|_{u=0},
\end{eqnarray*}
where $\psi(u)$ defines the characteristic function of $x$. Finally, $E^*(\cdot)$ and $\text{Pr}^*(\cdot)$ always refer to the operations induced by the sample space.

\section{The Setup and Asymptotic Properties}\label{Sec2}

We firstly present the main results achieved in this paper, and then outline the establishment of some key results before presenting any asymptotic results.

\begin{figure}[H]
\centering
\begin{tikzpicture}
\usetikzlibrary{arrows}
\usetikzlibrary{shapes}

\tikzstyle{annot} = [text width=4em, text centered]

\node[rectangle, minimum size = 6mm] (Ho-1) at (-6.2,1) {\footnotesize Homogeneous};
\node[rectangle, minimum size = 6mm] (He-1) at (-6.2,-0.5) {\footnotesize Heterogeneous};
\node[rectangle, text width=3.5cm] (Un-1) at (3.9,0.8) {\scriptsize Uniform Inference via Gaussian Approximation};

\node[rectangle, text width=2.1cm] (Eg-1) at (-3,2.4) {\footnotesize Edgeworth Expansion};
\node[rectangle, text width=2.3cm] (Eg-2) at (0.7,2.4) {\footnotesize Bootstrap Statistics};

\node[ellipse, text width=1.8cm, text height=0.2cm, fill=teal!10] (Ho-2) at (-3,1) {\footnotesize
Theorem \ref{THM.1}  
Corollary \ref{COL.1}};

\node[ellipse, text width=1.8cm, text height=0.2cm, fill=orange!30] (Ho-3) at (0.5,1) {\footnotesize
Lemma \ref{LM.A7}  
Theorem \ref{THM.2}};

\node[ellipse, text width=1.8cm, text height=0.2cm, fill=teal!10] (He-2) at (-3,-0.5) {\footnotesize
Theorem \ref{THM.3}};

\node[ellipse, text width=1.8cm, text height=0.2cm, fill=orange!30] (He-3) at (0.5,-0.5) {\footnotesize
Lemma \ref{LM.A8}  
Theorem \ref{THM.4}};

\node[ellipse, text width=1.8cm, text height=0.2cm, fill=blue!10] (Un-2) at (3.9,-0.5) {\footnotesize
Theorems \ref{THM.5} \& \ref{THM.6}};

\end{tikzpicture}
\caption*{Main Results Achieved}
\end{figure}

\begin{enumerate}
    \item \textit{Edgeworth expansion for both homogeneous and heterogeneous cases}. Our derivation relies on the high-dimensional Beveridge-Nelson (BN) decomposition, as panel data can be viewed as a high-dimensional (HD) time series when stacked along the cross-sectional dimension. The main structure of the development for Theorems \ref{THM.1} and \ref{THM.3} is established by equation \eqref{def.chiu} and Lemma \ref{LM.A2}, while some essential technical details are provided in Lemmas \ref{LM.A4} and \ref{LM.A3}. 
    
    The strategy  is the same for both homogeneous and heterogeneous cases overall, although in the heterogeneous case we only take summation over the time dimension for each individual. Our development offers theoretical justification for some high level assumptions of the literature (e.g., \citealp[Assumption 3]{goncalves_2011}, \citealp[Assumption C]{BN2002}), and also generalize Assumption 2 of \cite{Pesaran2006} and many extensions since then. 
  
    \item \textit{Bootstrap statistics for both homogeneous and heterogeneous cases}.  The key technique for investigating our bootstrap statistics relies on the properties of the \(m\)-dependent time series. In deriving Edgeworth expansions for the homogeneous and heterogeneous bootstrap statistics, we adopt the proof strategy of \cite{tik1981}, which decomposes the \(m\)-dependent series into a sequence of segments. Elements within each segment may remain dependent, while different segments are 1-dependent conditional on the observed data. 
    
    By applying the asymptotic results for each segment (established in Lemmas \ref{LM.A7} and \ref{LM.A8}) and exploiting the independence of nonadjacent segments, we extend the Edgeworth expansion results from independent data to the \(m\)-dependent bootstrap series in Theorems \ref{THM.2} and \ref{THM.4} for the homogeneous and heterogeneous statistics, respectively.
    
    \item \textit{Gaussian approximation for uniform inference}. Applying the martingale decomposition technique to HDMA($\infty$), we establish that its partial sum processes can be uniformly approximated by the summation of independent vectors with negligible approximation errors. Building on this result, the Gaussian approximation theorem for high-dimensional independent random vectors developed in \cite{chernozhuokov2022improved} is used to justify the validity of Gaussian approximation for HDMA($\infty$) in Theorem \ref{THM.5}. 
    
    Based on Theorem \ref{THM.5}, the continuity of the maximum of the Gaussian distribution and the robust inference of high-dimensional covariance matrix estimation proposed in \cite{gao2024robust}, we derive high-dimensional Gaussian multiplier bootstrap approximations in Theorem \ref{THM.6} for inference purpose.
     
\end{enumerate}

\medskip

To proceed, we recall the notations defined in the end of Section \ref{Sec1}, and present a result which is independent of the assumptions to be adopted, and facilitates the Edgeworth expansion for both homogeneous and heterogeneous cases.

\begin{lemma}\label{LM.A4}
Let $\{ H_n(x)\mid n\ge 0\}$ be Probabilist's Hermite polynomials. The Fourier transformation of $\phi(x) H_n(x)$ is $ \widetilde{\phi}(x)(\mathsf{i}x)^n$.
\end{lemma}

The result is built on the generating function of the Probabilist's Hermite polynomials, and, together with \eqref{def.chiu}, naturally connects cumulants of a random variable with the density function of a standard normal distribution.

\medskip

We are now ready to formulate our ideas. To exam \eqref{def.test}, we need to have a good understanding about both homogenous (i.e., $\mathbb{H}_0$) and heterogeneous (i.e., $\mathbb{H}_1$) cases. The investigation does not only offer post-testing inference under either $\mathbb{H}_0$ or $\mathbb{H}_1$, but also helps to establish the test statistic. Having said that, we respectively investigate \eqref{def.xit} under the null $\mathbb{H}_0$ in Section \ref{Sec2.1}, and under the alternative $\mathbb{H}_1$ in Section \ref{Sec2.2}. In Section \ref{Sec2.3}, we assemble the results of both sections to finalize the test statistic. 

\medskip

Firstly, we explain the necessity of accounting for WCD and SCD in a unified framework. For simplicity, suppose that $\mathbb{H}_0$ holds, and define the following HDMA($\infty$) process for $\{x_{it}\}$ of \eqref{def.xit}: 

\begin{eqnarray}\label{def.xt_1}
\mathbf{x}_t =\mu \cdot \mathbf{1}_N+\mathbf{B}(L)\pmb{\varepsilon}_t,
\end{eqnarray}
where $\mathbf{x}_t\coloneqq (x_{1t},\ldots, x_{Nt})^\top$, $\mathbf{B}(L)\coloneqq \sum_{\ell=0}^{\infty} \mathbf{B}_\ell L^\ell$ with $L$ being the lag operator, $\{\mathbf{B}_\ell\mid \ell\ge 0\}$ is a set of $N\times N$ matrices, $\pmb{\varepsilon}_t =(\varepsilon_{1t},\ldots, \varepsilon_{Nt})^\top$, and  $\{\varepsilon_{it}\}$ are independent and identically distributed (i.i.d.) over both $(i,t)$ with mean 0. Each $\mathbf{B}_\ell$ admits the following representation:

\begin{eqnarray}\label{def.B_ell}
    \mathbf{B}_\ell =(\mathbf{b}_{\ell 1}^\dag,\ldots, \mathbf{b}_{\ell N}^\dag)=(\mathbf{b}_{\ell 1}^\sharp,\ldots, \mathbf{b}_{\ell N}^\sharp)^\top.
\end{eqnarray}

The HDMA($\infty$) of \eqref{def.xt_1} offers the flexibility to account for different types of dependence. For example, simple algebra (Appendix \ref{SecA1}) shows three types of dependence as follows:

\begin{itemize}[leftmargin=24pt, parsep=2pt, topsep=2pt]
\item[] CD: $\text{Cov}(x_{i1},x_{j1}) = \sum_{\ell =0}^{\infty} \mathbf{b}_{\ell i}^{\sharp\top} \mathbf{b}_{\ell j}^\sharp$;
\item[] TSA: $\text{Cov}(x_{it},x_{is}) = \sum_{\ell=0}^{\infty} \mathbf{b}_{\ell+t-s, i}^{\sharp\top} \mathbf{b}_{\ell i}^\sharp$ for $t>s$;
\item[] $\text{CD}+\text{TSA}$: $\text{Cov}(x_{it},x_{js}) =\sum_{\ell=0}^{\infty} \mathbf{b}_{\ell+t-s, i}^{\sharp\top} \mathbf{b}_{\ell j}^\sharp$  for $t>s$.
\end{itemize}
Loosely speaking, we require the following conditions to hold: (i) CD does not necessarily shrink as $|i-j|$ increases, which is evident in view of Example \ref{EX2} below; (ii) TSA shrinks as $t-s$ increases, and may vary with respect to $i$; (iii) $\text{CD}+\text{TSA}$ inherits the properties of (i) and (ii).

To see the difference between WCD and SCD, we provide the following examples.

\medskip

\begin{example}\label{EX1}
When $\mathbf{B}_0=\mathbf{I}$ and  $\mathbf{B}_\ell=\mathbf{0}$ for $\ell\ge 1$, we have $\mathbf{x}_t=\mu\cdot \mathbf{1}_N+ \pmb{\varepsilon}_t$, which gives a set of i.i.d. panel data over both dimensions. This is an extreme case of WCD, and 

\begin{eqnarray*}
\frac{1}{\sqrt{NT}}\sum_{i=1}^N\sum_{t=1}^T(x_{it}-\mu) =\frac{1}{\sqrt{NT}}\sum_{t=1}^T\sum_{i=1}^N\varepsilon_{it},
\end{eqnarray*}
which will help us to drive the asymptotic distribution. In this case $\|\mathbf{B}_0\|_2=1<\infty$.
\end{example}

\begin{example}\label{EX2} 
Suppose that $\mathbf{x}_t$ is generated as follows. 

\begin{eqnarray*}
\mathbf{x}_t =\mu\cdot \mathbf{1}_N + \begin{pmatrix}
    \frac{1}{\sqrt{N}} & \cdots & \frac{1}{\sqrt{N}} \\
    \vdots & \ddots & \vdots \\
    \frac{1}{\sqrt{N}} & \cdots & \frac{1}{\sqrt{N}} \\
\end{pmatrix}  \pmb{\varepsilon}_t,
\end{eqnarray*}
which infers $\mathbf{B}_0 =\mathbf{1}_N (\mathbf{1}_N^\top/\sqrt{N}) $ and  $\mathbf{B}_\ell=\mathbf{0}$ for $\ell\ge 1$. Apparently, each time series $\{x_{it}\mid t\in[T]\}$ is perfectly correlated with the others, and simple algebra gives that

\begin{eqnarray*}
    \frac{1}{N\sqrt{T}}\sum_{i=1}^N\sum_{t=1}^T(x_{it} - \mu ) =\frac{1}{\sqrt{NT}}\sum_{t=1}^T\sum_{i=1}^N \varepsilon_{it}.
\end{eqnarray*}
Therefore, SCD requires a different normalizer to derive the asymptotic distribution, which directly affects the construction of the confidence interval in practice. Notably, $\|\mathbf{B}_0\|_2=\sqrt{\lambda_{\max} (\frac{1}{N}\mathbf{1}_N \mathbf{1}_N^\top\mathbf{1}_N  \mathbf{1}_N^\top )}=\sqrt{N}$, which diverges as $N\to \infty$. 
\end{example}

A few findings emerge in view of both examples. First, \eqref{def.xt_1} extends \citet[Assumption A10]{ROBINSON2011} and offers detailed data generating process for  \citet[Assumption 3]{goncalves_2011}. Second, inferring $\mu$ under WCD and SCD respectively requires different normalizers to construct the standard deviations. A challenge arises naturally, as one needs to decide the magnitude of CD prior to analysis.  Third, the magnitude of cross-sectional dependence will mathematically influence the Spectral norm of $\mathbf{B}_\ell$'s from the modelling perspective. In what follows, we shall account for these findings in our investigation.

\subsection{Inference under $\mathbb{H}_0$}\label{Sec2.1}

Firstly, we infer the homogenous mean via the following statistic:

\begin{eqnarray}\label{def.SNT}
    \widetilde{S}_{NT}\coloneqq\frac{1}{\sigma_x\sqrt{L_N T}}S_{NT}\coloneqq\frac{1}{\sigma_x\sqrt{L_N T}} \sum_{i=1}^N\sum_{t=1}^T(x_{it}-\mu),
\end{eqnarray}
where $L_N$ is generic notation, and varies with respect to the magnitude of CD. Obviously, $L_N=N$ for Example \ref{EX1}, and $L_N=N^2$ for Example \ref{EX2}. Here,  $\sigma_x^2 \coloneqq \lim_{N,T}\frac{1}{L_N T}E[S_{NT}^2]$. 

To facilitate development, we present the first assumption.

\begin{assumption}\label{AS1}
\item

\begin{enumerate}[leftmargin=24pt, parsep=2pt, topsep=2pt]
\item $\{\varepsilon_{it}\}$ are  i.i.d.  over both $i$ and $t$, and satisfy that $E[\varepsilon_{it}]=0$ and $E[\varepsilon_{it}^2]=1$. In addition, $\varepsilon_{11}$ has the characteristic function $\psi(u)\coloneqq E[\exp(\mathsf{i}u\varepsilon_{11})]$ with $u\in \mathbb{R}$, and has cumulants $\kappa_r$ for $r\in [J]$ with a fixed $J\ (\ge 4)$.

\item Suppose that for $L_N\in [N, N^2]$, 

\begin{enumerate}[leftmargin=24pt, parsep=2pt, topsep=2pt]
\item $\limsup_{N}\sum_{\ell =0}^{\infty}\ell \cdot C_{N\ell}<\infty$, where $C_{N\ell}\coloneqq  \sqrt{\frac{N}{L_N}}\|\mathbf{B}_\ell\|_2$;

\item $\limsup_{N}\sqrt{\frac{N}{L_N}}\|\mathbf{B} \|_1<\infty$, where $\mathbf{B}\coloneqq \mathbf{B}(1)$.
\end{enumerate}

\end{enumerate}
\end{assumption}

Assumption \ref{AS1}.1 is standard. See \citet[Chapter 1]{saulis} for detailed definition of cumulants. In general, $\mathbf{B}_\ell\ne \mathbf{0}$ for $\ell \ge 1$, so we need to regulate the elements of $\mathbf{B}_\ell$ using Assumption \ref{AS1}.2, which ensures that after suitable normalization, the Spectral norms of $\mathbf{B}_\ell$'s are summable. Obviously, Assumption \ref{AS1}.2 is satisfied for both Examples \ref{EX1} and \ref{EX2}. While $L_N\in (N, N^2)$, the magnitude of cross-sectional dependence is in between both cases. By Lemma \ref{LM.A5} of the appendix, Assumption \ref{AS1} entails that 

\begin{eqnarray*}
\sigma_x^2=\lim_{N}\frac{1}{L_N}\mathbf{1}_N^\top\mathbf{B} \mathbf{B} ^\top \mathbf{1}_N ,
\end{eqnarray*}
which does not only incorporate long run covariance along the time dimension, but also accounts for WCD and SCD automatically. 

Having these conditions in hand, we present the first main result of this paper.

\begin{theorem}\label{THM.1}
Under Assumption \ref{AS1}, as $(N,T)\to (\infty,\infty)$,  
\begin{eqnarray*}
\sup_{u\in\mathbb{R}}\left|F_{NT}(u)-\Phi(u)-\frac{\beta_3}{6} (1-u^2)\phi(u) \right|=O\left(NT\left(\frac{\|\mathbf{B} \|_1}{\sqrt{L_NT}}\right)^4 \vee \frac{1}{T^2}\right),
\end{eqnarray*}
where $F_{NT}(u)\coloneqq \Pr(\widetilde{S}_{NT}\le u) $, $\beta_3$ is defined in \eqref{def_betar3} for the sake of presentation, and 
\[
|\beta_3|=O\left(NT\left(\frac{\|\mathbf{B} \|_1}{\sqrt{L_NT}}\right)^3 \vee \frac{1}{T^{3/2}}\right).
\]If $\kappa_3= 0$, then $\beta_3= 0$.
\end{theorem}

Theorem \ref{THM.1} gives an Edgeworth expansion for $\widetilde{S}_{NT}$, and a considerably simplified form will be presented in Corollary \ref{COL.1} under a slightly more restrictive condition. The detailed definition of $\beta_3$ is omitted here due to its cumbersome notation. It is worth mentioning that if the skewness of $\varepsilon_{it}$ is 0 (i.e., $\kappa_3= 0$), the term $\frac{\beta_3}{6} (1-u^2)\phi(u)$ automatically vanishes. Then $F_{NT}(u)$ converges to $\Phi(u)$ in a much faster rate. In general, we do not have $\kappa_3= 0$, so we keep the statement of Theorem \ref{THM.1} as it is.

It is noteworthy that for either Example \ref{EX1} or Example \ref{EX2}, the result reduces to 

\begin{eqnarray*}
\sup_{u\in\mathbb{R}}\left|F_{NT}(u)-\Phi(u)-\frac{\beta_3}{6} (1-u^2)\phi(u) \right|=O\left( \frac{1}{NT}\vee \frac{1}{T^2}\right),
\end{eqnarray*}
where $|\beta_3|=O(\frac{1}{\sqrt{NT}} \vee \frac{1}{T^{3/2}})$. Sequentially, for both examples, Theorem \ref{THM.1} gives the following Berry-Esseen bound:

\begin{eqnarray*}
\sup_{u\in\mathbb{R}}\left|F_{NT}(u)-\Phi(u) \right|=O\left(\frac{1}{\sqrt{T}(\sqrt{N}\wedge T)}\right),
\end{eqnarray*}
which says that, for the panel data with TSA and WCD/SCD, the Berry-Esseen bound is usually at the order of $\frac{1}{\sqrt{NT}}$ unless $N\gg T$ such that $\frac{T}{\sqrt{N}}\to 0$. The presence of the term $\frac{1}{T^{3/2}}$ is due to the BN decomposition which introduces a truncation residual along the time dimension. In a typical panel data setting $N\asymp T$, the truncation residual is negligible. Our result complements the work of \cite{fan2015power}, who specifically study cases where $\frac{T}{\sqrt{N}}\to 0$, by considering a broader range of scenarios and relaxing the independence assumptions employed in \cite{PY2008} and \cite{YU2024105458}.

With a minor additional restriction, the following corollary holds.

\begin{corollary}\label{COL.1}
Under Assumption \ref{AS1}, as $(N,T)\to (\infty,\infty)$ and $\frac{N}{T^2}\to 0$,
\begin{eqnarray*}
\sup_{u\in\mathbb{R}}\left|F_{NT}(u)-\Phi(u)-\frac{\beta_3^*}{6} (1-u^2)\phi(u) \right|=O_P\left(\frac{\sqrt{N}\|\mathbf{B}\|_1^2}{L_NT}\right),
\end{eqnarray*}
where 
\[
\beta_3^*\coloneqq \frac{\kappa_3}{\sigma_x^{3/2}L_N^{3/2}T^{1/2}} \sum_{j=1}^N (\mathbf{1}_N^\top \mathbf{b}_j^\dag )^3
\] with $\mathbf{b}_j^\dag $ being the $j^{th}$ column of $\mathbf{B}$, and $\kappa_3$ is defined in Assumption \ref{AS1}.1.
\end{corollary}

The condition $\frac{N}{T^2}\to 0$ is rather common in the literature of panel data analysis, and is apparently fulfilled given $N\asymp T$. Using Assumption \ref{AS1}.2, it is straightforward to see that $\frac{\sqrt{N}\|\mathbf{B}\|_1^2}{L_NT} =\frac{N\|\mathbf{B}\|_1^2}{L_N}\cdot \frac{1}{\sqrt{N}T}\to 0.$ 

\medskip
 
\noindent \textbf{Bootstrap Inference} --- Below, we provide a bootstrap procedure to conduct inference. Given the nature of panel data, an intuitive thought is to draw a set of random variables, say, $\{\zeta_{it}\mid i\in [N], t\in [T] \}$ under certain restrictions, and construct the bootstrap counterpart of $S_{NT}$ as follows:

\begin{eqnarray*}
\widetilde{S}_{NT}^*\coloneqq\frac{1}{\sigma_x\sqrt{L_NT}}\sum_{i=1}^N \sum_{t=1}^T(x_{it}-\mu)\zeta_{it}.
\end{eqnarray*}
It is then straightforward to obtain that

\begin{eqnarray}
E^*(\widetilde{S}_{NT}^{*2}) &=&\frac{1}{\sigma_x^2L_NT} \sum_{t=1}^T (\mathbf{x}_{t}-\mu\cdot \mathbf{1}_N)^\top \mathbf{W}_{0}(\mathbf{x}_{t}-\mu\cdot \mathbf{1}_N)\label{def.Sstar_1} \\
&& +\frac{2}{\sigma_x^2L_NT}\sum_{k=1}^{T-1}\sum_{t=1}^{T-k} (\mathbf{x}_{t}-\mu\cdot \mathbf{1}_N)^\top \mathbf{W}_{t,t+k}(\mathbf{x}_{t+k}-\mu\cdot \mathbf{1}_N),\label{def.Sstar_2}
\end{eqnarray}
where $\pmb{\zeta}_t=(\zeta_{1t},\ldots, \zeta_{Nt})^\top$,  $\mathbf{W}_{ts} =E[\pmb{\zeta}_t\pmb{\zeta}_s^\top]$ for $t\ne s$, and $\mathbf{W}_{0}\coloneqq  \mathbf{W}_{tt}$ for simplicity. Correspondingly, we can write $\widetilde{S}_{NT}^2 $ as follows:

\begin{eqnarray}
\widetilde{S}_{NT}^2 &=& \frac{1}{\sigma_x^2L_NT} \sum_{t=1}^T (\mathbf{x}_{t}-\mu\cdot \mathbf{1}_N)^\top \mathbf{1}_N\mathbf{1}_N^\top (\mathbf{x}_{t}-\mu\cdot \mathbf{1}_N) \label{def.S_1} \\
&&+ \frac{2}{\sigma_x^2 L_NT}\sum_{k=1}^{T-1}\sum_{t=1}^{T-k} (\mathbf{x}_{t}-\mu\cdot \mathbf{1}_N)^\top \mathbf{1}_N  \mathbf{1}_N^\top(\mathbf{x}_{t+k}-\mu\cdot \mathbf{1}_N).\label{def.S_2}
\end{eqnarray}
Comparing the right hand sides of $E^*(\widetilde{S}_{NT}^{*2})$ and $\widetilde{S}_{NT}^2$, it is obvious that we need  \eqref{def.Sstar_1} and \eqref{def.Sstar_2} to mimic \eqref{def.S_1} and \eqref{def.S_2} respectively. As $\mathbf{1}_N  \mathbf{1}_N^\top $ has rank one, by Lemma \ref{LM.A1} the most obvious form of $\mathbf{W}_{ts}$ should be

\begin{eqnarray*}
\mathbf{W}_{ts} = \alpha_{ts}\mathbf{1}_N \mathbf{1}_N^\top,
\end{eqnarray*}
where $\alpha_{ts}$ is a scalar varying with respect to the distance between $t$ and $s$. Therefore, we conclude that, to have a valid bootstrap procedure, one should replace $\zeta_{it}$ with $\zeta_t$ to ensure there is no cross-sectional variation. The finding nicely fits our study, as we assume  no prior information about the magnitude of CD. By doing so, we preserve the dependence along the individual dimension in the bootstrap draws, and also avoid imposing certain order on individuals implicitly.

Formally, the bootstrap procedure is as follows.

\medskip

\hrule
\begin{enumerate}[leftmargin=24pt, parsep=2pt, topsep=2pt]
\item Draw the bootstrap version of $S_{NT}$ by

\begin{eqnarray*}
\widetilde{S}_{NT}^*\coloneqq\frac{1}{\sigma_x\sqrt{L_NT}}\sum_{i=1}^N\sum_{t=1}^T(x_{it}-\mu)\zeta_{t},
\end{eqnarray*}
where $\zeta_{t}$ is an $m$-dependent time series.

\item Repeat the above procedure $R$ times to obtain the sampling distribution of $\widetilde{S}_{NT}^*$.
\end{enumerate}
\hrule

\medskip

Accordingly, we impose the following assumption.

\begin{assumption}\label{AS2}
\item 

\begin{enumerate}[leftmargin=24pt, parsep=2pt, topsep=2pt]
\item Let $a(\cdot)$ be a symmetric Lipschitz continuous kernel defined on $[-1, 1]$ such that $a(0)=1$ and  $\lim_{|x|\to 0 } \frac{1-a(x)}{|x|^{q_a}} = C_{q_a}$ for $q_a \in \{1,2\}$ and  $0 < C_{q_a} < \infty$. Assume $\limsup_{N}\sum_{\ell =0}^{\infty}\ell^{q_a} \cdot C_{N\ell}<\infty$, and $\int_{-1}^1a(u) \exp(-\mathsf{i}u x) \mathrm{d}u\ge 0$ for $x\in \mathbb{R}$.
    
\item Let $E[\zeta_{t}]=0$, $E[\zeta^2_{t}]=1$,  $E|\zeta_{t}|^4<\infty$, and $E[\zeta_{t}\zeta_{s}]=a(\frac{t-s}{m})$ for $\forall t, s\in [T]$, where $m$ satisfies $m\rightarrow\infty $ and $\frac{m}{\sqrt{T}}\rightarrow 0$ as $T\rightarrow\infty$.
    
\end{enumerate}
\end{assumption}

Assumption \ref{AS2} is a typical assumption in the literature of dependent wild bootstrap. We refer interested readers to \cite{Shao2015} for a comprehensive review of this line of research. Several conventional kernel functions satisfy such conditions. For example, for the Bartlett kernel, $q_a= 1$ and $C_1 = 1$; for the Parzen, Tukey-Hanning, QS kernels, and the trapezoidal functions, $q_a= 2$ and the values of $C_2$ vary and all satisfy $C_2<\infty$. See \cite{Andrews1991} for comments on different kernel functions. Notably, when the Bartlett kernel is adopted, the condition $\limsup_{N}\sum_{\ell =0}^{\infty}\ell^{q_a} \cdot C_{N\ell}<\infty$ reduces to Assumption \ref{AS1}.2.a.

Using Assumptions \ref{AS1}-\ref{AS2}, the following lemma and theorem hold for the bootstrap procedure.

\begin{lemma}\label{LM.A7}  
Under Assumptions \ref{AS1}-\ref{AS2}, as $(N,T)\to (\infty,\infty)$, the following results hold:
    
\begin{enumerate}[leftmargin=24pt, parsep=2pt, topsep=2pt]  
\item $E^*[(\widetilde{S}_{NT}^*)^2]= 1+o_P(1)$;

\item $E^\ast[(\sum_{t=(s-1)m+1}^{sm} \mathbf{1}_N^\top \mathbf{x}_t\zeta_t)^2]=O_P(mL_N)$, for  $s=1,\ldots,\lfloor \frac{T}{m} \rfloor$;

\item $E^\ast[(\sum_{t=\lfloor\frac{T}{m}\rfloor m+1}^{T} \mathbf{1}_N^\top \mathbf{x}_t\zeta_t)^2]=O_P(mL_N)$.
\end{enumerate}
\end{lemma}

Lemma \ref{LM.A7} summarizes some basic statistical properties to facilitate the investigation on the bootstrap statistic for the homogeneous cases. With them in hand, we present the following theorem.

\begin{theorem}\label{THM.2}
Under Assumptions \ref{AS1}-\ref{AS2}, as $(N,T)\to (\infty,\infty)$,

\begin{enumerate}[leftmargin=24pt, parsep=2pt, topsep=2pt]
\item $\sup_{u\in \mathbb{R}}\left|\text{\normalfont Pr}^*(\widetilde{S}_{NT}^* \le u ) -\Phi(u)-\frac{1}{6\widetilde{\sigma}^{\ast3}}(1-u^2)E^\ast [\widetilde{S}_{NT}^{*3}]\phi(x)\right|=O_P\left(\frac{m}{T}\right)$;
\item $\sup_{u\in \mathbb{R}}\left|\text{\normalfont Pr}^*(\widetilde{S}_{NT}^* \le u) -\Pr(\widetilde{S}_{NT}\le u)\right|=O_P\left(\sqrt{\frac{m}{T}}\right)$;
\item $\text{\normalfont MSE}(\widetilde{\sigma}^{\ast 2})=\frac{2m}{T}\int_{-1}^1 a^2(u)du+\frac{C_{q_a}^2}{ \sigma_x^4m^{2q_a}}\Delta_{q_a}^2+o_P(m^{-2q_a})+o\left(\frac{m}{T}\right)$;
\end{enumerate}   
where $\widetilde{\sigma}^{\ast2}= E^\ast[\widetilde{S}_{NT}^{*2}]$ and $\Delta_{q_a}=L_N^{-1}\sum_{s=-\infty}^{\infty}\sum_{\ell=0}^{\infty} |s|^{q_a}\mathbf{1}_N^\top \mathbf{B}_\ell  \mathbf{B}_{\ell+|s|}^\top \mathbf{1}_N $.
\end{theorem}
    
Theorem \ref{THM.2}.1 establishes the first-order Edgeworth expansion for the bootstrap statistic $\widetilde{S}_{NT}^*$, extending the results of \cite{tik1981} to the panel data framework. Conditional on the sample, the rates of the first two results of Theorem \ref{THM.2} are optimal, as the bootstrap draws $\{\zeta_t\}$ are $m$-dependent time series data only. Theorem \ref{THM.2}.3 demonstrates that the bootstrap covariance estimator can consistently estimate the true covariance, and infers that  MSE is minimized at  
$$m_{\text{opt}}= \big(\frac{C_{q_a}\Delta_{q_a}}{2\sigma_x^4\int_{-1}^1 a^2(u)du}\big)^{2/(2q_a+1)} T^{1/(2q_a+1)}.$$

\subsection{Inference under $\mathbb{H}_1$}\label{Sec2.2}

Under $\mathbb{H}_1$, the model \eqref{def.xit} admits the following vector form:

\begin{eqnarray}\label{def.xt_2}
\mathbf{x}_t =\pmb{\mu}+\mathbf{B}(L)\pmb{\varepsilon}_t,
\end{eqnarray}
where $\pmb{\mu}=(\mu_1,\ldots, \mu_N)^\top$, and the rest settings are identical to those in \eqref{def.xt_1}.

Inferring heterogeneity is slightly more complicated, as there are two options:

\begin{enumerate}[leftmargin=24pt, parsep=2pt, topsep=2pt]
\item Infer a specific individual $\mu_i$;
\item Infer $\pmb{\mu}$ as a whole.
\end{enumerate}
It is noteworthy that, when studing each $\mu_i$, the conditional expectations and probabilities involved in Lemma \ref{LM.A8} and Theorem \ref{THM.4} are still random variables, we therefore do not take $\max$ over $i$ in these results. In order to study the uniform inference over all individuals (i.e., inferring $\pmb{\mu}$ as a whole), we provide Theorems \ref{THM.5} and \ref{THM.6}, which will help us establish the corresponding test statistic in Section \ref{Sec2.3}.

\medskip

We start with the first choice, and consider the following quantify:

\begin{eqnarray*}
F_i(u)\coloneqq \Pr(\widetilde{p}_i\le u),
\end{eqnarray*}
where $\widetilde{p}_i\coloneqq \frac{1}{\sigma_{p,i}\sqrt{T}}p_i\coloneqq \frac{1}{\sigma_{p,i}\sqrt{T}} \sum_{t=1}^T(\mathbf{e}_i^\top \mathbf{x}_t-\mu_i)$, $\sigma_{p,i}^2 \coloneqq \lim_N\mathbf{e}_i^\top \mathbf{B}\mathbf{B}^\top \mathbf{e}_i>0$, and $\mathbf{e}_i$ is a selection vector as defined in Section \ref{Sec1}.

To proceed, we need more structures to investigate heterogeneity.

\begin{assumption}\label{AS3}
Suppose that $\max_i\sum_{\ell=1}^{\infty}\ell^2 \|\mathbf{b}_{\ell i}^{\sharp}\|_2^2<\infty$ and $\max_i\|\mathbf{b}_{i}^{\sharp}\|_2<\infty$, where $\mathbf{b}_{\ell i}^{\sharp}$ is defined in \eqref{def.B_ell}, and $\mathbf{B}=(\mathbf{b}_{1}^{\sharp},\ldots, \mathbf{b}_{N}^{\sharp})^\top$.
\end{assumption}

Assumption \ref{AS3} regulates the rows of $\mathbf{B}_\ell$'s and $\mathbf{B}$, and is rather minor in view of Examples \ref{EX1} and \ref{EX2}. Under this condition, we are able to present the following theorem.

\begin{theorem}\label{THM.3}
Under Assumptions \ref{AS1} and \ref{AS3}, as $(N,T)\to (\infty, \infty)$, 
\begin{eqnarray*}
\max_{i}\sup_{u\in\mathbb{R}}\left|F_i(u)-\Phi(u)-\frac{\beta_{i3}}{6} (1-u^2)\phi(u) \right|=O\left( \frac{1}{T}\right),
\end{eqnarray*}
where $\max_{i}|\beta_{i3}|=O( \frac{1}{T^{1/2}})$. If $\kappa_3= 0$, then $\beta_{i3}= 0$.
\end{theorem}

Again, if the skewness of $\varepsilon_{it}$ is 0, the term $\frac{\beta_{i3}}{6} (1-u^2)\phi(u) $ vanishes. Then $F_i(u)$ converges to $\Phi(u)$ at a faster rate. In general, we do not have $\kappa_3= 0$, so we keep the statement of Theorem \ref{THM.3} as it is. 

To infer the distribution of $\widetilde{p}_i$ in practice, we provide a heterogeneous version of the dependent wild bootstrap procedure. 

\medskip

\hrule
\begin{enumerate}[leftmargin=24pt, parsep=2pt, topsep=2pt]
\item Draw the bootstrap version of $\widetilde{p}_{i}$ by

\begin{eqnarray*}
    \widetilde{p}^\ast_{i}\coloneqq  \frac{1}{\sigma_{p,i}\sqrt{T}}\sum_{t=1}^T(\mathbf{e}_i^\top \mathbf{x}_t-\mu_i)\zeta_{t},
\end{eqnarray*}
where $\zeta_{t}$ is an $m$-dependent time series satisfying Assumption \ref{AS2}.

\item Repeat the above procedure $R$ times to obtain the sampling distribution of $\widetilde{p}^\ast_{i}$.
\end{enumerate}
\hrule

\medskip
\medskip

For the heterogeneous case, we present the following lemma to facilitate the development of the asymptotic distribution of $\widetilde{p}^\ast_{i}$ conditional on the sample.

\begin{lemma}\label{LM.A8}  
Under Assumptions \ref{AS1}-\ref{AS3}, as $(N,T)\to (\infty,\infty)$, for each $i$,
    
\begin{enumerate}[leftmargin=24pt, parsep=2pt, topsep=2pt]  
\item $ E^*[(\widetilde{p}_{i}^*)^2]= 1+o_P(1)$;

\item $ E^\ast[(\sum_{t=(s-1)m+1}^{sm} \mathbf{e}_i^\top \mathbf{x}_t\zeta_t)^2]=O_P(mL_N)$, for  $s=1,\ldots,\lfloor \frac{T}{m} \rfloor$;

\item $ E^\ast[(\sum_{t=\lfloor\frac{T}{m}\rfloor m+1}^{T} \mathbf{e}_i^\top \mathbf{x}_t\zeta_t)^2]=O_P(mL_N)$.
\end{enumerate}
\end{lemma}

Using Lemma \ref{LM.A8}, we are able to establish the asymptotic distribution of $\widetilde{p}^\ast_{i}$ conditional on the sample, which is an extension of Theorem \ref{THM.2} for the heterogeneous bootstrap statistics.

\begin{theorem}\label{THM.4}
Under Assumptions \ref{AS1}-\ref{AS3}, as $(N,T)\to (\infty,\infty)$, for each $i$,

\begin{enumerate}[leftmargin=24pt, parsep=2pt, topsep=2pt]
\item $\sup_{u\in \mathbb{R}}\left|\text{\normalfont Pr}^*(\widetilde{p}^\ast_{i} \le u ) -\Phi(u)-\frac{1}{6\widetilde{\sigma}_i^{\ast3}}(1-u^2)E^\ast [\widetilde{p}_{i}^{*3}]\phi(x)\right|=O_P\left(\frac{m}{T}\right)$;
\item $\sup_{u\in \mathbb{R}}\left|\text{\normalfont Pr}^*(\widetilde{p}^\ast_{i} \le u) -F_i(u)\right|=O_P\left(\sqrt{\frac{m}{T}}\right)$;
\item $\text{\normalfont MSE}(\widetilde{\sigma}_i^{\ast 2})=\frac{2m}{T}\int_{-1}^1 a^2(u)du+\frac{C_{q_a}^2}{ \sigma_{p,i}^4m^{2q_a}}\Delta_{q_a,i}^2+o_P(m^{-2q_a})+o\left(\frac{m}{T}\right)$;
\end{enumerate}   
where $\widetilde{\sigma}_i^{\ast 2}= E^\ast[\widetilde{p}_{i}^{*2}]$ and $\Delta_{q_a,i}=\sum_{s=-\infty}^{\infty}\sum_{\ell=0}^{\infty} |s|^{q_a}\mathbf{e}_i^\top \mathbf{B}_\ell  \mathbf{B}_{\ell+|s|}^\top \mathbf{e}_i$.
\end{theorem}

Theorems \ref{THM.3} and \ref{THM.4} jointly imply that we are able to infer every single $\mu_i$. The discussion under Theorem \ref{THM.2} still applies here. 

\medskip

We then explore the second option under \eqref{def.xt_2} when inferring heterogeneity. Mathematically, it means that one is concerned with $\max_i \widetilde{p}_i$ rather than any individual $\widetilde{p}_i$.  This is useful, as there is an increasing literature concerning about the uniform inference in the panel data setting (e.g., \citealp{LLS2024, SJW_2024}). The following result contributes to this line of research.

\begin{theorem}\label{THM.5}
Let $\{\mathbf{z}_t\mid t\in [T]\}$ be independent Gaussian random vectors in $\mathbb{R}^{N}$ such that $E[\mathbf{z}_t]=\mathbf{0}$ and $\mathrm{Var}(\mathbf{z}_t) = \mathbf{B}\mathbf{B}^\top$. Let Assumption \ref{AS1}.1 hold and $\max_{i}\sum_{\ell=1}^{\infty}\ell \|\mathbf{b}_{\ell i}^{\sharp}\|_2<\infty$. As $(N,T)\to (\infty,\infty)$, 

\begin{eqnarray*}
&&\sup_{u\in \mathbb{R}}\left|\Pr\left(\left|\frac{1}{\sqrt{T}}\sum_{t=1}^{T}(\mathbf{x} _t-\pmb{\mu})\right|_{\infty} \leq u\right) - \Pr\left(\left|\frac{1}{\sqrt{T}}\sum_{t=1}^{T}\mathbf{z} _t\right|_{\infty} \leq u\right)\right| \notag \\
&=&O\left( \left(\frac{ N^{2/J}(\log N)^5}{T}\right)^{1/4}+\sqrt{\frac{N^{2/J}(\log N)^{3-2/J}}{T^{1-2/J}}}\right).
\end{eqnarray*}
\end{theorem}

Theorem \ref{THM.5} establishes a Gaussian approximation within the panel data framework. Notably, this approximation remains valid irrespective of whether the cross-sectional dependence is WCD or SCD. Also, it is worth mentioning that under the sub-Gaussian condition of $\varepsilon_{it}$, the term $N^{2/J}$ in the above theorem can be replaced by $\log N$. Thus if we impose an exponential tail assumption, $N$ can even diverge at an exponential rate of $T$.

To have a practically feasible version of Theorem \ref{THM.5}, we need to know $\mathrm{Var}(\mathbf{z}_t) = \mathbf{B}\mathbf{B}^\top \eqqcolon \pmb{\Omega}$. Thus, define the high-dimensional long-run covariance matrix estimator  by

\begin{eqnarray*}
    \widehat{\bm{\Omega}}\coloneqq \frac{1}{T}\sum_{t,s=1}^{T}a\left(\frac{t-s}{\widetilde{m}}\right)(\mathbf{x} _t-\overline{\mathbf{x}})(\mathbf{x}_s-\overline{\mathbf{x}})^\top \quad\text{with}\quad \overline{\mathbf{x}} =\frac{1}{T}\sum_{t=1}^T\mathbf{x}_t,
\end{eqnarray*}
and accordingly, define the Gaussian multiplier bootstrap approximate by

\begin{eqnarray*}
    \left|\widehat{\bm{\Omega}}^{1/2}\frac{1}{\sqrt{T}}\sum_{t=1}^{T}\mathbf{z}_t^* \right|_{\infty},
\end{eqnarray*}
where $\{\mathbf{z}_t^* \mid t\in [T]\}$ is a vector of i.i.d. $N$-dimensional Gaussian random variables with $\mathbf{z}_t^*\sim N(\mathbf{0},\mathbf{I}_N)$.

\begin{theorem}\label{THM.6}
Let Assumption \ref{AS1}.1 hold with $J>4$, and let Assumption  \ref{AS2}.1  hold. Suppose that (1) $\max_{i}\sum_{\ell=1}^{\infty}\ell^{q_\alpha}\|\mathbf{b}_{\ell i}^{\sharp}\|_2<\infty$, (2) $\frac{N^2T\log T}{(T\widetilde{m}\log N)^{J/4}}\to 0$, and (3) $\widetilde{m} \asymp T^{1/(2q_{\alpha}+1)}$. Then
{\small
\begin{eqnarray*}
&&\sup_{u\in \mathbb{R}}\left|\mathrm{Pr}^*\left(\left|\widehat{\bm{\Omega}}^{1/2}\frac{1}{\sqrt{T}}\sum_{t=1}^{T}\mathbf{z}_t^* \right|_{\infty} \leq u\right) - \Pr\left(\left|\frac{1}{\sqrt{T}}\sum_{t=1}^{T} (\mathbf{x} _t-\pmb{\mu}) \right|_{\infty} \leq u\right)\right| \notag \\
&=&O_P\left((\log N)^{5/4} (\widetilde{m}/T)^{1/4}\right) + O\left( \left(\frac{ N^{2/J}(\log N)^5}{T}\right)^{1/4}+\sqrt{\frac{N^{2/J}(\log N)^{3-2/J}}{T^{1-2/J}}}\right).
\end{eqnarray*}}
\end{theorem}

Theorem \ref{THM.6} builds upon the robust inference of high-dimensional covariance matrix estimation presented in \cite{gao2024robust} and the continuity of the maximum of a Gaussian distribution.  Similar to Theorem \ref{THM.5}, Theorem \ref{THM.6} accommodates various types of CD. 

The condition $\max_{i}\sum_{\ell=1}^{\infty}\ell^{q_\alpha}\|\mathbf{b}_{\ell i}^{\sharp}\|_2<\infty$ is slightly more restrictive than those required in Assumption \ref{AS3}. This is not surprising, as we now need to investigate the uniform inference rather than derive CLT for any individual. Nevertheless, we only require an algebraic decay rate of the temporal dependence. $\widetilde{m}\asymp T^{1/(2q_{\alpha}+1)}$ corresponds to the optimal level of bandwidth in terms of minimizing the asymptotic mean squared error of each element in $\widehat{\bm{\Omega}}$. The condition $\frac{N^2T\log T}{(T\widetilde{m}\log N)^{J/4}}\to 0$ imposes restrictions on $(\widetilde{m}, N, T)$ jointly, which are easy to realize. If $\varepsilon_{it}$ is sub-Gaussian, $J$ can be arbitrarily large. Then the restrictions can be much simplified.

In the following subsection, we show that the above results offer the theoretical framework for the purpose of inference.

\subsection{Test Statistic}\label{Sec2.3}

According to the results established in Section \ref{Sec2.1} and \ref{Sec2.2}, we are now ready to  consider an $L_\infty$-based test statistic:

$$
Q_{NT} = \max_i \sqrt{T}(\overline{x}_i - \overline{x}),
$$
where $\overline{x}_i = \frac{1}{T}\sum_{t=1}^{T}x_{it}$ and $\overline{x} = \frac{1}{NT}\sum_{i=1}^{N}\sum_{t=1}^{T}x_{it}$. Under $\mathbb{H}_0$, assume $N \leq L_N < N^2$ which implies that $\overline{x}-\mu = o_P(1/\sqrt{T})$, then simple algebra yields that
$$
Q_{NT} = \left|\frac{1}{\sqrt{T}}\sum_{t=1}^{T}(\mathbf{x}_t-\mu\cdot \mathbf{1}_N) \right|_{\infty}  + o_P(1).
$$
Hence, we can calculate the critical value of $Q_{NT}$ using Theorem \ref{THM.6}. In addition, our test statistic can detect a class of sparse local alternatives at the rate $T^{-1/2}$. Specifically, under the sparse local alternatives such that $\mathbb{H}_1:\ \mu_i = \mu + a_{T}\delta_i$ with $a_T = 1/\sqrt{T}$ and $|\delta_i| > 0$ for some $i$ (our $L_\infty$-based test is still powerful even only one individual violates the null hypothesis), we have
$$
Q_{NT} = \max_i\left(\frac{1}{\sqrt{T}}\sum_{t=1}^{T}(x_{it} - \mu) + \delta_i\right) + o_P(1),
$$
and thus $Q_{NT}$ diverges to infinity if $ \sqrt{T}a_T \to \infty$. We now state the following proposition.

\begin{proposition}\label{P0}
Let Assumptions \ref{AS1} and \ref{AS2}.1 hold with $J>4$. Additionally, suppose that (1) $\max_{i}\sum_{\ell=1}^{\infty}\ell^{q_\alpha}\|\mathbf{b}_{\ell i}^{\sharp}\|_2<\infty$, (2) $\frac{N^2T\log T}{(T\widetilde{m}\log N)^{J/4}}\to 0$, (3) $\widetilde{m} \asymp T^{1/(2q_{\alpha}+1)}$ and (4) $\overline{x} - \mu = o_P(1/\sqrt{T})$. Then under the null hypothesis
{\small
\begin{eqnarray*}
\sup_{u\in \mathbb{R}}\left|\Pr(Q_{NT}\le u)- \mathrm{Pr}^*\left(\left|\widehat{\bm{\Omega}}^{1/2}\frac{1}{\sqrt{T}}\sum_{t=1}^{T}\mathbf{z}_t^* \right|_{\infty} \leq u\right)\right|=o_P(1),
\end{eqnarray*}
where $\{\mathbf{z}_t^* \mid t\in [T]\}$} is a sequence of i.i.d. $N$-dimensional Gaussian random vectors with $\mathbf{z}_t^*\sim N(\mathbf{0},\mathbf{I}_N)$.
\end{proposition}

The distributional approximation established in Proposition 1 enables us to explore the known features of the partial sum of the bootstrapped versions of Gaussian samplers for inferential purposes.

\section{Extensions}\label{Sec3}

In this section, we consider three extensions. We firstly connect the above investigation with the literature on grouping structure analysis such as those surveyed in \cite{BM2015} and \cite{SSP2016}. In the second extension, we relax some restrictions imposed on the panel data unit root processes of \cite{PM1999}. Finally, we revisit the model of \cite{Pesaran2006} and the corresponding CCE estimators using the results established above. 

\subsection{Connection with Grouping Analysis}

In this subsection, we explain how the above study can be connected with the current literature on grouping structure analysis. The numerical implementation of grouping can be easily done via $K$-mean or the agglomerative hierarchical clustering algorithm (\citealp{hastie2009elements}). In what follows, we explain how the above results are connected with the grouping analysis.

Note that $\mu_i$ can be estimated by $\overline{x}_i$ as in Section \ref{Sec2.3}, so $\mu_i$ is already known approximately. To proceed, we introduce the grouping structure, and impose the necessary conditions. Denote $J_0$ sets of indices (i.e., $\mathscr{G}_1,\ldots, \mathscr{G}_{J_0}$) such that

\begin{eqnarray*}
    \cup_{j=1}^{J_0} \mathscr{G}_j=[N],\quad\mathscr{G}_{j_1} \cap\mathscr{G}_{j_2}=\emptyset \text{ for }j_1\ne j_2,\quad \text{and}\quad \sharp \mathscr{G}_{j} \asymp N \text{ for }\forall j,
\end{eqnarray*}
where $J_0 \ge 1$ is a fixed positive integer, and $\sharp \mathscr{G}_{j}$ stands for the cardinality of $\mathscr{G}_{j}$.  

\begin{assumption}\label{AS.group}
    \item Suppose that
    \begin{enumerate}[leftmargin=24pt, parsep=2pt, topsep=2pt]
        \item $\displaystyle \max_{j\in [J]}\max_{i\in \mathscr{G}_{j}} | \mu_{i}-\overline{\mu}_j|\le c_{N} $ with $c_{N}\to 0$, where $\overline{\mu}_j$'s are fixed numbers;
        \item $\displaystyle \min_{j_1\ne j_2}  | \overline{\mu}_{j_1}- \overline{\mu}_{j_2}|\ge c_0 >0$, where $c_0$ is a constant.
    \end{enumerate}
\end{assumption}

Typically, one assumes that $\mu_{i}\equiv \overline{\mu}_j$ for all $i\in \mathscr{G}_{j}$, so $c_{N}\equiv 0$. Assumption \ref{AS.group}.1 relaxes this restriction slightly. Assumption \ref{AS.group}.2 requires all groups to be well partitioned. If $J_0$ is known, we can estimate the group structure as follows:

\begin{eqnarray}\label{est.group}
    (\widehat{\mathcal{G}},\widehat{\pmb{\nu}}) =\argmin_{\mathcal{G}, \pmb{\nu}} S(\mathcal{G}, \pmb{\nu}) ,
\end{eqnarray}
where $S(\mathcal{G}, \pmb{\nu}) \coloneqq \frac{1}{N}\sum_{j=1}^{J_0}\sum_{i\in \mathcal{G}_j}| \overline{x}_i- \nu_j|^2$, $\mathcal{G} \coloneqq (\mathcal{G}_1,\ldots, \mathcal{G}_{J_0})$, $\pmb{\nu} \coloneqq(\nu_1,\ldots, \nu_{J_0})$, $\widehat{\mathcal{G}}\coloneqq (\widehat{\mathcal{G}}_1,\ldots, \widehat{\mathcal{G}}_{J_0})$, and $\widehat{\pmb{\nu}}\coloneqq (\widehat{\nu}_1,\ldots, \widehat{\nu}_{J_0})$. Here $\mathcal{G}_1,\ldots, \mathcal{G}_{J_0}$  are $J_0$ sets of indices such that 

\begin{eqnarray*}
    \cup_{j=1}^J \mathcal{G}_j=[N] ,\quad \text{and}\quad \mathcal{G}_{j_1} \cap\mathcal{G}_{j_2}=\emptyset \text{ for }j_1\ne j_2.
\end{eqnarray*}
In \eqref{est.group}, $\widehat{\mathcal{G}}$ and $\widehat{\pmb{\nu}}$ estimate $\mathscr{G}\coloneqq (\mathscr{G}_1,\ldots, \mathscr{G}_N)$ and $\overline{\pmb{\mu}} \coloneqq (\overline{\mu}_1,\ldots, \overline{\mu}_{J_0})$, respectively.

Since $J$ is not known in practice, we estimate $J_0$ using the following information criterion:

\begin{eqnarray}\label{est.J}
    \widehat{J}=\argmin_{J\le J^*} \text{IC}(J) ,
\end{eqnarray}
where $\text{IC}(J)\coloneqq S(\widehat{\mathcal{G}}_{\mid J}, \widehat{\pmb{\nu}}_{\mid J}) + \rho_{NT}J$, $J^*$ is a user-specified large fixed constant,  and $\rho_{NT}$ is a tuning parameter satisfying that 

\begin{eqnarray*}
    \rho_{NT}\to 0\quad\text{and}\quad \rho_{NT} /(\sqrt{\log(N)/T} +c_{N})\to \infty.
\end{eqnarray*}
In \eqref{est.J}, $\widehat{\pmb{\nu}}_{\mid J}\coloneqq (\widehat{\nu}_{1\mid J},\ldots,\widehat{\nu}_{J\mid J})$ and $\widehat{\mathcal{G}}_{\mid J}\coloneqq (\widehat{\mathcal{G}}_{1\mid J},\ldots, \widehat{\mathcal{G}}_{J\mid J})$ are obtained via \eqref{est.group} by assuming the number of groups being $J$. A natural choice of $\rho_{NT}$ is $[\log(N+T)]^{-1}$.

\begin{proposition}\label{LM.group}
    Let Assumptions \ref{AS1}-\ref{AS3} hold. Then the following results hold:

    \begin{enumerate}[leftmargin=24pt, parsep=2pt, topsep=2pt]
        \item $\displaystyle\max_{i\in [N]}|\overline{x}_i -\sum_{j\in [J_0]}\overline{\mu}_j I(i\in \mathscr{G}_j)|=O_P(\sqrt{\log(N)/T} +c_{N})$.
        \item Given $J_0$, $\Pr(\widehat{\mathcal{G}}_j=\mathscr{G}_j)\to 1$ for all $j\in [J_0]$, where $\widehat{\mathcal{G}}_j$ is obtained via \eqref{est.group}.
        \item When $J_0$ is unknown, $\Pr(\widehat{J}=J_0)\to 1$, where $\widehat{J}$ is obtained via \eqref{est.J}.
    \end{enumerate}
\end{proposition}
The first result of this proposition is evident building on the investigation in Section \ref{Sec2.2} and Assumption \ref{AS.group}.1. After successfully partitioning the individuals, one can further investigate the homogeneity/heterogeneity for the $j^{th}$ group, e.g.,

\begin{eqnarray*}
    \mathbb{H}_0: \mu_{i}\equiv \overline{\mu}_j\text{ for all }i \in \mathscr{G}_j
\end{eqnarray*}
using the results developed in Sections \ref{Sec2.1}-\ref{Sec2.3}.

\subsection{Nonstationary Panel Data}

We now revisit the nonstationary panel data model studied in \cite{PM1999}. Consider the following data generating process:

\begin{eqnarray}
\mathbf{y}_t = \mathbf{y}_{t-1}+\mathbf{B}(L)\pmb{\varepsilon}_t,
\end{eqnarray}
where $\mathbf{y}_t=(y_{1t},\ldots, y_{Nt})^\top$, $\mathbf{y}_0=\mathbf{0}$ without loss of generality, and $\mathbf{B}(L)\pmb{\varepsilon}_t$ is the same as that in \eqref{def.xt_1}. As presented in Section 2 of \cite{PM1999}, when studying nonstationary panel data, a key quantity is

\begin{eqnarray}
    \frac{1}{NT^2}\sum_{i=1}^N\sum_{t=1}^T y_{it}^2 = \frac{1}{NT^2}\sum_{t=1}^T\mathbf{y}_t^\top \mathbf{y}_t,
\end{eqnarray}
which is also the foundation of some basic results of \cite{BAI200982} and \cite{DGP2021}. Using the results of Section \ref{Sec2}, we are now able to account for the cross-sectional dependence of these unit root processes, and present the following proposition.

\begin{proposition}\label{P2}
Suppose that Assumption \ref{AS1} holds with $L_N=N$, and let further that  $\lim_{N}\frac{1}{N}\|\mathbf{B} \|^2\to b$. Then as $(N,T)\to (\infty,\infty)$, $\frac{1}{NT^2}\sum_{i=1}^N\sum_{t=1}^T y_{it}^2\to_P\frac{b}{2}$.
\end{proposition}
Previously, one has to impose cross-sectional independence such as \cite{PM1999}, \cite{DGP2021} and \cite{HUANG2021198} due to technical constraints.

\subsection{CCE Estimators}\label{Sec3.1}

We now revisit the model of \cite{Pesaran2006} and the corresponding CCE estimators using the results established above. Accordingly, we let $L_N\equiv N$ in what follows. The model is as follows:

\begin{eqnarray*}
y_{it}&=& \mathbf{w}_{it}^\top \pmb{\theta}_i + \pmb{\gamma}_i^\top \mathbf{f}_t+ \epsilon_{it}, \notag \\
\mathbf{w}_{it}&=& \pmb{\Gamma}_i^\top \mathbf{f}_t +\mathbf{v}_{it},
\end{eqnarray*}
where only $\{(y_{it}, \mathbf{w}_{it})\mid i\in [N], t\in [T] \}$ are observable, $\mathbf{w}_{it}$ is a $k\times 1$ vector, $\mathbf{f}_t$ is an $m\times 1$ vector, and the dimensions of the other variables are defined accordingly. Both $m$ and $k$ are finite. The model admits a vector form:

\begin{eqnarray*}
\mathbf{Y}_i =\mathbf{W}_i\pmb{\theta}_i+\mathbf{F}\pmb{\gamma}_i +\pmb{\epsilon}_i,
\end{eqnarray*}
where we have $\mathbf{Y}_i =(y_{i1},\ldots, y_{iT})^\top$, $\mathbf{W}_i =(\mathbf{w}_{i1},\ldots, \mathbf{w}_{iT})^\top$, $\mathbf{F} =(\mathbf{f}_1,\ldots, \mathbf{f}_T)^\top$, and $\pmb{\epsilon}_i =(\epsilon_{i1},\ldots, \epsilon_{iT})^\top$. In the homogenous setting (e.g., \citealp{Westerlund2018} and references therein), one further assumes that

\begin{eqnarray}\label{def.theta}
\pmb{\theta}_i\equiv \pmb{\theta}\text{ for all } i\in [N]. 
\end{eqnarray}
It is then natural to question whether \eqref{def.theta} holds practically.  

To exam \eqref{def.theta}, we briefly review the CCE approach. When eliminating the unobservable factor structure, the CCE approach utilizes the following form:

\begin{eqnarray}\label{def.yw}
\begin{pmatrix}
y_{it} \\
\mathbf{w}_{it}
\end{pmatrix} =\left[ \begin{pmatrix}
\pmb{\gamma}_i &  \pmb{\Gamma}_i
\end{pmatrix}\begin{pmatrix}
1 & \mathbf{0} \\
\pmb{\beta}_i & \mathbf{I}_k
\end{pmatrix}\right]^\top \mathbf{f}_t+\begin{pmatrix}
\epsilon_{it} \\
\mathbf{v}_{it}
\end{pmatrix},
\end{eqnarray}
and for notational simplicity, we rewrite \eqref{def.yw} as 

\begin{eqnarray*}
\mathbf{z}_{it}=\mathbf{C}_i^\top \mathbf{f}_t+ \mathbf{u}_{it},
\end{eqnarray*}
where the definitions of $\mathbf{z}_{it}$, $\mathbf{C}_i$, and $\mathbf{u}_{it}$ are self-evident. Simple algebra yields that

\begin{eqnarray*}
\overline{\mathbf{Z}} &=&  \mathbf{F}\overline{\mathbf{C}}+\overline{\mathbf{U}},
\end{eqnarray*}
where $\overline{\mathbf{C}}\coloneqq \frac{1}{N}\sum_{i=1}^N \mathbf{C}_i $, $\overline{\mathbf{Z}}\coloneqq \frac{1}{N}\sum_{i=1}^N \mathbf{Z}_i $ with $\mathbf{Z}_i =(\mathbf{z}_{i1},\ldots, \mathbf{z}_{iT})^\top$,  and $\overline{\mathbf{U}}\coloneqq \frac{1}{N}\sum_{i=1}^N \mathbf{U}_i $ with $\mathbf{U}_i =(\mathbf{u}_{i1},\ldots, \mathbf{u}_{iT})^\top$. Consequently, the CCE estimators of $\pmb{\theta}_i$ and $\pmb{\theta}$ are respectively defined by 

\begin{eqnarray}\label{def.theta_hat}
\widehat{\pmb{\theta}}_i &=&(\mathbf{W}_i^\top \mathbf{M}_{\overline{\mathbf{Z}}} \mathbf{W}_i)^{-1} \mathbf{W}_i^\top \mathbf{M}_{\overline{\mathbf{Z}}} \mathbf{Y}_i \quad \text{for} \quad \forall i\in [N],\notag \\
\widehat{\pmb{\theta}}&=&\left(\sum_{i=1}^N\mathbf{W}_i^\top \mathbf{M}_{\overline{\mathbf{Z}}} \mathbf{W}_i\right)^{-1} \sum_{i=1}^N\mathbf{W}_i^\top \mathbf{M}_{\overline{\mathbf{Z}}} \mathbf{Y}_i,
\end{eqnarray}
where $\mathbf{M}_{\overline{\mathbf{Z}}} =\mathbf{I}_T-\overline{\mathbf{Z}}(\overline{\mathbf{Z}}^\top \overline{\mathbf{Z}})^+ \overline{\mathbf{Z}}^\top$.

The key part of the CCE approach is that $\mathbf{M}_{\overline{\mathbf{Z}}} $ offers a good approximation of $\mathbf{M}_{\mathbf{F}}$, because

\[
\overline{\mathbf{C}}^+ =\overline{\mathbf{C}}^\top(\overline{\mathbf{C}}\, \overline{\mathbf{C}}^\top)^{-1}
\] 
under the conditions: (1) $\overline{\mathbf{C}}$ has full row rank, and (2) $\overline{\mathbf{U}}$ is asymptotically negligible.

Based on \eqref{def.theta_hat}, we construct the following quantities for all $j\in [k]$

\begin{eqnarray*}
Q_{j} = \max_i \frac{1}{\sqrt{T}} \mathbf{e}_j^\top (\mathbf{W}_i^\top \mathbf{M}_{\overline{\mathbf{Z}}} \mathbf{W}_i)(\widehat{\pmb{\theta}}_i-\widehat{\pmb{\theta}}) ,
\end{eqnarray*}
where $\mathbf{e}_j$ is a selection vector as defined in Section \ref{Sec1}. We further let

\begin{eqnarray*}
\widehat{\pmb{\Omega}}_{j} &=& \frac{1}{T}\sum_{t,s=1}^Ta\left(\frac{t-s}{\widetilde{m}}\right)\widehat{\pmb{\xi}}_{jt}\widehat{\pmb{\xi}}_{js}^\top,
\end{eqnarray*}
where $(\widehat{\pmb{\xi}}_{j1},\ldots,\widehat{\pmb{\xi}}_{jT}) = (\widehat{\pmb{\mathbf{U}}}_{1,1}\circ \widehat{\pmb{\mathbf{U}}}_{1,1+j},\ldots, \widehat{\pmb{\mathbf{U}}}_{N,1}\circ \widehat{\pmb{\mathbf{U}}}_{N,1+j})^\top$, $\widehat{\pmb{\mathbf{U}}}_i=  \mathbf{M}_{\overline{\mathbf{Z}}} \mathbf{Z}_i,$  and $\widehat{\pmb{\mathbf{U}}}_{i,j}$ stands for the $j^{th}$ column of $\widehat{\pmb{\mathbf{U}}}_{i}$.

To facilitate the development, we impose the following conditions.

\begin{assumption}\label{AS5}
\item 
\begin{enumerate}[leftmargin=24pt, parsep=2pt, topsep=2pt]
\item (1) Suppose that $\mathbf{f}_t$  is covariance stationary with absolute summable auto covariances, and $\frac{1}{T}\mathbf{F}^\top \mathbf{F}\to_P \pmb{\Sigma}_{\mathbf{f}}>0$. (2) $\overline{\mathbf{C}}\to_P \mathbf{C}$, where $\mathbf{C}$ has full row rank $m$. (3) $\frac{1}{T}\overline{\mathbf{U}}^\top \overline{\mathbf{U}}=O_P(\frac{1}{N})$, $\frac{1}{T} \mathbf{F}^\top \overline{\mathbf{U}}=O_P(\frac{1}{\sqrt{NT}})$ and $\frac{1}{T}\mathbf{V}_i^\top \mathbf{M}_{\mathbf{F}} \pmb{\epsilon}_i=\frac{1}{T}\mathbf{V}_i^\top   \pmb{\epsilon}_i+O_P(\frac{1}{T})$ for $\forall i$. (4) $\frac{1}{T}\mathbf{V}_i^\top \mathbf{V}_i\to_P\pmb{\Sigma}_{\mathbf{V}_i}>0$ for $\forall i$, and $\frac{1}{NT}\sum_{i=1}^N\mathbf{V}_i^\top \mathbf{V}_i\to_P\pmb{\Sigma}_{\mathbf{V}}>0$.  

\item Suppose that  $\pmb{\epsilon}_t =(\epsilon_{1t},\ldots, \epsilon_{Nt})^\top \coloneqq\sum_{\ell=0}^{\infty} \mathbf{B}_\ell  \pmb{\varepsilon}_{t-\ell}$ fulfills Assumption \ref{AS1}, and the conditions of Theorem \ref{THM.6}. Let $\mathbf{v}_t\coloneqq f(\pmb{\epsilon}_{t-1},\ldots, \pmb{\epsilon}_{t-p})$, where  $p$ is fixed, and $f(\pmb{\epsilon}_{t-1},\ldots, \pmb{\epsilon}_{t-p})$ admits a linear combination of $\pmb{\epsilon}_{t-1},\ldots, \pmb{\epsilon}_{t-p}$.

\end{enumerate}
\end{assumption}

The first two conditions of Assumption \ref{AS5}.1 are standard. In the third condition,  $\frac{1}{T}\overline{\mathbf{U}}^\top \overline{\mathbf{U}}=O_P(\frac{1}{N})$ and $\frac{1}{T} \mathbf{F}^\top \overline{\mathbf{U}}=O_P(\frac{1}{\sqrt{NT}})$ are rather standard in view of the development of Section \ref{Sec2.1}, and the following expansions:

\begin{eqnarray*}
\frac{1}{T}\overline{\mathbf{U}}^\top \overline{\mathbf{U}} =\frac{1}{TN^2}\sum_{t=1}^T \mathbf{U}_t^\top\mathbf{1}_N \mathbf{1}_N^\top  \mathbf{U}_t  \quad\text{and}\quad \frac{1}{T} \mathbf{F}^\top \overline{\mathbf{U}} = \frac{1}{TN}\sum_{t=1}^T  \mathbf{f}_t \mathbf{1}_N^\top\mathbf{U}_t,
\end{eqnarray*}
where $\mathbf{U}_t=(\mathbf{u}_{1t},\ldots, \mathbf{u}_{Nt})^\top$. The requirement $\frac{1}{T}\mathbf{V}_i^\top \mathbf{M}_{\mathbf{F}} \pmb{\epsilon}_i=\frac{1}{T}\mathbf{V}_i^\top  \pmb{\epsilon}_i+O_P(\frac{1}{T})$ can be verified by the development of Section \ref{Sec2.2}. Assumption \ref{AS5}.2 allows $\pmb{\epsilon}_t$ and $\mathbf{v}_t$ to be weakly dependent.

Under these condition, the following proposition holds.

\begin{proposition}\label{P1}
Suppose that Assumption \ref{AS5} holds and $N\asymp T$. Then for all $j\in [k]$

\begin{eqnarray*}
\sup_{u\in \mathbb{R}}\left|\Pr(Q_j\le u)- \mathrm{Pr}^*\left(\left|\widehat{\bm{\Omega}}_j^{1/2}\frac{1}{\sqrt{T}}\sum_{t=1}^{T}\mathbf{z}_t^* \right|_{\infty} \leq u\right)\right|=o_P(1),
\end{eqnarray*}
where $\{\mathbf{z}_t^* \mid t\in [T]\}$ is a sequence of i.i.d. $N$-dimensional Gaussian random vectors with $\mathbf{z}_t^*\sim N(\mathbf{0},\mathbf{I}_N)$.
\end{proposition}

According to Proposition \ref{P1}, all $Q_j$'s should not fall in the rejection region yielded by the bootstrap draws under  \eqref{def.theta}. We further examine this result in the simulation studies.

\medskip

On top of these extensions, we may also revisit the time trend analyses of \cite{gh2006}, \cite{cgl2012}, \cite{ROBINSON20124} and \cite{WSX2023}, which are of course mathematically involved due to the nonparametric nature. We leave them for future study.

\section{Simulation}\label{Sec4}

In this section, we conduct simulation studies to examine the theoretical results about testing and inference of Section \ref{Sec2}, and also verify our argument about the CCE estimators in Section \ref{Sec3.1}.

\textbf{Simulation 1} (Testing) --- The data generating process (DGP) is simplified as follows:

\begin{eqnarray*}
    \mathbf{x}_t =\pmb{\mu} + \rho_x \mathbf{x}_{t-1}+ \pmb{\Sigma}_\nu^{1/2}\pmb{\nu}_t,
\end{eqnarray*}
where  $t=-200, \ldots,0, 1,\ldots, T$, $\rho_x=0.3$, $\pmb{\mu}\coloneqq (\mu_1,\ldots, \mu_N)^\top$, $\mathbf{x}_t\coloneqq (x_{1t},\ldots, x_{Nt})^\top$,  and $\pmb{\nu}_t\coloneqq (\nu_{1t},\ldots, \nu_{Nt})^\top$. As well understood, the AR(1) process admits an MA($\infty$) representation thus suiting the definition of \eqref{def.xt_1}. To introduce cross-sectional dependence, we let $\pmb{\Sigma}_\nu=\{\rho_\nu^{|i-j|}\}_{N\times N}$. The observations from $t=-200,\ldots, 0$ are burn-in sample in order to eliminate the impact of the initial value. We let $\{\nu_{it} \}$ be i.i.d. over both $i$ and $t$, and be generated in three cases:

\begin{enumerate}[leftmargin=48pt, parsep=2pt, topsep=2pt]
\item[Case 1.] $\nu_{it}\sim N(0,1)$;
\item[Case 2.] $\nu_{it}\sim t_8$, where $t_8$ stands for a $t$-distribution with a degree freedom 8;
\item[Case 3.] $\nu_{it}\sim \Gamma(2,2)-1$, where $\Gamma(2,2)$ stands for a Gamma distribution with a shape parameter 2 and a scale parameter 0.5. Therefore, $\Gamma(2,2)-1$ has mean 0.
\end{enumerate}
Case 1 is a symmetric distribution with thin tails; Case 2 is a symmetric distribution with heavy tails; and Case 3 is an asymmetric distribution.  To exam the size and power of the proposed test in Section \ref{Sec2.3}, for each case we consider three scenarios for $\pmb{\mu}$:

\begin{enumerate}[parsep=2pt, topsep=2pt]
    \item[(a).] $\mu_i\equiv 0$ for all $i$;
    \item[(b).] $\mu_1=\frac{4}{\sqrt{T}}$, and $\mu_i\equiv 0$ for $i\ge 2$;
    \item[(c).] $\mu_1=1$, and $\mu_i\equiv 0$ for $i\ge 2$.
\end{enumerate}
Scenarios (a)-(c) are designed to examine  size, local power, and  global power respectively. 

Notably, the above DGP is a special case of \cite{PY2008} and \cite{YU2024105458}. Thus, for the purpose of comparison, we also consider the approaches of these two papers (referred to as PY and YYX respectively). To put everything on equal footing, we consider the null of \cite{YU2024105458} (i.e., $\mathbb{H}_0:\mu_i\equiv 0$ for all $i$), and modify the test statistic of \cite{PY2008} accordingly in an obvious manner. For the sake of space, we refer interested readers to their papers for detailed implementation. As PY and YYX methods calculate the asymptotic variances neglecting the dependence of the residuals (e.g., \citealp[Eq. (14)]{PY2008} and \citealp[Eq. (2.7)]{YU2024105458}), we anticipate some distorted size or power. Additionally, \cite{YU2024105458} rely on the Gaussian assumption, so we anticipate the DGPs of Cases 2 and 3 will further distort size or power. For our method (referred to as GLPY), we calculate $Q_{NT}$ under the null for each dataset, and obtain the 95\% confidence interval (denoted by $\text{CI}_\infty$) via $  |\widehat{\bm{\Omega}}^{1/2}\frac{1}{\sqrt{T}}\sum_{t=1}^{T}\mathbf{z}_t^* |_{\infty}$ based on 399 bootstrap replications. After $R$ replications, we calculate the rejection rate as follows:

\begin{eqnarray*}
    \Delta_{test} = \frac{1}{R}\sum_{j=1}^R I(Q_{NT,j}\not\in \text{CI}_{\infty,j}),
\end{eqnarray*}
where the subindex $j$ stands for the corresponding values obtained in the $j^{th}$ simulation replication. Similarly, we will report the rejection rates for PY and YYX approaches. For our method, we expect that $\Delta_{test}$ is sufficiently close to 0.05 for the scenario (a), is reasonably close to 1 for the scenario (c), and is in between 0 and 1 for the scenario (b). For simplicity, we let  $a(\cdot)$ be Bartlett kernel, and take suggestions from \cite{gao2024robust} to set $\widetilde{m}=\lfloor 1.75 T^{1/3}\rfloor$ for simplicity. Additionally, we let $R=1000$, $N \in \{100,150,200 \}$, $T\in \{ 200, 400 \}$, and $\rho_\nu \in \{0.5, 0.95 \}$. 

The results are summarized in Table \ref{TB_sim1}. Overall, our approach has reasonable size, local power, and global power irrespective to the magnitude of CD (i.e., the value of $\rho_\nu$) as expected. Due to omitting dependence, PY and YYX methods tend to over reject, which is evident in view of the rejection rates of scenario (a) of Cases 1-3.

\begin{table}[tbh!]\small
\setlength{\tabcolsep}{3pt} 
\renewcommand{\arraystretch}{1} 
\caption{Values of $\Delta_{test}$ of Simulation 1}\label{TB_sim1}
\centering\begin{tabular}{rrrrrrlrrrlrrr}
\hline\hline
 &  &  & \multicolumn{1}{c}{GLPY} & \multicolumn{1}{c}{PY} & \multicolumn{1}{c}{YYX} &  & \multicolumn{1}{c}{GLPY} & \multicolumn{1}{c}{PY} & \multicolumn{1}{c}{YYX} &  & \multicolumn{1}{c}{GLPY} & \multicolumn{1}{c}{PY} & \multicolumn{1}{c}{YYX} \\
 &  & $T\setminus N$ & \multicolumn{3}{c}{100} &  & \multicolumn{3}{c}{150} &  & \multicolumn{3}{c}{200} \\
 $\rho_\nu =0.5$ & Case 1 (a) & 200 & 0.042 & 0.979 & 0.939 &  & 0.036 & 0.996 & 0.971 &  & 0.034 & 0.999 & 0.976 \\
 &  & 400 & 0.040 & 0.978 & 0.941 &  & 0.037 & 0.995 & 0.964 &  & 0.058 & 0.999 & 0.974 \\
 &  (b) & 200 & 0.727 & 0.999 & 0.999 &  & 0.717 & 1.000 & 0.999 &  & 0.677 & 1.000 & 1.000 \\
 &  & 400 & 0.762 & 0.998 & 0.996 &  & 0.722 & 1.000 & 0.996 &  & 0.717 & 1.000 & 1.000 \\
 &  (c) & 200 & 1.000 & 1.000 & 1.000 &  & 1.000 & 1.000 & 1.000 &  & 1.000 & 1.000 & 1.000 \\
 &  & 400 & 1.000 & 1.000 & 1.000 &  & 1.000 & 1.000 & 1.000 &  & 1.000 & 1.000 & 1.000 \\
 & Case 2 (a) & 200 & 0.045 & 0.986 & 0.926 &  & 0.050 & 0.998 & 0.966 &  & 0.042 & 1.000 & 0.983 \\
 &  & 400 & 0.049 & 0.982 & 0.938 &  & 0.050 & 0.999 & 0.963 &  & 0.056 & 0.999 & 0.977 \\
 &  (b) & 200 & 0.749 & 0.998 & 1.000 &  & 0.702 & 1.000 & 1.000 &  & 0.655 & 1.000 & 0.999 \\
 &  & 400 & 0.753 & 0.998 & 0.996 &  & 0.705 & 1.000 & 1.000 &  & 0.706 & 1.000 & 0.999 \\
 &  (c) & 200 & 1.000 & 1.000 & 1.000 &  & 1.000 & 1.000 & 1.000 &  & 1.000 & 1.000 & 1.000 \\
 &  & 400 & 1.000 & 1.000 & 1.000 &  & 1.000 & 1.000 & 1.000 &  & 1.000 & 1.000 & 1.000 \\
 & Case 3 (a) & 200 & 0.048 & 0.976 & 0.931 &  & 0.050 & 0.998 & 0.975 &  & 0.042 & 1.000 & 0.984 \\
 &  & 400 & 0.042 & 0.973 & 0.931 &  & 0.047 & 0.997 & 0.969 &  & 0.049 & 1.000 & 0.985 \\
 &  (b) & 200 & 0.990 & 1.000 & 1.000 &  & 0.994 & 1.000 & 1.000 &  & 0.985 & 1.000 & 1.000 \\
 &  & 400 & 0.991 & 1.000 & 1.000 &  & 0.987 & 1.000 & 1.000 &  & 0.989 & 1.000 & 1.000 \\
 &  (c) & 200 & 1.000 & 1.000 & 1.000 &  & 1.000 & 1.000 & 1.000 &  & 1.000 & 1.000 & 1.000 \\
 &  & 400 & 1.000 & 1.000 & 1.000 &  & 1.000 & 1.000 & 1.000 &  & 1.000 & 1.000 & 1.000 \\ \cline{4-14}
 &  &  &  &  &  &  &  &  &  &  &  &  &  \\
$\rho_\nu =0.95$ & Case 1 (a) & 200 & 0.042 & 0.731 & 0.481 &  & 0.036 & 0.781 & 0.537 &  & 0.034 & 0.838 & 0.602 \\
 &  & 400 & 0.040 & 0.705 & 0.491 &  & 0.037 & 0.787 & 0.538 &  & 0.058 & 0.835 & 0.582 \\
 & (b) & 200 & 0.850 & 0.818 & 0.995 &  & 0.850 & 0.866 & 0.988 &  & 0.812 & 0.883 & 0.989 \\
 &  & 400 & 0.872 & 0.827 & 0.993 &  & 0.828 & 0.862 & 0.987 &  & 0.822 & 0.887 & 0.988 \\
 & (c) & 200 & 1.000 & 1.000 & 1.000 &  & 1.000 & 1.000 & 1.000 &  & 1.000 & 1.000 & 1.000 \\
 &  & 400 & 1.000 & 1.000 & 1.000 &  & 1.000 & 1.000 & 1.000 &  & 1.000 & 1.000 & 1.000 \\
 & Case 2 (a) & 200 & 0.045 & 0.696 & 0.486 &  & 0.050 & 0.825 & 0.584 &  & 0.042 & 0.841 & 0.610 \\
 &  & 400 & 0.049 & 0.742 & 0.468 &  & 0.050 & 0.794 & 0.546 &  & 0.056 & 0.836 & 0.593 \\
 &  (b) & 200 & 0.698 & 0.782 & 0.961 &  & 0.663 & 0.847 & 0.953 &  & 0.617 & 0.898 & 0.957 \\
 &  & 400 & 0.711 & 0.789 & 0.957 &  & 0.655 & 0.838 & 0.955 &  & 0.630 & 0.847 & 0.953 \\
 &  (c) & 200 & 1.000 & 1.000 & 1.000 &  & 1.000 & 1.000 & 1.000 &  & 1.000 & 1.000 & 1.000 \\
 &  & 400 & 1.000 & 1.000 & 1.000 &  & 1.000 & 1.000 & 1.000 &  & 1.000 & 1.000 & 1.000 \\
 & Case 3 (a) & 200 & 0.048 & 0.729 & 0.483 &  & 0.050 & 0.797 & 0.527 &  & 0.042 & 0.833 & 0.598 \\
 &  & 400 & 0.042 & 0.734 & 0.493 &  & 0.047 & 0.776 & 0.553 &  & 0.049 & 0.841 & 0.580 \\
 & (b) & 200 & 0.998 & 0.952 & 1.000 &  & 0.995 & 0.945 & 1.000 &  & 0.996 & 0.950 & 1.000 \\
 &  & 400 & 0.999 & 0.914 & 1.000 &  & 0.998 & 0.932 & 1.000 &  & 0.992 & 0.921 & 1.000 \\
 & (c) & 200 & 1.000 & 1.000 & 1.000 &  & 1.000 & 1.000 & 1.000 &  & 1.000 & 1.000 & 1.000 \\
 &  & 400 & 1.000 & 1.000 & 1.000 &  & 1.000 & 1.000 & 1.000 &  & 1.000 & 1.000 & 1.000 \\ 
 \hline\hline
\end{tabular}
\end{table}

\textbf{Simulation 2} (Inference) --- In this simulation, we examine the bootstrap inferences documented in Sections \ref{Sec2.1} and \ref{Sec2.2}. For simplicity, we consider the three cases identical to Simulation 1 with $\pmb{\mu}=\mathbf{0}_{N\times 1}$, and  we set $m=\lfloor 1.75 T^{1/3}\rfloor$. For each dataset, we infer a few quantities. For the homogenous case, we calculate $\widetilde{S}_{NT}$ of Section \ref{Sec2.1}, and simulate its distribution via $\widetilde{S}_{NT}^*$ based on 399 bootstrap replications, where $\widetilde{S}_{NT}^*$ is self-normalized by construction and does not require any prior knowledge about $L_N$. Using the bootstrap draws, we construct the 95\% confidence interval of  $\widetilde{S}_{NT}$, denoted by $\text{CI}_{S}$. Second, for each $i$, we calculate $\widetilde{p}_i$ of Section \ref{Sec2.2}, and construct its distribution via  $\widetilde{p}_i^*$ based on 399 bootstrap replications. Accordingly, we construct the 95\% confidence interval of each $\widetilde{p}_i$, denoted by $\text{CI}_{p_i}$.  

After $R$ simulation replications, we calculate the following measures:

\begin{eqnarray*} 
\Delta_{\text{HM}} &=& \frac{1}{R} \sum_{j=1}^R I(\widetilde{S}_{NT,j}\not\in \text{CI}_{S,j}) ,\notag \\ 
\Delta_{\text{HE}} &=& \frac{1}{N}\sum_{i=1}^N |\Delta_i|\quad \text{with}\quad \Delta_i=\frac{1}{R} \sum_{j=1}^R  I(\widetilde{p}_{i,j}\not\in \text{CI}_{p_i,j}),\notag \\
\text{Sd}_{\text{HE}} &=& \left\{\frac{1}{N}\sum_{i=1}^N(|\Delta_i|-\Delta_{\text{HE}})^2\right\}^{1/2},
\end{eqnarray*}
where again $j$ indexes the $j^{th}$ simulation replication. We anticipate that $\Delta_{\text{HM}}$ and $\Delta_{\text{HE}} $ are close to 0.05, and $\text{Sd}_{\text{HE}}$ is close to 0 indicating the proposed method in Section \ref{Sec2.2} is stable for all $i$'s. 

\begin{table}[tbh!]\small
    \centering\caption{Results of Simulation 2}\label{TB_sim2}
\renewcommand{\arraystretch}{0.8}
\begin{tabular}{llrrrrlrrr}
    \hline\hline
    &  &  & \multicolumn{3}{c}{$\rho_\nu =0.5$} &  & \multicolumn{3}{c}{$\rho_\nu =0.95$} \\ \cline{4-10} 
    &  & $T\setminus N$ & 100 & 150 & 200 &  & 100 & 150 & 200 \\
Case 1 & $\Delta_{\text{HM}}$ & 200 & 0.059 & 0.049 & 0.048 &  & 0.039 & 0.056 & 0.050 \\
    &  & 400 & 0.049 & 0.056 & 0.057 &  & 0.055 & 0.042 & 0.050 \\
    &  &  &  &  &  &  &  &  &  \\
    & $\Delta_{\text{HE}}$ & 200 & 0.055 & 0.056 & 0.055 &  & 0.055 & 0.056 & 0.056 \\
    &  & 400 & 0.056 & 0.055 & 0.056 &  & 0.055 & 0.055 & 0.056 \\
    & $\text{Sd}_{\text{HE}}$ & 200 & 0.006 & 0.007 & 0.007 &  & 0.006 & 0.007 & 0.007 \\
    &  & 400 & 0.007 & 0.007 & 0.007 &  & 0.006 & 0.006 & 0.007 \\  \cline{4-10} 
    &  &  &  &  &  &  &  &  &  \\
Case 2 & $\Delta_{\text{HM}}$ & 200 & 0.045 & 0.050 & 0.042 &  & 0.046 & 0.048 & 0.059 \\
    &  & 400 & 0.049 & 0.050 & 0.056 &  & 0.046 & 0.034 & 0.047 \\
    &  &  &  &  &  &  &  &  &  \\
    & $\Delta_{\text{HE}}$ & 200 & 0.055 & 0.056 & 0.055 &  & 0.055 & 0.056 & 0.055 \\
    &  & 400 & 0.056 & 0.056 & 0.056 &  & 0.055 & 0.055 & 0.055 \\
    & $\text{Sd}_{\text{HE}}$ & 200 & 0.007 & 0.007 & 0.006 &  & 0.006 & 0.007 & 0.007 \\
    &  & 400 & 0.007 & 0.007 & 0.007 &  & 0.006 & 0.006 & 0.007 \\  \cline{4-10} 
    &  &  &  &  &  &  &  &  &  \\
Case 3 & $\Delta_{\text{HM}}$ & 200 & 0.057 & 0.061 & 0.055 &  & 0.045 & 0.048 & 0.051 \\
    &  & 400 & 0.053 & 0.052 & 0.048 &  & 0.044 & 0.049 & 0.050 \\
    &  &  &  &  &  &  &  &  &  \\
    & $\Delta_{\text{HE}}$ & 200 & 0.055 & 0.056 & 0.055 &  & 0.056 & 0.056 & 0.056 \\
    &  & 400 & 0.056 & 0.056 & 0.056 &  & 0.056 & 0.055 & 0.056 \\
    & $\text{Sd}_{\text{HE}}$ & 200 & 0.006 & 0.007 & 0.006 &  & 0.007 & 0.007 & 0.007 \\
    &  & 400 & 0.007 & 0.007 & 0.007 &  & 0.007 & 0.006 & 0.007 \\
    \hline\hline
\end{tabular}
\end{table}

Table \ref{TB_sim2} shows that most values are as expected. While the DGPs cover a symmetric distribution with thin tails, a symmetric distribution with heavy tails, and an asymmetric distribution, the above results are reasonably good irrespective of the magnitude of CD.

\medskip

\textbf{Simulation 3} (CCE) --- We then consider the following panel data model: 

\begin{eqnarray*}
y_{it}&=& w_{it}  \theta_i + \gamma_i  f_t+ \epsilon_{it}, \notag \\
w_{it}&=&\Gamma_i  f_t +v_{it},
\end{eqnarray*}
where all elements are scalar for simplicity. In this case, $m=1$ and $k=1$, so $m\le k+1$ (a typical requirement for CCE estimators) is fulfilled. We examine the size and power of Proposition \ref{P1}.

The DGP is as follows. $\{\epsilon_{it}\}$ follows the identical DGP of $\{x_{it}\}$ as in Cases 1-3 of Simulation 1 with $\pmb{\mu}=\mathbf{0}_{N\times 1}$ and $\rho_{\nu}=0.5$. We let $f_t\sim N(0,1.5)$, $\gamma_i \sim N(0.8, 1)$, $\Gamma_i\sim N(-0.2, 2)$, and $\{v_{it}\coloneqq \epsilon_{i,t-1} \}$. For each case of  $\{\epsilon_{it}\}$, we further consider three scenarios for $\{\theta_i\}$:

\begin{itemize}[parsep=2pt, topsep=2pt]
\item[(a).] $\theta_i\equiv 1$;
\item[(b).] $\theta_1=1+\frac{4}{\sqrt{T}}$, and $\theta_i\equiv 1$ for $i\ge 2$;
\item[(c).] $\theta_1=2$, and $\theta_i\equiv 1$ for $i\ge 2$.
\end{itemize}
Similar to Simulation 1, scenarios (a)-(c) are designed to evaluate size, local power, and global power. For each dataset, we calculate $Q_1$ and generate the corresponding confidence interval (say, $\text{CI}_{Q_1}$) as in Proposition \ref{P1}. After $R$ simulation replications, we report the following measure:

\begin{eqnarray*}
\Delta_{\text{Q}} &=& \frac{1}{R} \sum_{j=1}^R I(Q_{1,j}\not\in \text{CI}_{Q_1,j}) ,
\end{eqnarray*}
where  $j$ still indexes the $j^{th}$ simulation replication.  $\Delta_{\text{Q}}$ should be close to 0.05 and 1 for the scenarios (a) and (c) respectively, and should be in between 0 and 1 for the scenario (b). As shown in Table \ref{TB2}, the results are as expected. For scenario (a) of Case 3, our approach is slightly under-size with large $N$, and it might be due to the fact that the error component of Case 3 is skewed.

\begin{table}[tbh!]\small
\centering \caption{Values of $\Delta_{\text{Q}}$ of Simulation 3}\label{TB2}
\renewcommand{\arraystretch}{0.8}
\begin{tabular}{llrrrrrrrr}
\hline\hline
 &  & \multicolumn{2}{c}{(a)} & \multicolumn{1}{l}{} & \multicolumn{2}{c}{(b)} & \multicolumn{1}{l}{} & \multicolumn{2}{c}{(c)} \\ \cline{3-4} \cline{6-7} \cline{9-10}
 & $T\setminus N$ & 100 & 150 &  & 100 & 150 &  & 100 & 150 \\
Case 1 & 150 & 0.058 & 0.032 &  & 0.513 & 0.540 &  & 1.000 & 1.000 \\
 & 200 & 0.060 & 0.045 &  & 0.648 & 0.603 &  & 1.000 & 1.000 \\ \cline{3-10}
 & & & & & & & & & \\
Case 2 & 150 & 0.040 & 0.027 &  & 0.528 & 0.473 &  & 1.000 & 1.000 \\
 & 200 & 0.047 & 0.040 &  & 0.590 & 0.500 &  & 1.000 & 1.000 \\ \cline{3-10}
  & & & & & & & & & \\ 
Case 3 & 150 & 0.047 & 0.025 &  & 0.423 & 0.408 &  & 1.000 & 1.000 \\
 & 200 & 0.068 & 0.027 &  & 0.518 & 0.385 &  & 1.000 & 1.000 \\
\hline\hline
\end{tabular}
\end{table}

\section{A Case Study}\label{Sec5}

Heterogeneous expectations, which may arise in financial markets because investors interpret and react to available information differently, are crucial for  understanding variations in asset prices, portfolio allocations, and market dynamics. As a result, the study of heterogeneous expectations has become an important focus in recent years, with a number of works highlighting its implications for financial theory and practice \citep[see, for example,][]{ame2020, Brun2021, Giglio2021}.
Notably, \cite{Giglio2021} document significant and persistent cross-sectional variations in individual investors' return expectations, which can only be partially explained by demographic factors. Similarly, \cite{DI2024} reveal substantial heterogeneity in equity return expectations among institutional investors and investment consultants and it can be linked to asset managers' portfolios. These findings challenge traditional finance models based on homogeneous rational expectations and motivate the development of alternative frameworks.

In this study, we revisit some prior work in this area and formally test the heterogeneity in equity return expectations using the newly proposed inference method. Furthermore, we investigate the relationships between equity premium expectations and equity valuations, as explored by \cite{DI2024}.

\subsection{Variables and Data}

Following \cite{DI2024}, we examine three types of subjective return expectations held by asset managers for U.S. equities. The equity return expectation ($ere_{it}$) is defined as the (geometric) nominal equity return forecast for large-cap U.S. equities over a 10-year horizon, as reported in public disclosures\footnote{As noted by \cite{DI2024}, the forecast horizons for return expectation data range from 1 to 50 years. However, most asset managers provide forecasts close to a 10-year horizon. Consequently, this study focuses on 10-year horizon forecasts.}. The equity premium expectation over yield ($epey_{it}$) is calculated by subtracting the horizon-matched log nominal Treasury yield from the nominal equity return expectation. Similarly, the equity premium expectation over cash ($epec_{it}$) is derived by subtracting the expected annualized return on cash (over the next 10 years) from the equity return expectation.

For the equity valuation, we utilize the cyclically adjusted price-to-earnings ratio ($cape_t$), defined as the log ratio of real equity prices to real earnings averaged over the past decade. Introduced by \cite{campbell1988}, this measure of equity valuation is widely recognized in the finance literature \citep[see, for example,][]{JL2019}. Additionally, we include the past 12-month return of the S\&P 500 index ($pr_t$) to capture momentum effects and the horizon-matched Treasury yield ($rf_t$) to account for the risk-free rate.

The data used to construct these variables are available on the website maintained by the authors of \cite{DI2024}. By focusing on a fixed 10-year forecast horizon, we obtain a  panel dataset comprising observations from 45 asset managers across 109 months, covering the period from November 1997 to April 2021.  Descriptive statistics for these variables are presented in Table \ref{Table_e1}.

\begin{table}[tbh!]\small
    \centering
    \caption{Descriptive Statistics of Variables}\label{Table_e1}
    \begin{tabular}{lcrr}
        \hline\hline
    Variables                                     & Abbreviation & Mean  & StD   \\
    \hline
    Equity premium expectation over yield         & $epey$  & 3.15  & 1.82  \\
    Equity premium expectation over cash          & $epec$  & 3.14  & 1.84  \\
    Nominal equity return expectation             & $ere$   & 5.15  & 1.73  \\
    Cyclically adjusted price-to-earnings   ratio & $cape$     & 3.37  & 0.11  \\
    Past 12-month return                          & $pr$       & 11.94 & 12.76 \\
    Riskfree rate                                 & $rf$       & 2.00  & 0.78 \\
    \hline
    \end{tabular}
\end{table}

   \subsection{Model Specifications and Inference}\label{S.Emp.2}
To  analyze the heterogeneity in equity return expectations, we estimate three econometric models, each progressively incorporating more explanatory variables to account for the potential determinants of the expectations:

\begin{itemize}[leftmargin=60pt, parsep=2pt, topsep=2pt]
       \item[Model 1:] $y_{it} = \mu_i+\epsilon_{it}$;
       \item[Model 2:] $y_{it} = \alpha_i+\beta_icape_t+\epsilon_{it}$;
       \item[Model 3:] $y_{it} = \alpha_i+\beta_icape_t+\gamma_ipr_t+\delta_irf_t+\epsilon_{it}$,
\end{itemize}
where $y_{it}\in\{epey_{it}, epec_{it}, ere_{it}\}$. Model 1 is the baseline model that captures the individual-specific mean levels of equity return expectations, where $\mu_i$ represents the time-invariant mean for each asset manager $i$. In Model 2,  the equity return expectations are modeled as a function of the $cape_t$, with $\beta_i$ representing the sensitivity of each manager's expectations to equity valuations.  \cite{DI2024} reveal countercyclical expectations for U.S. equity returns, by showing that the expectations are negatively associated with  $cape_t$. In this study, we further investigate  whether the relationship between return expectations and valuations differs across managers.  The inclusion of $pr_t$ and$rf_t$  in Model 3 further enables us to examine whether these additional factors, such as historical market performance, can help explain cross-sectional differences in expectations.

For each specification, we conduct hypothesis tests to evaluate the presence of heterogeneity in key parameters. Specifically:

\begin{enumerate}[leftmargin=24pt, parsep=2pt, topsep=2pt]
    \item We test $\mathbb{H}_0^1$: $\mu_i = \mu$ for all $i$ in Model 1;
    \item We test $\mathbb{H}_0^2$: $\alpha_i = \alpha$ and $\mathbb{H}_0^3$: $\beta_i = \beta$ for all $i$ respectively in Models 2 and 3. 
\end{enumerate}

These tests provide evidence regarding both the variation in average return expectations and the diversity in how managers incorporate fundamental and market-driven information into their forecasts. Rejecting these null hypotheses would indicate the heterogeneity of return expectations among the institutional investors and investment consultants.

\subsection{Estimation and Testing Results}

We first estimate three specifications of equity return expectations using both heterogeneous and homogeneous regression models. The estimated coefficients for each model, along with their 95\% confidence intervals, are presented in Table \ref{Table_e2}. Across all models, the cyclically adjusted price-to-earnings ratio is found to be significantly negatively associated with asset managers' forecasts of equity returns and premiums, indicating countercyclical patterns in equity return expectations: when equity valuations are higher, asset managers expect lower future returns. This result aligns with the findings of \cite{DI2024} and contrasts with the procyclical expectations among retail investors revealed by  \cite{GS2014}. 

Comparing the outcomes from heterogeneous and homogeneous estimations, we can identify noticeable differences. In particular, most homogeneous coefficient estimates for the intercept and $cape$ are larger (in absolute values) than the corresponding average values from heterogeneous models. This discrepancy highlights the importance of allowing for heterogeneity in the data. Homogeneous models, by assuming no diversity across asset managers, may overestimate the influence of equity valuation on return expectations. 

Furthermore, both estimation approaches consistently suggest that past returns exert no significant influence on asset managers' expectations regarding future equity premiums. This finding suggests that asset managers may not rely on recent return trends when forming long-term return expectations, potentially due to the forward-looking nature of their strategies. It also indicates that their expectations are primarily driven by fundamental factors, such as valuations and macroeconomic conditions, rather than past market performance. Importantly, these results hold consistently across all three measures of equity return expectations.

To further investigate the heterogeneity in asset managers' expectations for U.S. equity returns, we perform a sequence of heterogeneity tests as outlined in Section \ref{S.Emp.2}. The results of these tests are summarized in Panel A of Table \ref{Table_e3}. For Model 1, the null hypothesis of homogeneous means is  rejected at the 0.01 significance level for all three types of expectations: equity returns, equity premium over yield, and equity premium over cash. This indicates substantial variation in the average levels of return expectations across asset managers.  

For Models 2 and 3, the heterogeneity in expectations remains significant even after accounting for equity valuations, past returns, and risk-free rates.  The results for $\mathbb{H}_0^3$ further reveal that the relationship between equity return expectations and equity valuations varies significantly across institutional investors and investment consultants. This finding indicates that differences among asset managers extend beyond simple average return expectations; they also exhibit diverse sensitivities to fundamental factors such as equity evaluations. Collectively, these findings confirm the presence of  heterogeneity in U.S. equity premium expectations, which aligns with recent literature emphasizing the role of diverse beliefs and preferences in shaping market outcomes \citep[see,][among others]{Brun2021}. 

In order to illustrate the heterogeneity test results for the intercept in each model, we provide the distributions of bootstrap test statistics in Figure \ref{Fig_Emp}.  For comparison, the test statistics and critical values for the homogeneity test of slope coefficients proposed by \cite{PY2008} (PY) and the homogeneity test of intercepts developed by \cite{YU2024105458} (YYX) are also computed for $\mathbb{H}_0^3$ and $\mathbb{H}_0^2$, respectively, under Model 2. The results are reported in Panel A of Table \ref{Table_e3}.
Notably, the \cite{YU2024105458}'s test originally examines heterogeneity in intercepts under the null hypothesis of known homogeneous intercepts. To adapt this framework to our setting, we modify their statistics to test heterogeneity against an unknown constant by specifying the intercept in the null as its homogeneous estimator. This adjustment ensures comparability with our proposed methodology while maintaining the robustness of their test.


\begin{table}[tbh!]
\footnotesize
\caption{\textbf{Estimation Results from Heterogeneous and Homogeneous Regressions.} In this table, Panel A presents the average values of estimated coefficients (Coef) and their bootstrap 95\% confidence intervals (CI) using  heterogeneous regression models. 
Panel B provides the estimates  for homogeneous regression.  }\label{Table_e2}
\begin{tabular}{llrrrrrrrrr}
\hline\hline
 &  &  & \multicolumn{2}{c}{\begin{tabular}[c]{@{}c@{}}Equity premium \\ (over yield)\end{tabular}} &  & \multicolumn{2}{c}{\begin{tabular}[c]{@{}c@{}}Equity premium \\ (over cash)\end{tabular}} &  & \multicolumn{2}{c}{\begin{tabular}[c]{@{}c@{}}Equity return \\ (nominal)\end{tabular}} \\
 &  &  & Coef & CI &  & Coef & CI &  & Coef & CI \\ \hline
\textbf{Panel A} &   &  &  &  &  &  &  &  &  &  \\
 Model 1 & Intercept &  & 3.80 & (3.46, 4.13) &  & 3.95 & (3.56, 4.28) &  & 5.83 & (5.65, 6.02) \\
 &  &  &  &  &  &  &  &  &  &  \\
 Model 2& Intercept &  & 11.38 & (5.33, 17.27) &  & 10.80 & (2.55, 18.32) &  & 15.13 & (12.25, 18.17) \\
 & $cape$ &  & -2.32 & (-4.09, -0.47) &  & -2.24 & (-4.41, 0.11) &  & -2.88 & (-3.77, -2.03) \\
 &  &  &  &  &  &  &  &  &  &  \\
 Model 3& Intercept &  & 6.40 & (1.57, 10.65) &  & 5.45 & (-2.60, 12.95) &  & 6.40 & (1.46, 11.32) \\
 & $cape$ &  & -0.47 & (-1.78, 1.05) &  & -0.39 & (-2.65, 2.01) &  & -0.47 & (-1.99, 1.14) \\
 & $pr$ &  & -0.02 & (-0.03, 0.00) &  & -0.02 & (-0.04, 0.01) &  & -0.02 & (-0.03, 0.00) \\
 & $rf$ &  & -0.77 & (-0.90, -0.64) &  & -0.42 & (-0.63, -0.21) &  & 0.23 & (0.10, 0.36) \\ \hline
\textbf{Panel B} & &  &  &  &  &  &  &  &  &  \\
 Model 1 & Intercept &  & 3.15 & (2.78, 3.53) &  & 3.14 & (2.82, 3.40) &  & 5.15 & (4.94, 5.34) \\
 &  &  &  &  &  &  &  &  &  &  \\
 Model 2& Intercept &  & 20.67 & (14.85, 25.74) &  & 20.85 & (15.88, 25.30) &  & 24.48 & (22.28, 26.77) \\
 & $cape$ &  & -5.25 & (-6.81, -3.45) &  & -5.42 & (-6.75, -3.92) &  & -5.79 & (-6.50, -5.12) \\
 &  &  &  &  &  &  &  &  &  &  \\
 Model 3& Intercept &  & 25.66 & (24.35, 26.93) &  & 21.95 & (19.15, 25.06) &  & 25.66 & (24.29, 26.95) \\
 & $cape$ &  & -6.33 & (-6.78, -5.91) &  & -5.50 & (-6.52, -4.58) &  & -6.33 & (-6.74, -5.89) \\
 & $pr$ &  & 0.01 & (-0.00, 0.02) &  & 0.01 & (-0.01, 0.03) &  & 0.01 & (-0.00, 0.02) \\
 & $rf$ &  & -0.88 & (-0.96, -0.78) &  & -0.49 & (-0.65, -0.33) &  & 0.12 & (0.03, 0.22) \\
\hline\hline                                       
\end{tabular}
\end{table}

As a robustness check, we follow \cite{DI2024} and expand the dataset to include observations with forecast horizons close to ten years, rather than limiting the sample to exactly 10-year-horizon. The heterogeneity test results for this expanded sample are reported in Panel B of Table \ref{Table_e3}. Most tests continue to reject the null hypothesis of homogeneity at the 0.01 significance level. The only exception is the one for the  mean value of equity premiums over cash, which indicates a slightly weaker rejection (at a 0.05 significance level) of the homogeneity. Nevertheless, the overall robustness of the test results confirms that our findings are not sensitive to the inclusion of additional observations with wider forecast horizons.

\begin{table}[tbh!]
\footnotesize 
\begin{center} 
\caption{\textbf{Results of the Heterogeneity Tests.}  For each  test, we report the test statistic values (TS) with bootstrap critical values (CV$_1$, CV$_2$, and CV$_3$) corresponding to significance levels of 0.1, 0.05, and 0.01, respectively. Panel A provides results for the dataset  with forecast horizons exactly equal to 10 years, and Panel B presents the results for the expanded dataset including those with horizons close to 10 years.  
For each test, the value of the test statistic is highlighted in bold if it is greater than CV$_3$.}\label{Table_e3}
\setlength{\tabcolsep}{2pt}
\renewcommand{\arraystretch}{1}
\begin{tabular}{lllrrrrlrrrrlrrrr}
\hline\hline
&  &  & \multicolumn{4}{c}{\begin{tabular}[c]{@{}c@{}}Equity premium \\ (over yield)\end{tabular}} &  & \multicolumn{4}{c}{\begin{tabular}[c]{@{}c@{}}Equity premium \\ (over cash)\end{tabular}} &  & \multicolumn{4}{c}{\begin{tabular}[c]{@{}c@{}}Equity return \\ (nominal)\end{tabular}} \\
 &  &  & TS & CV$_1$ & CV$_2$ & CV$_3$ &  & TS & CV$_1$ & CV$_2$ & CV$_3$ &  & TS & CV$_1$ & CV$_2$ & CV$_3$ \\ \hline
\textbf{Panel A} &  &  &   &  &  &  &  &   &  &  &  &  &   &  &  &  \\
Model 1 & $\mathbb{H}_0^1$: $\mu_i = \mu$ &  & \textbf{6.94} & 3.52 & 4.13 & 5.50 &  & \textbf{5.47} & 2.89 & 3.25 & 3.96 &  & \textbf{8.12} & 2.64 & 3.25 & 3.75 \\
 &  &  &   &  &  &  &  &   &  &  &  &  &   &  &  &  \\
 Model 2& $\mathbb{H}_0^2$: $\alpha_i = \alpha$ &  & \textbf{6.34} & 2.58 & 3.15 & 4.41 &  & \textbf{8.09} & 2.47 & 2.90 & 4.12 &  & \textbf{6.23} & 1.84 & 2.10 & 2.75 \\
 & YYX     &  & \textbf{6.94}   & 0.52 & 1.05 & 1.91 &  & -0.01           & 0.52 & 1.05 & 1.91 &  & \textbf{15.64}  & 0.52 & 1.05 & 1.91 \\
 & $\mathbb{H}_0^3$: $\beta_i = \beta$ &  & \textbf{21.60} & 8.49 & 10.08 & 12.97 &  & \textbf{27.67} & 8.57 & 10.07 & 13.38 &  & \textbf{21.36} & 6.58 & 7.48 & 8.93 \\
 & PY      &  & \textbf{51.73}  & 1.28 & 1.64 & 2.33 &  & \textbf{16.15}  & 1.28 & 1.64 & 2.33 &  & \textbf{199.08} & 1.28 & 1.64 & 2.33 \\
 &  &  &   &  &  &  &  &   &  &  &  &  &   &  &  &  \\
 Model 3& $\mathbb{H}_0^2$: $\alpha_i = \alpha$ &  & \textbf{7.33} & 2.03 & 2.43 & 2.99 &  & \textbf{7.93} & 2.06 & 2.38 & 3.20 &  & \textbf{7.33} & 1.84 & 2.15 & 3.07 \\
 & $\mathbb{H}_0^3$: $\beta_i = \beta$ &  & \textbf{25.09} & 6.17 & 6.91 & 10.47 &  & \textbf{27.13} & 7.24 & 8.10 & 10.48 &  & \textbf{25.09} & 6.39 & 7.38 & 9.59 \\ \hline
\textbf{Panel B} &  &  &   &  &  &  &  &   &  &  &  &  &   &  &  &  \\
 Model 1& $\mathbb{H}_0^1$: $\mu_i = \mu$ &  & \textbf{9.91} & 5.43 & 6.90 & 9.71 &  & 9.14 & 5.45 & 6.67 & 10.78 &  & \textbf{12.33} & 5.12 & 6.56 & 10.00 \\
 &  &  &   &  &  &  &  &   &  &  &  &  &   &  &  &  \\
  Model 2 & $\mathbb{H}_0^2$: $\alpha_i = \alpha$ &  & \textbf{11.09} & 3.65 & 4.01 & 5.11 &  & \textbf{11.54} & 4.29 & 5.22 & 7.03 &  & \textbf{13.77} & 4.11 & 4.98 & 6.76 \\
  & YYX     &  & \textbf{44.00}  & 0.52 & 1.05 & 1.91 &  & \textbf{75.65}  & 0.52 & 1.05 & 1.91 &  & \textbf{82.31}  & 0.52 & 1.05 & 1.91 \\ 
 & $\mathbb{H}_0^3$: $\beta_i = \beta$ &  & \textbf{36.86} & 13.66 & 16.12 & 22.46 &  & \textbf{39.47} & 15.08 & 19.01 & 22.86 &  & \textbf{45.94} & 16.10 & 19.69 & 26.71 \\
 & PY      &  & \textbf{317.09} & 1.28 & 1.64 & 2.33 &  & \textbf{386.59} & 1.28 & 1.64 & 2.33 &  & \textbf{392.03} & 1.28 & 1.64 & 2.33\\
 &  &  &   &  &  &  &  &   &  &  &  &  &   &  &  &  \\
 Model 3& $\mathbb{H}_0^2$: $\alpha_i = \alpha$ &  & \textbf{11.36} & 3.62 & 4.80 & 6.91 &  & \textbf{11.60} & 4.12 & 5.33 & 7.36 &  & \textbf{11.36} & 3.94 & 5.11 & 6.60 \\
 & $\mathbb{H}_0^3$: $\beta_i = \beta$ &  & \textbf{37.04} & 13.34 & 17.76 & 25.59 &  & \textbf{39.68} & 15.99 & 19.15 & 25.27 &  & \textbf{37.04} & 13.53 & 14.91 & 21.49 \\
\hline\hline              
\end{tabular}
\end{center}
\end{table}

\begin{figure}[htp!]
	\centering
	\subfloat[Model 1, $epey$ ]
	{\includegraphics[width=0.35\textwidth]{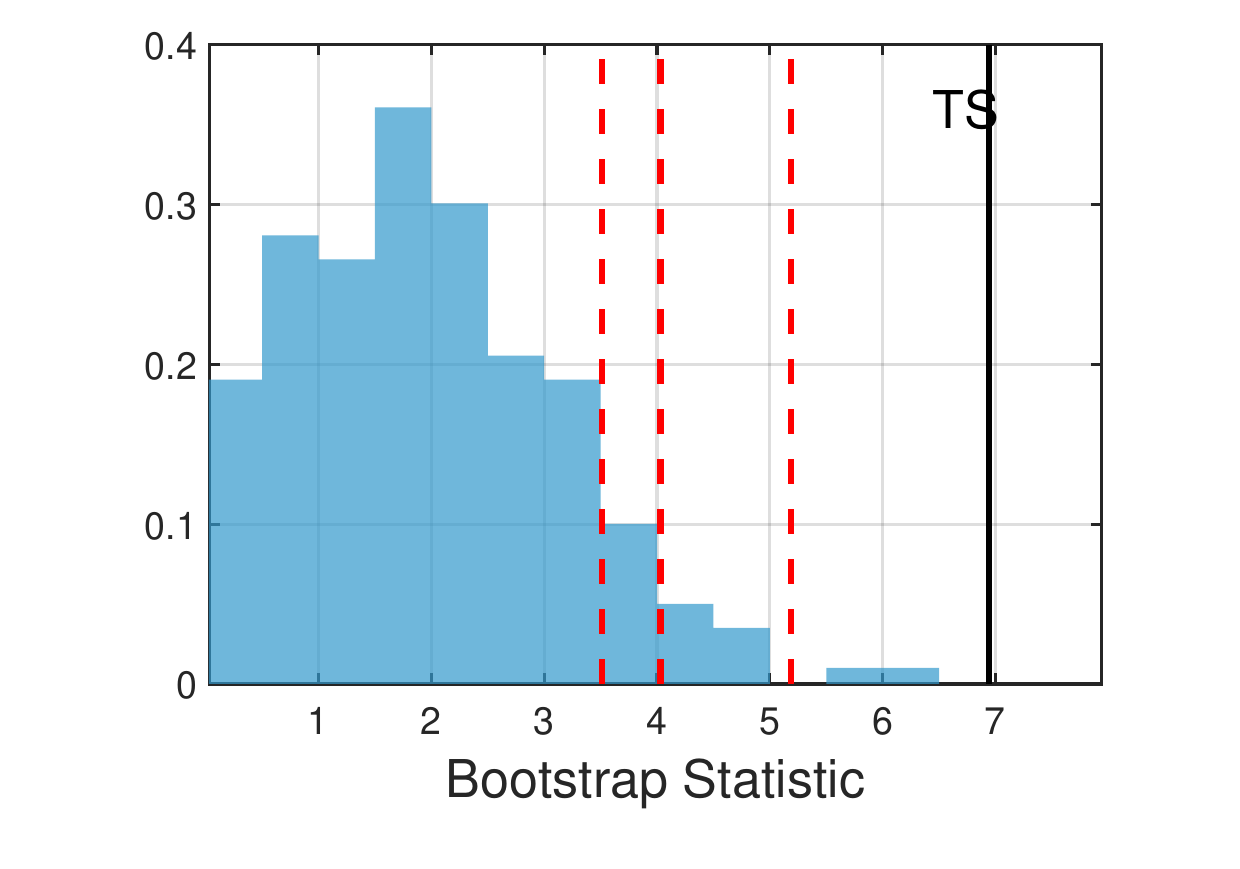}}
	\subfloat[Model 2, $epey$]
    {\includegraphics[width=0.35\textwidth]{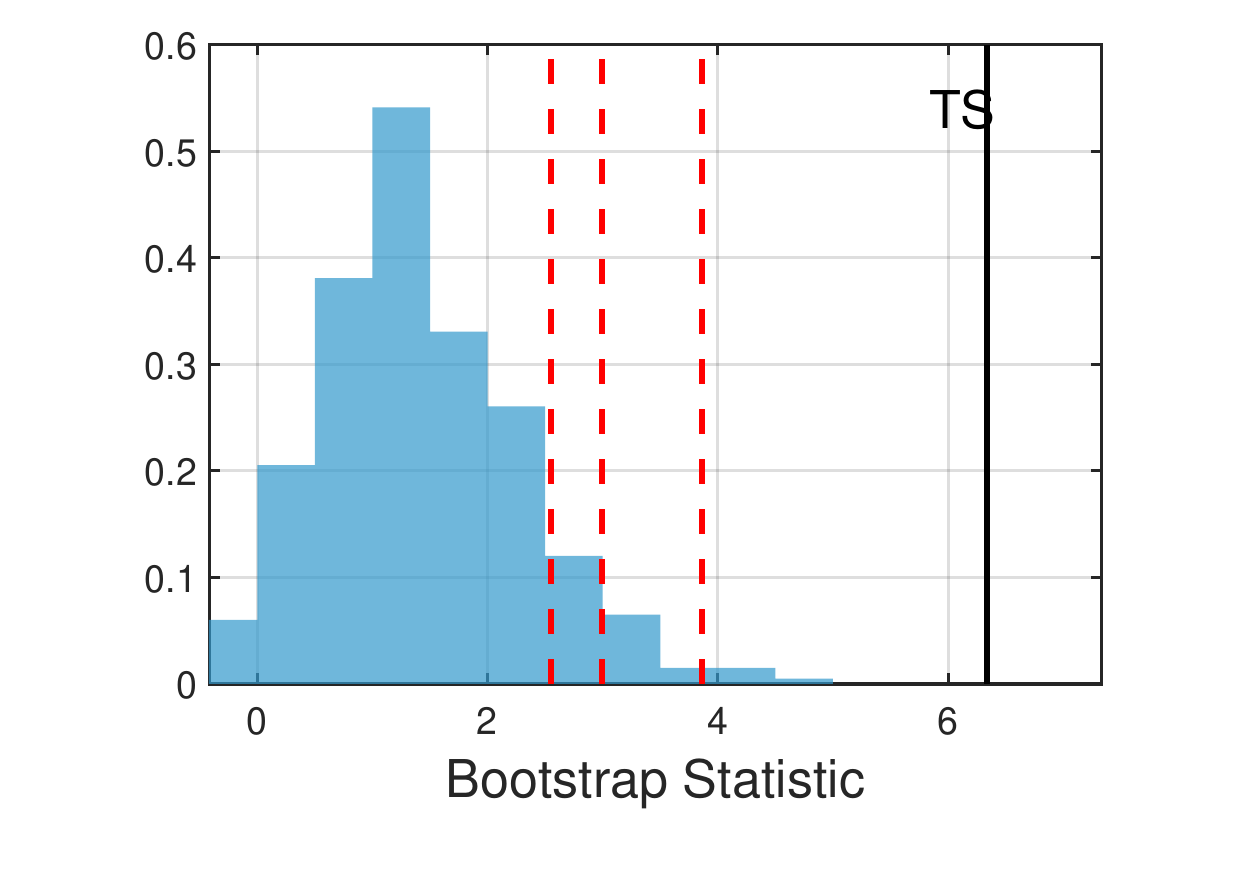}}
	\subfloat[Model 3, $epey$]
	{\includegraphics[width=0.35\textwidth]{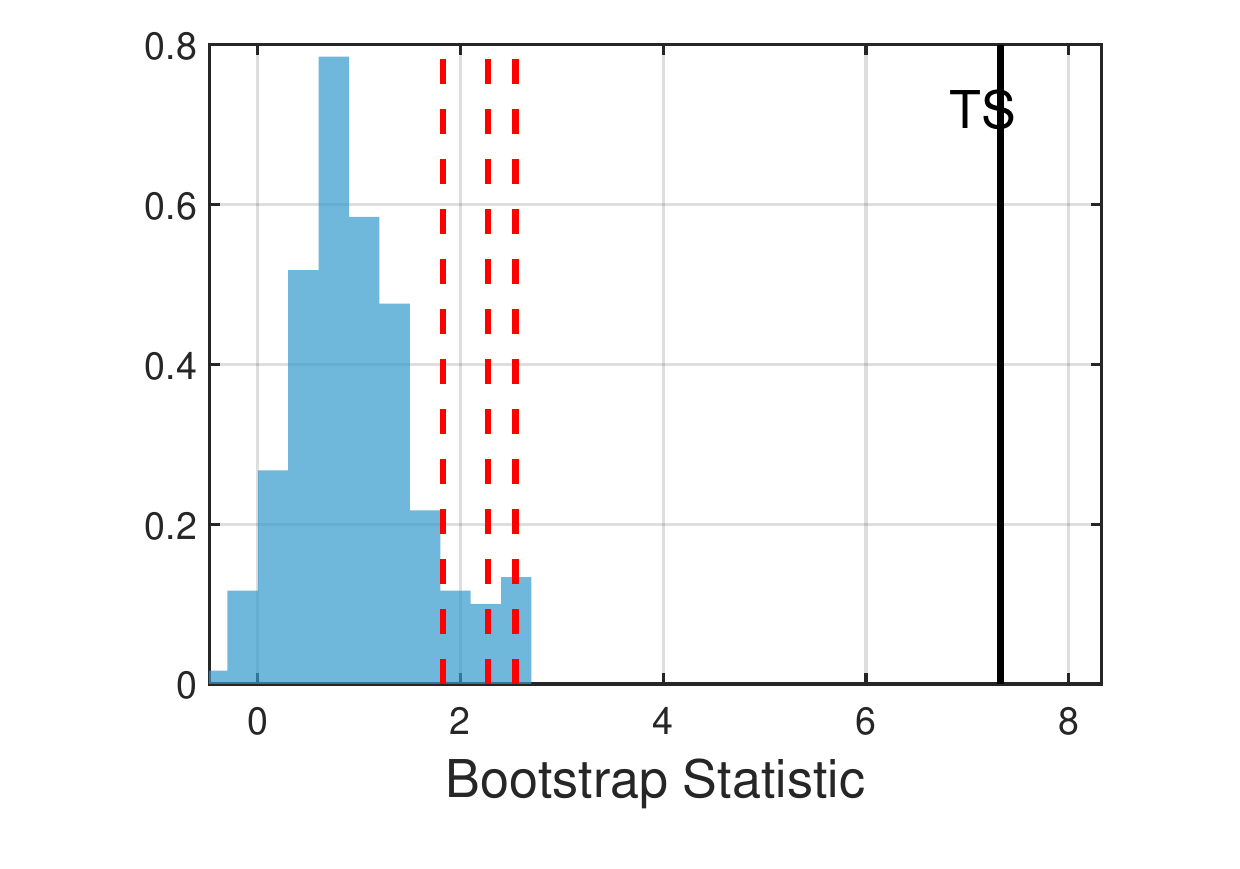}}\\
	\subfloat[Model 1, $epec$]
	{\includegraphics[width=0.35\textwidth]{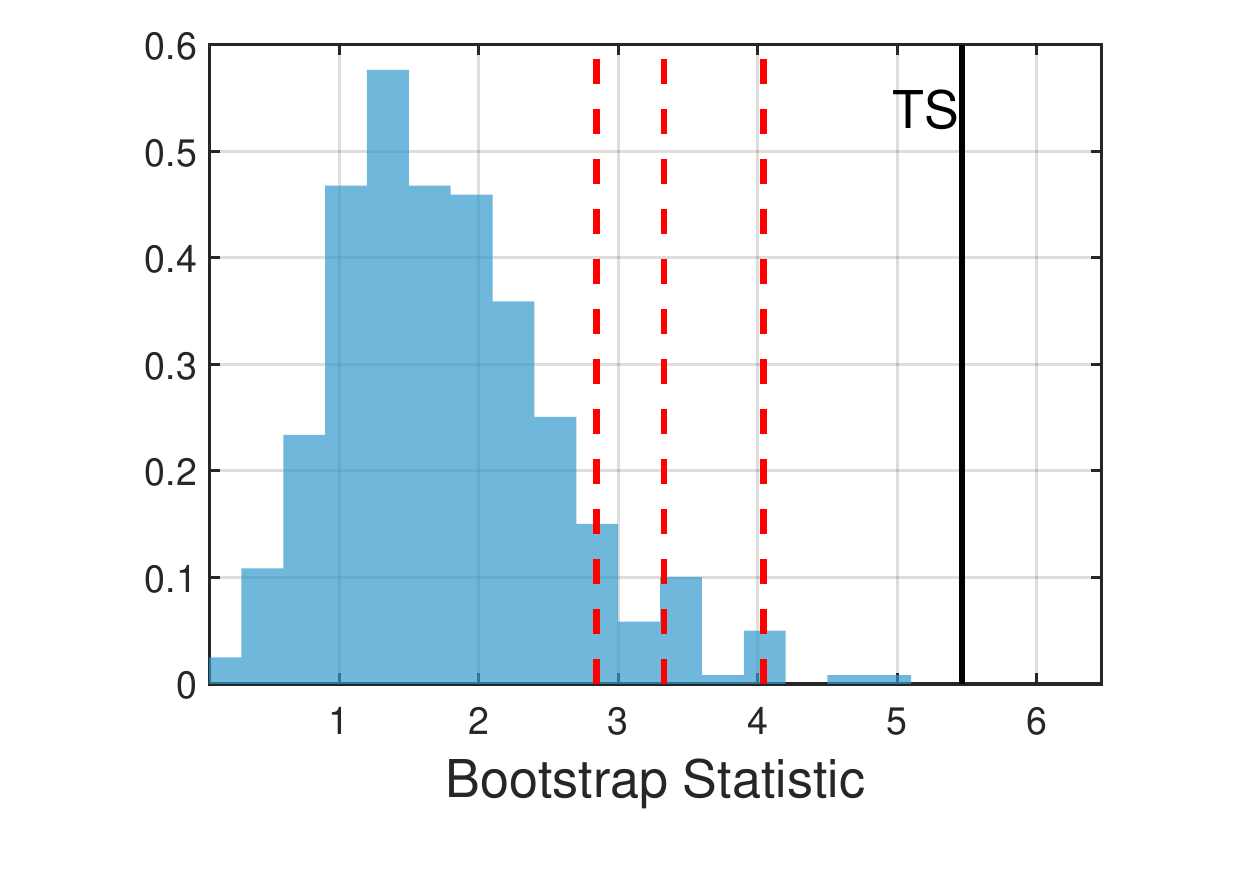}}	
    \subfloat[Model 2, $epec$]
	{\includegraphics[width=0.35\textwidth]{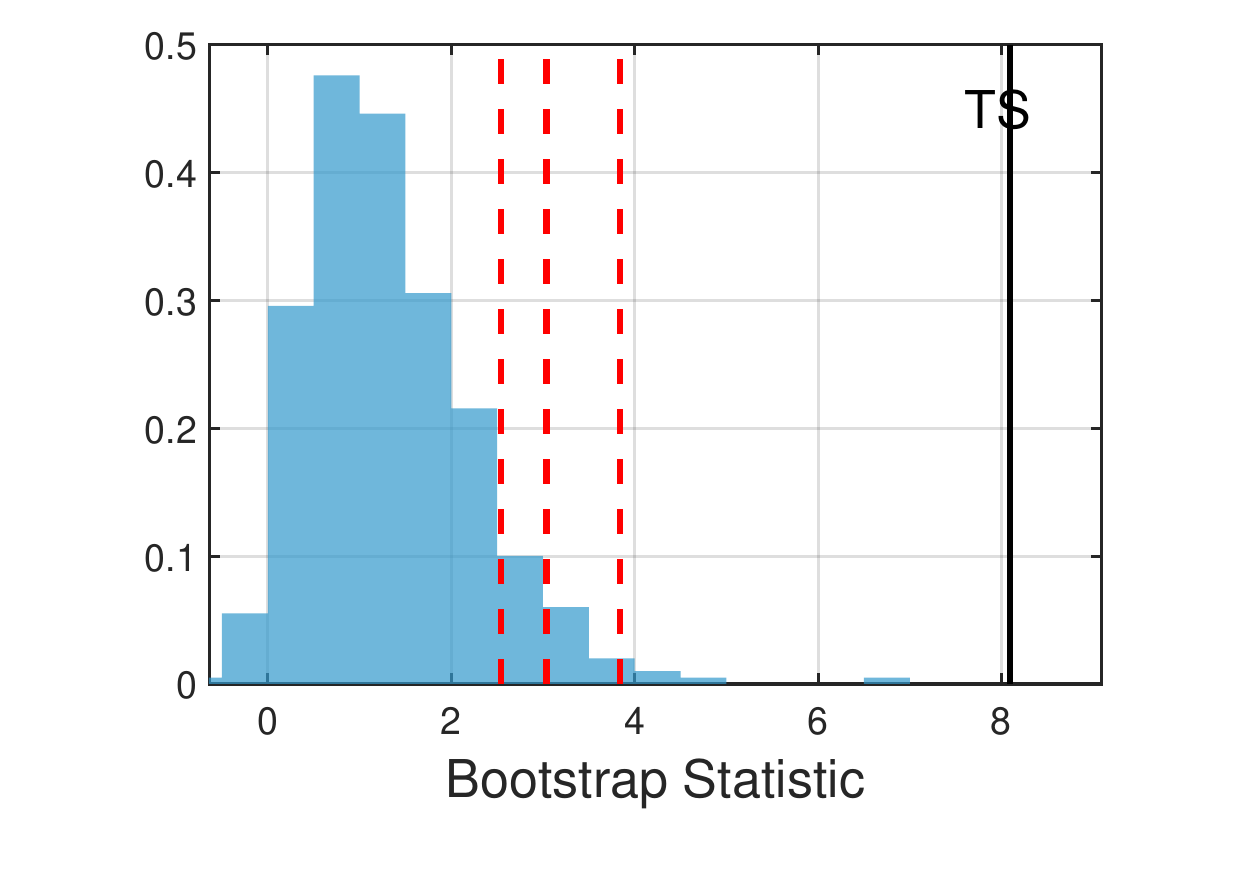}}
	\subfloat[Model 3, $epec$]
	{\includegraphics[width=0.35\textwidth]{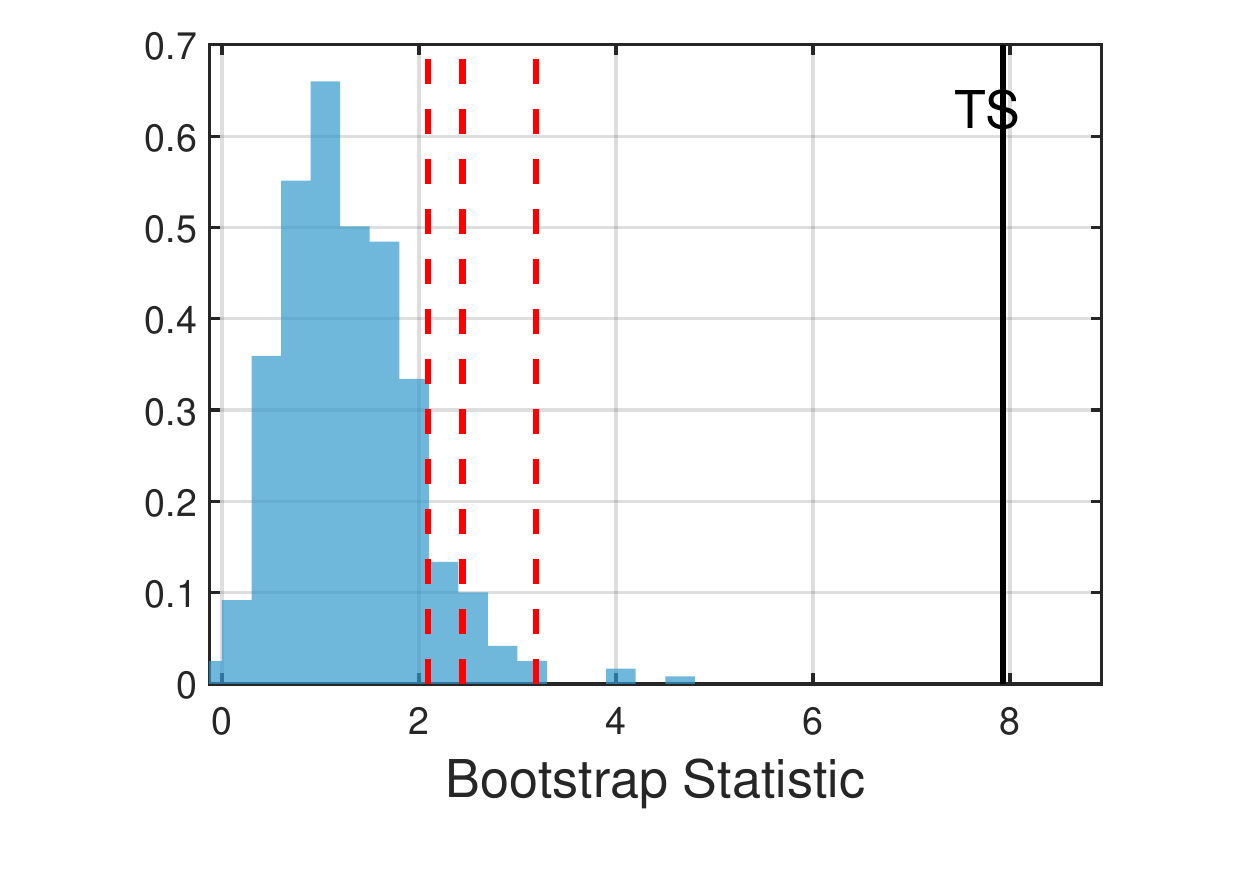}}\\
    \subfloat[Model 1, $ere$]
	{\includegraphics[width=0.35\textwidth]{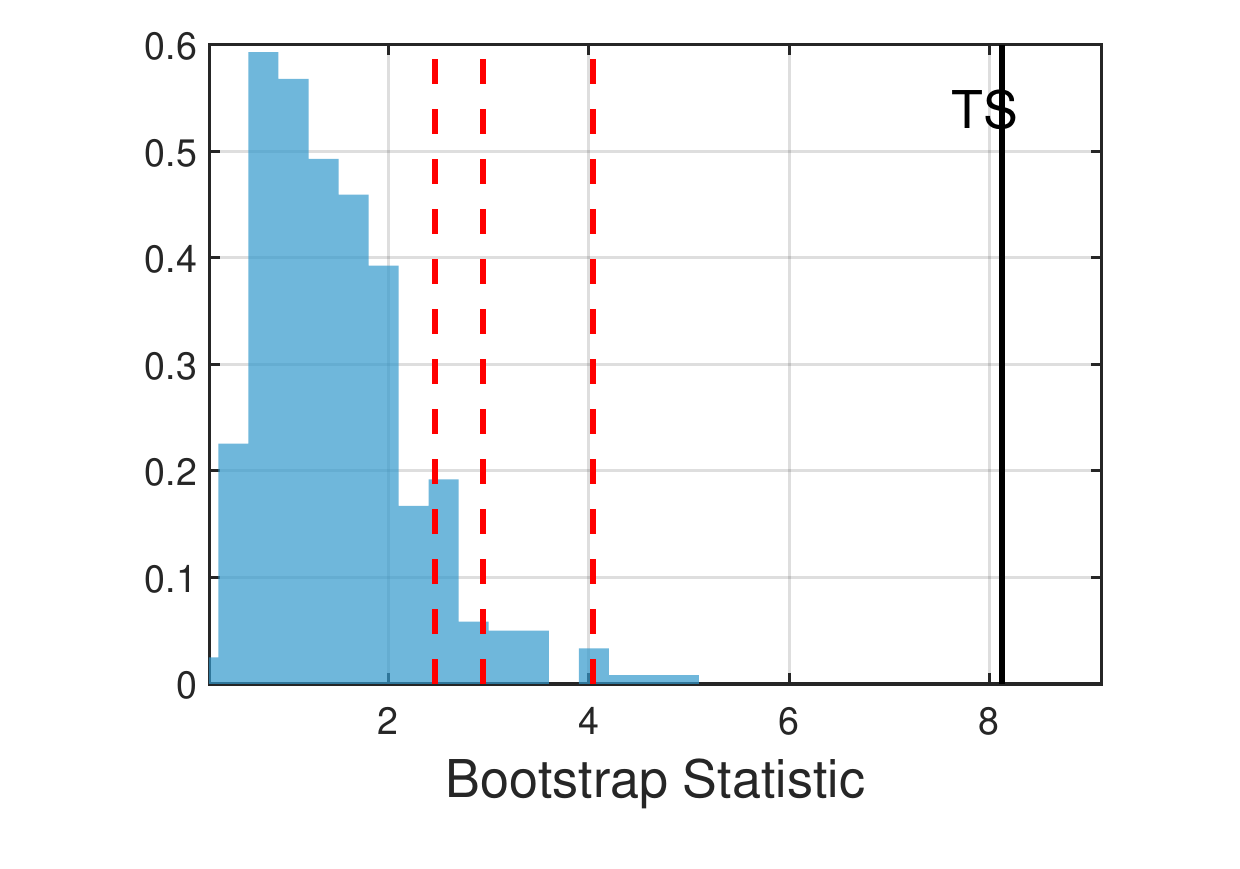}}
	\subfloat[Model 2, $ere$]
	{\includegraphics[width=0.35\textwidth]{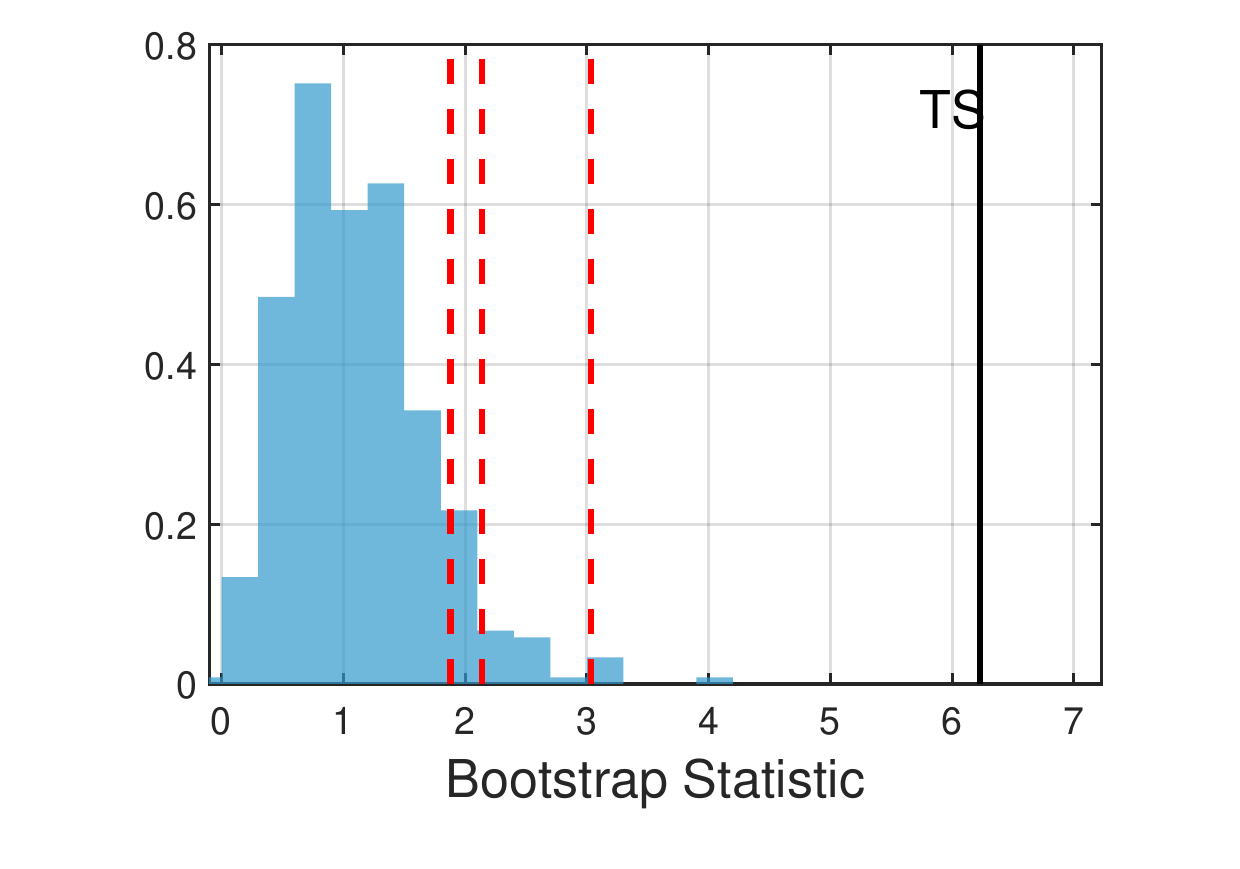}}
	\subfloat[Model 3, $ere$]
	{\includegraphics[width=0.35\textwidth]{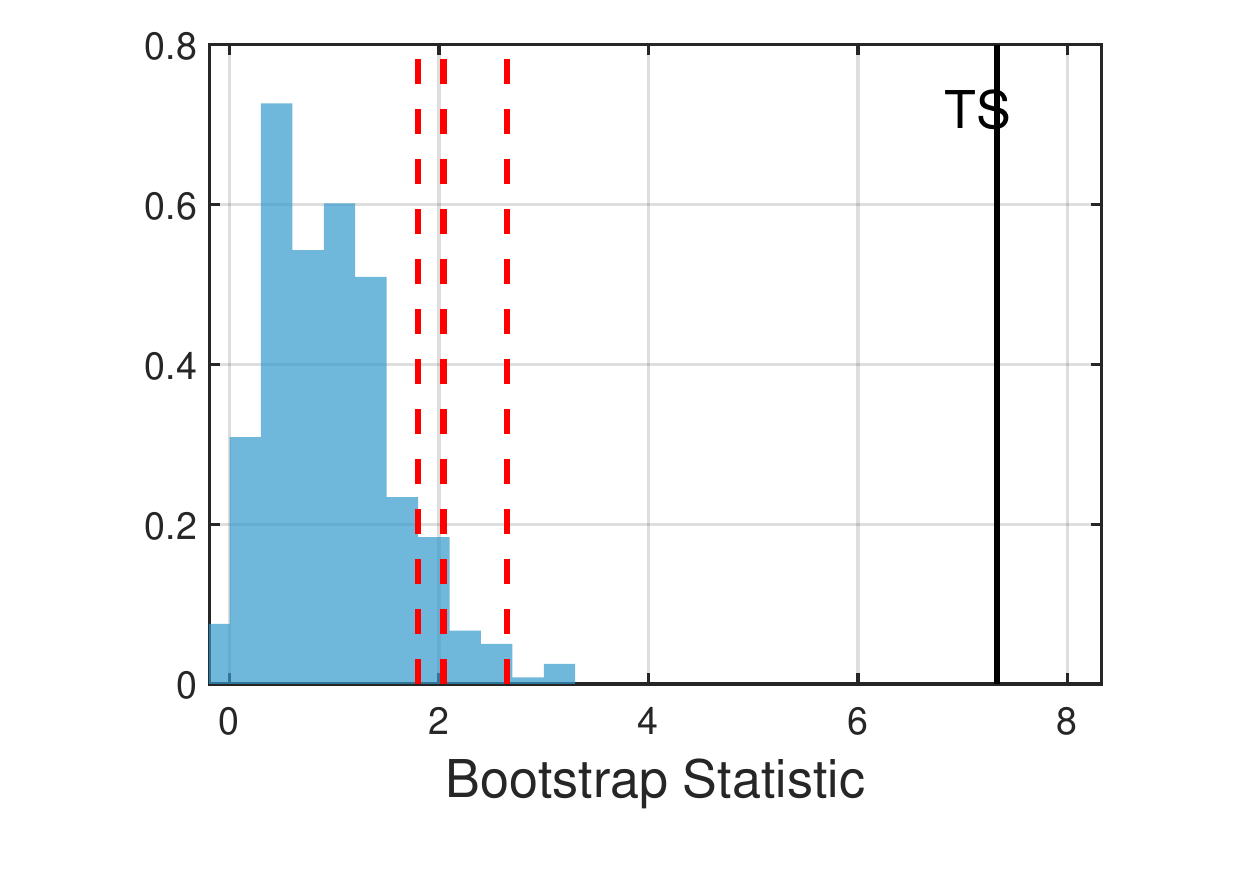}}\\
	\caption{\textbf{Distribution of Bootstrap Test Statistics.} In this figure, we present the distribution of bootstrap test statistics to infer the heterogeneity of the intercept in each model. The solid line in each plot represents the value of the test statistic, and three red dashed lines correspond to the bootstrap critical values at the 0.1, 0.05, and 0.01 significance levels.  }
	\label{Fig_Emp}
\end{figure}

\section{Conclusion}\label{Sec6}

In this paper, we introduce an underlying data generating process that allows for different magnitude of CD, along with TSA. This is achieved via high-dimensional moving average processes of infinite order (HDMA($\infty$)), which automatically generalizes the spatial structure introduced by Robinson and his co-authors in recent years. The framework is important in the sense that as noted by \citet[p. 187]{brockwell1991time} and \citet[pp. 33 \& 190]{FanYao}, the Wold decomposition theorem  ensures a formal linear representation exists for any stationary time series with no deterministic components, and HDMA($\infty$) naturally incorporates this result into a panel data framework. 

Our setup and investigation significantly integrates and enhances both homogenous and heterogeneous panel data modelling and testing (such as \citealp{Pesaran2006, PY2008, fan2015power, YU2024105458}).  To study HDMA($\infty$), we extend the BN decomposition (e.g., \citealp{BN1981,PS1992}) to a high-dimensional time series setting, and derive a complete set of toolkit. It is worth mentioning our investigation complements the work of \cite{fan2015power}, who specifically study cases where $\frac{T}{\sqrt{N}}\to 0$, by considering a broader range of scenarios and relaxing the independence assumptions employed in \cite{PY2008} and \cite{YU2024105458}.

We exam homogeneity against heterogeneity using Gaussian approximation, a prevalent technique for establishing uniform inference (e.g., \citealp{chernozhuokov2022improved}, and references therein). For post-testing inference, we derive Central Limit theorems through Edgeworth expansions for both homogenous and heterogeneous settings. Notably, the demand for Gaussian approximation in panel data analysis has been increasing recently, as exemplified in Section 4 of \cite{SJW_2024} and Section 4 of \cite{LLS2024}. Our study also contributes to this research direction by providing a set of foundational conditions and deriving a set of useful basic results.
    
We showcase the practical relevance of the established asymptotic theory by (1). connecting our results with the literature on grouping structure analysis such as those surveyed in \cite{BM2015} and \cite{SSP2016}, (2). examining a nonstationary panel data generating process presented in \cite{PM1999}, and (3). revisiting the common correlated effects (CCE) estimators of \cite{Pesaran2006}. Typically, when investigating nonstationary panel data, one has to impose cross-sectional independence such as \cite{PM1999}, \cite{DGP2021} and \cite{HUANG2021198} due to technical constraints. Our study offers a set of toolkit to account for the dependence of unit root precesses.

Finally, we verify our theoretical findings via extensive numerical studies using both simulated and real datasets.

{\small  

{\footnotesize
\setlength{\bibsep}{2.pt plus 0ex}
\bibliography{Ref.bib}
}

\newpage

\setcounter{page}{1}

\begin{center}
\begin{spacing}{1.5}
{\Large \bf Online Supplementary Appendices to \\``Panel Data Estimation and Inference: \\Homogeneity versus Heterogeneity"}
\end{spacing}

\bigskip

$^{\ast}${\sc Jiti Gao}, $^{\dag}${\sc Fei Liu}, $^{\ast}${\sc Bin Peng} and $^{\ddag}${\sc Yayi Yan}

$^{\ast}$Monash University

$^{\dag}$Nankai University

$^\ddag$Shanghai University of Finance and Economics
\end{center}

\renewcommand{\theequation}{A.\arabic{equation}}
\renewcommand{\thesection}{A\arabic{section}}
\renewcommand{\thefigure}{A\arabic{figure}}
\renewcommand{\thetable}{A\arabic{table}}
\renewcommand{\thelemma}{A\arabic{lemma}}
\renewcommand{\theremark}{A\arabic{remark}}
\renewcommand{\thecorollary}{A\arabic{corollary}}

\setcounter{equation}{0}
\setcounter{lemma}{0}
\setcounter{section}{0}
\setcounter{table}{0}
\setcounter{figure}{0}
\setcounter{remark}{0}
\setcounter{corollary}{0}

\bigskip

Throughout the proofs, we suppose that without loss of generality, $\mu \equiv 0$ for the homogenous case, and $\mu_i\equiv 0$ for all $ i\in [N]$ for the heterogeneous case. We shall not mention them again unless misunderstanding may arise.

This document is structured as follows: Appendix \ref{SecA0} points out that the proposed inference methods can accommodate unbalanced panels after some necessary modifications;  Appendix \ref{SecA1} provides some foundational facts used throughout the proofs; Appendix \ref{SecA2} contains all preliminary lemmas; Appendix \ref{SecA3} presents the theoretical proofs of the main results; and Appendix \ref{SecA4} details the proofs of the preliminary lemmas.

\section{Unbalanced Panel Data}\label{SecA0}

More often than not, one encounters unbalanced panel data practically. Given the missing proportion is asymptotically negligible, the proposed inference methods can accommodate unbalanced panels after some necessary modifications. 

For instance, we let $\mathbb{T}_i$ denote the  sample set for individual $i$ and $T_i\coloneqq \sharp \mathbb{T}_i$. The $L_\infty$-based test statistic  is then redefined as

$$
Q_{NT} = \max_i\left|\sqrt{T_i}(\overline{x}_i - \overline{x})\right|,
$$
where  $\overline{x}_i = \frac{1}{T_i}\sum_{t=1}^{T_i}x_{it}$ and $\overline{x} = \frac{1}{N}\sum_{i=1}^{N}\frac{1}{T_i}\sum_{t=1}^{T_i}x_{it}$. 

Additionally, the high-dimensional long-run covariance matrix estimator for unbalanced panels is defined as $ \widehat{\bm{\Omega}}\coloneqq (\widehat{\Omega}_{in})$, where

\begin{eqnarray*}
    \widehat{\Omega}_{in}\coloneqq \frac{1}{\sqrt{T_iT_n}}\sum_{t\in \mathbb{T}_i}\sum_{s\in \mathbb{T}_n}a\left(\frac{t-s}{\widetilde{m}}\right)(x _{it}-\overline{x}_i)(x _{ns}-\overline{x}_n),
\end{eqnarray*}
and the Gaussian multiplier bootstrap approximation is accordingly defined by
\begin{eqnarray*}
    \left|\widehat{\bm{\Omega}}^{1/2}\Big(\frac{1}{\sqrt{T_1}}\sum_{t\in \mathbb{T}_1}z_{1t}^*,\cdots, \frac{1}{\sqrt{T_N}}\sum_{t\in \mathbb{T}_N}z_{Nt}^*\Big)^\top \right|_{\infty},
\end{eqnarray*}
where $\{z_{it}^* \mid t\in\mathbb{T}_i \}$ is a sequence of i.i.d.  $ N(0,1)$. 

For the CCE estimation, the estimators and test statistics for unbalanced panels can also be updated. Let $\mathbb{N}_t$ contain the indices for individuals that have valid observations at time $t$ and let $N_t$ be the number of indices in this set. Additionally, for each $i$, define $\overline{\mathbf{Z}}_{\cdot i}\coloneqq (\overline{\mathbf{z}}_{t_1},\cdots, \overline{\mathbf{z}}_{t_{T_i}})^\top$, where $t_1,\cdots,t_{T_i}\in \mathbb{T}_i$ and  $\overline{\mathbf{z}}_t\coloneqq \frac{1}{N_t}\sum_{i\in\mathbb{N}_t}\mathbf{z}_{it}$.

The CCE estimators of $\pmb{\theta}_i$ and $\pmb{\theta}$  can be then defined as 

\begin{eqnarray*}
\widehat{\pmb{\theta}}_i &=&(\mathbf{W}_i^\top \mathbf{M}_{\overline{\mathbf{Z}}_{\cdot i}} \mathbf{W}_i)^{-1} \mathbf{W}_i^\top \mathbf{M}_{\overline{\mathbf{Z}}_{\cdot i}} \mathbf{Y}_i \quad \text{for} \quad \forall i\in [N],\notag \\
\widehat{\pmb{\theta}}&=&\left(\sum_{i=1}^N\mathbf{W}_i^\top \mathbf{M}_{\overline{\mathbf{Z}}_{\cdot i}} \mathbf{W}_i\right)^{-1} \sum_{i=1}^N\mathbf{W}_i^\top \mathbf{M}_{\overline{\mathbf{Z}}_{\cdot i}} \mathbf{Y}_i,
\end{eqnarray*}
where $\mathbf{M}_{\overline{\mathbf{Z}}_{\cdot i}} =\mathbf{I}_T-\overline{\mathbf{Z}}_{\cdot i}(\overline{\mathbf{Z}}_{\cdot i}^\top \overline{\mathbf{Z}}_{\cdot i})^{-1} \overline{\mathbf{Z}}_{\cdot i}^\top $.

The heterogeneity test statistic is then given by

\begin{eqnarray*}
Q_{j} = \max_i \frac{1}{\sqrt{T}} \mathbf{e}_j^\top (\mathbf{W}_i^\top \mathbf{M}_{\overline{\mathbf{Z}}_{\cdot i}} \mathbf{W}_i)(\widehat{\pmb{\theta}}_i-\widehat{\pmb{\theta}}) ,
\end{eqnarray*}
where $\mathbf{e}_j$ is a selection vector. For constructing the bootstrap statistics, we use $\widehat{\pmb{\Omega}}_{j}=(\widehat{\pmb{\Omega}}_{j,in})$, where 

\begin{eqnarray*}
\widehat{\pmb{\Omega}}_{j,in} &=& \frac{1}{\sqrt{T_iT_n}}\sum_{t\in\mathbb{T}_i}\sum_{s\in\mathbb{T}_n}a\left(\frac{t-s}{\widetilde{m}}\right)\widehat{\xi}_{j,it}\widehat{\xi}_{j,ns},
\end{eqnarray*}
where $\widehat{\xi}_{j,it}$ denotes the $t$-th element of $\widehat{\pmb{\mathbf{U}}}_{i,1}\circ \widehat{\pmb{\mathbf{U}}}_{i,1+j}$, with
 $\widehat{\pmb{\mathbf{U}}}_i=  \mathbf{M}_{\overline{\mathbf{Z}}_{\cdot i}} \mathbf{Z}_i,$  and $\widehat{\pmb{\mathbf{U}}}_{i,j}$ standing for the $j^{th}$ column of $\widehat{\pmb{\mathbf{U}}}_{i}$. The distribution of $Q_j$ is then approximated by 
 \begin{eqnarray*}
    \left|\widehat{\bm{\Omega}}_j^{1/2}\Big(\frac{1}{\sqrt{T_1}}\sum_{t\in \mathbb{T}_1}z_{1t}^*,\cdots, \frac{1}{\sqrt{T_N}}\sum_{t\in \mathbb{T}_N}z_{Nt}^*\Big)^\top \right|_{\infty},
\end{eqnarray*}
where $\{z_{it}^* \mid t\in\mathbb{T}_i \}$ is a sequence of i.i.d.  $ N(0,1)$.

\section{Some Facts}\label{SecA1}

We present a few facts in this section, which will be repeatedly used in the proofs.

\smallskip

\noindent \textbf{On Cumulant} --- Note that for a generic cumulant, we have for a constant $c$ 

\begin{eqnarray}\label{prop.cum}
\kappa_r(cx)=c^r \kappa_r(x).
\end{eqnarray}
We refer the interested reader to \cite{saulis} for more details about cumulants. 

\medskip

\noindent \textbf{On BN Decomposition} --- Simple algebra shows the Beveridge and Nelson (BN) decomposition under the high-dimensional (HD) setting is as follows:

\begin{eqnarray}\label{def.BN1}
\mathbf{B}(L) = \mathbf{B} -(1-L)\widetilde{\mathbf{B}}(L) ,
\end{eqnarray}
where $\widetilde{\mathbf{B}}(L)\coloneqq\sum_{\ell=0}^{\infty}\widetilde{\mathbf{B}}_\ell L^\ell$ with $\widetilde{\mathbf{B}}_\ell\coloneqq \sum_{k=\ell+1}^{\infty}\mathbf{B}_k$. 

\medskip

\noindent \textbf{On Dependence} --- To calculate CD, for $\forall i,j\in [N]$ it is easy to obtain that

\begin{eqnarray*}
E[x_{i1}x_{j1}] = \sum_{\ell =0}^{\infty} \mathbf{b}_{\ell i}^{\sharp\top} \mathbf{b}_{\ell j}^\sharp.
\end{eqnarray*}

To calculate TSA, for $t> s$ we write

\begin{eqnarray*}
E[x_{it}x_{is}] &=& E\left[\left(\sum_{\ell=0}^{\infty} \mathbf{b}_{\ell i}^{\sharp\top}\pmb{\varepsilon}_{t-\ell} \right)\left(\sum_{\ell=0}^{\infty} \mathbf{b}_{\ell i}^{\sharp\top}\pmb{\varepsilon}_{s-\ell} \right)\right] \notag \\
&=& E\left[\left(\sum_{\ell=t-s}^{\infty} \mathbf{b}_{\ell i}^{\sharp\top}\pmb{\varepsilon}_{t-\ell} \right)\left(\sum_{\ell=0}^{\infty} \mathbf{b}_{\ell i}^{\sharp\top}\pmb{\varepsilon}_{s-\ell} \right)\right] \notag \\
&=& E\left[\left(\sum_{\ell=0}^{\infty} \mathbf{b}_{\ell+t-s, i}^{\sharp\top}\pmb{\varepsilon}_{s-\ell} \right)\left(\sum_{\ell=0}^{\infty} \mathbf{b}_{\ell i}^{\sharp\top}\pmb{\varepsilon}_{s-\ell} \right)\right] \notag \\
&=& \sum_{\ell=0}^{\infty} \mathbf{b}_{\ell+t-s, i}^{\sharp\top} \mathbf{b}_{\ell i}^\sharp.
\end{eqnarray*}

To calculate CD + TSA, for $t> s$ we write

\begin{eqnarray*}
E[x_{it}x_{js}] &=& E\left[\left(\sum_{\ell=0}^{\infty} \mathbf{b}_{\ell i}^{\sharp\top}\pmb{\varepsilon}_{t-\ell} \right)\left(\sum_{\ell=0}^{\infty} \mathbf{b}_{\ell j}^{\sharp\top}\pmb{\varepsilon}_{s-\ell} \right)\right] \notag \\
&=& E\left[\left(\sum_{\ell=t-s}^{\infty} \mathbf{b}_{\ell i}^{\sharp\top}\pmb{\varepsilon}_{t-\ell} \right)\left(\sum_{\ell=0}^{\infty} \mathbf{b}_{\ell j}^{\sharp\top}\pmb{\varepsilon}_{s-\ell} \right)\right] \notag \\
&=& E\left[\left(\sum_{\ell=0}^{\infty} \mathbf{b}_{\ell+t-s, i}^{\sharp\top}\pmb{\varepsilon}_{s-\ell} \right)\left(\sum_{\ell=0}^{\infty} \mathbf{b}_{\ell j}^{\sharp\top}\pmb{\varepsilon}_{s-\ell} \right)\right] \notag \\
&=& \sum_{\ell=0}^{\infty} \mathbf{b}_{\ell+t-s, i}^{\sharp\top}\mathbf{b}_{\ell j}^\sharp.
\end{eqnarray*}

\medskip

\noindent \textbf{On Cumulants} --- Recall the notation of Section \ref{Sec1}. Let $\widetilde{\phi}(x)$ be the characteristic function of a standard normal distribution. We further define its cumulants by $\gamma_r$ for $r\ge 1$. We now consider a generic distribution function $G(x)$ with a characteristic functions $\chi(u)$ and cumulants $\beta_r$. Then by Taylor expansion, we have

\[
\log \frac{\chi(u)}{\widetilde{\phi}(x)} =\log \chi(u) - \log \widetilde{\phi}(x)=\sum_{r=1}^\infty (\beta_r-\gamma_r)\cdot \frac{(\mathsf{i} u)^r}{r!}
\]
and 

\begin{eqnarray}\label{def.chiu}
\chi(u) =\widetilde{\phi}(x)\cdot\exp\left\{ \sum_{r=1}^\infty (\beta_r-\gamma_r)\cdot \frac{(\mathsf{i} u)^r}{r!}\right\}.
\end{eqnarray}

\medskip

\noindent \textbf{On Exponential function} --- By Taylor expansion,

\begin{eqnarray}\label{def.expu}
\exp(u)=\sum_{r=0}^\infty \frac{u^r}{r!}.
\end{eqnarray}
Taylor theorem yields that

\begin{eqnarray}\label{def.expu2}
\exp(u)-\sum_{r=0}^k\frac{u^r}{r!} =\frac{\exp(\widetilde{u})}{(k+1)!}u^{k+1},
\end{eqnarray}
where $\widetilde{u}$ is between 0 and $u$.

\section{Preliminary Lemmas}\label{SecA2}

In this appendix, we provide some useful preliminary lemmas. The first three lemmas are either obvious or have been carefully proved in the literature, so we do not provide the proofs herewith. We will provide the proofs for Lemmas \ref{LM.A4} to \ref{LM.A8}.

\begin{lemma}\label{LM.A1}
A matrix $\mathbf{A}$ in $\mathbb{R}^{n\times n}$ has rank one if and only if it can be written as the outer product of two nonzero vectors in $\mathbb{R}^{n}$ (i.e., $\mathbf{A} =\mathbf{x}\mathbf{y}^\top$).
\end{lemma}

\begin{lemma}[Esseen's smoothing Lemma]\label{LM.A2}
Let $H(\cdot)$ be a distribution with 0 expectation and characteristic function $\chi(\cdot)$. Suppose $H(x)-G(x)$ vanishes at $\pm \infty$ and that $G(\cdot)$ has a derivative $g(\cdot)$ such that $|g|_\infty \le m$. Finally, suppose that $g$ has a continuously differentiable Fourier transform $\xi(\cdot)$ such that $\xi(0)=1$ and $\xi^{(1)}(0)=0$. Then

\begin{eqnarray*}
|H -G|_\infty\le \frac{1}{\pi}\int_{-a}^a \left|\frac{\chi(u)-\xi(u)}{u} \right|\mathrm{d}u+\frac{24m}{\pi \cdot a}
\end{eqnarray*}
where $a >0$.
\end{lemma}

See Chapter XVI, section 3 of  \cite{feller1971introduction} for details about Esseen's smoothing Lemma.

\begin{lemma}\label{LM.A3}
Let $\mathsf{i}$ be the imaginary unit. Then for $\forall \theta\in [0,1)$

\begin{eqnarray*}
\sup_x\left|\exp(\mathsf{i}x) -\sum_{j=0}^r \frac{(\mathsf{i}x)^j}{j!}  \right| \le \min\left\{\frac{2}{r!}|x|^{r+\theta}, \frac{|x|^{r+1}}{(r+1)!} \right\}.
\end{eqnarray*}
\end{lemma}

See  \citet[3.2]{tik1981} for example.

\begin{lemma}\label{LM.A5}
\item 

Under Assumption \ref{AS1}, the following results hold:
\begin{enumerate}[leftmargin=24pt, parsep=2pt, topsep=2pt] 
\item For $t> 1$, $\sum_{s=1}^t \mathbf{x}_{s}$ admits two representations:

\begin{enumerate}[leftmargin=24pt, parsep=2pt, topsep=2pt] 
\item $\sum_{s=1}^t \mathbf{x}_{s} = \mathbf{B} \sum_{s=1}^t\pmb{\varepsilon}_{s}-\widetilde{\mathbf{B}}(L)\pmb{\varepsilon}_{t}+\widetilde{\mathbf{B}}(L)\pmb{\varepsilon}_{0}$,

\item  $\sum_{s=1}^t \mathbf{x}_{s} = \sum_{\ell=1}^t (\mathbf{B}-\widetilde{\mathbf{B}}_{t-\ell})\pmb{\varepsilon}_{\ell}-\sum_{\ell=-\infty}^{0}(\widetilde{\mathbf{B}}_{t-\ell} -\widetilde{\mathbf{B}}_{-\ell})\pmb{\varepsilon}_{\ell}\eqqcolon \sum_{\ell=-\infty}^t \pmb{\mathcal{B}}_{t\ell} \pmb{\varepsilon}_\ell$,
\end{enumerate}
where $\sum_{\ell=0}^{\infty}\sqrt{\frac{N}{L_N}}\|\widetilde{\mathbf{B}}_\ell\|_2<\infty$, and $\pmb{\mathcal{B}}_{t\ell} = - \widetilde{\mathbf{B}}_{t-\ell} + \widetilde{\mathbf{B}}_{-\ell}$ for $-\infty\leq \ell\leq 0$; $\pmb{\mathcal{B}}_{t\ell} = \mathbf{B}-\widetilde{\mathbf{B}}_{t-\ell}$ for $1\leq \ell \leq t$;

\item $\left|\frac{1}{\sqrt{L_N}} \sum_{\ell=-\infty}^0 \mathbf{1}_N^\top\pmb{\mathcal{B}}_{T\ell} \pmb{\varepsilon}_{\ell}\right|=O_P(1)$;

\item $\left|\frac{1}{\sqrt{L_N}}\sum_{t=1}^T \mathbf{1}_N^\top\widetilde{\mathbf{B}}_{T-\ell} \pmb{\varepsilon}_{\ell} \right|=O_P(1)$.
\end{enumerate}
\end{lemma}

Write

\begin{eqnarray}\label{def.xx}
\mathbf{1}_N^\top \mathbf{x}_{t}\mathbf{x}_{t}^\top \mathbf{1}_N&=&\sum_{\ell=0}^{\infty} \mathbf{1}_N^\top \mathbf{B}_\ell \pmb{\varepsilon}_{t-\ell}\pmb{\varepsilon}_{t-\ell}^\top \mathbf{B}_\ell^\top \mathbf{1}_N  +2\sum_{v=1}^{\infty}\sum_{\ell=0}^{\infty}\mathbf{1}_N^\top \mathbf{B}_{\ell+v} \pmb{\varepsilon}_{t-\ell-v} \pmb{\varepsilon}_{t-\ell}^\top \mathbf{B}_\ell^\top  \mathbf{1}_N\notag \\
&= & \sum_{\ell=0}^{\infty} [(\mathbf{1}_N^\top \mathbf{B}_\ell)\otimes (\mathbf{1}_N^\top \mathbf{B}_\ell)]\text{vec}(\pmb{\varepsilon}_{t-\ell}  \pmb{\varepsilon}_{t-\ell}^\top)\notag \\
&&+2\sum_{v=1}^{\infty}\sum_{\ell=0}^{\infty}[(\mathbf{1}_N^\top \mathbf{B}_\ell)\otimes (\mathbf{1}_N^\top \mathbf{B}_{\ell+v})] \text{vec}(\pmb{\varepsilon}_{t-\ell-v} \pmb{\varepsilon}_{t-\ell}^\top )\notag \\
&\eqqcolon & \mathbf{B}_0^*(L)\text{vec}(\pmb{\varepsilon}_{t}  \pmb{\varepsilon}_{t}^\top) +\sum_{v=1}^\infty \mathbf{B}_v^*(L)\text{vec}(\pmb{\varepsilon}_{t-v} \pmb{\varepsilon}_{t}^\top ),
\end{eqnarray}
where the definition of $\mathbf{B}_v^*(L)$ for $v\ge 0$ should be obvious. In connection with the BN decomposition of \eqref{def.BN1}, we have

\begin{eqnarray}\label{def.BN2}
\mathbf{B}_v^*(L) = \mathbf{B}_v^* -(1-L)\widetilde{\mathbf{B}}_v^*(L) 
\end{eqnarray}
where 

\begin{eqnarray*}
&&\mathbf{B}_v^*\coloneqq \sum_{\ell=0}^{\infty} [(\mathbf{1}_N^\top \mathbf{B}_\ell)\otimes (\mathbf{1}_N^\top \mathbf{B}_{\ell+v})] ,\notag \\
&&\widetilde{\mathbf{B}}_v^*(L)\coloneqq\sum_{\ell=0}^{\infty}\widetilde{\mathbf{B}}_{v \ell}^* L^\ell\quad\text{with}\quad \widetilde{\mathbf{B}}_{v \ell}^*\coloneqq \sum_{k=\ell+1}^{\infty}(\mathbf{1}_N^\top \mathbf{B}_k)\otimes (\mathbf{1}_N^\top \mathbf{B}_{k+v}).
\end{eqnarray*}

\begin{lemma}\label{LM.A6}  
\item 

Under Assumptions \ref{AS1} and \ref{AS2}, the following results hold:
    
\begin{enumerate}[leftmargin=24pt, parsep=2pt, topsep=2pt] 
 
\item  $\| E[(\pmb{\varepsilon}_{t}\otimes \pmb{\varepsilon}_{t})(\pmb{\varepsilon}_{t}^\top \otimes \pmb{\varepsilon}_{t}^\top)]\|_2 =O(N)$;

\item  $\left|\frac{1}{L_NT}\sum_{t=1}^T \mathbf{B}_0^*(L)\text{\normalfont vec}(\pmb{\varepsilon}_{t}  \pmb{\varepsilon}_{t}^\top) -   \frac{1}{L_N} \sum_{\ell=0}^{\infty} \mathbf{1}_N^\top \mathbf{B}_\ell \mathbf{B}_\ell^\top \mathbf{1}_N \right|=O_P\left(\frac{1}{\sqrt{T}}\right)$;

\item $\left|\frac{1}{L_NT}\sum_{t=1}^T\sum_{v=1}^\infty \mathbf{B}_v^*(L)\text{\normalfont vec}(\pmb{\varepsilon}_{t-v} \pmb{\varepsilon}_{t}^\top ) \right|=O_P\left(\frac{1}{\sqrt{T}}\right)$.
   
\end{enumerate}
\end{lemma}

\section{Proofs of the Main Results}\label{SecA3}

\begin{proof}[Proof of Theorem \ref{THM.1}]
\item 

First, we denote the following notation to facilitate the development. Let

\begin{eqnarray}\label{def_bdd}
\mathbf{B} =(\mathbf{b}_1^\dag,\ldots, \mathbf{b}_N^\dag)\quad \text{and}\quad\widetilde{\mathbf{B}}_{\ell}=(\widetilde{\mathbf{b}}_{\ell, 1}^\dag,\ldots, \widetilde{\mathbf{b}}_{\ell, N}^\dag)
\end{eqnarray}
where $\mathbf{B}$ and $\widetilde{\mathbf{B}}_{\ell}$ have been defined in \eqref{def.BN1} already.

By Lemma \ref{LM.A5}, we write

\[
\sum_{s=1}^t \mathbf{x}_{s} =  \sum_{\ell=-\infty}^t \pmb{\mathcal{B}}_{t\ell}\pmb{\varepsilon}_{\ell},
\]
where the definition of $\pmb{\mathcal{B}}_{t\ell}$ is the same as that in Lemma \ref{LM.A5}. Thus, $\widetilde{S}_{NT}$ can also be written as

\[
\widetilde{S}_{NT}= \frac{1}{\sigma_x\sqrt{L_N T}} \sum_{\ell=-\infty}^T\mathbf{1}_N^\top\pmb{\mathcal{B}}_{T\ell} \pmb{\varepsilon}_{\ell}\eqqcolon \sum_{\ell=-\infty}^T\mathbf{1}_N^\top\overline{\pmb{\mathcal{B}}}_{\ell} \pmb{\varepsilon}_{\ell},
\]
where  $\frac{1}{\sigma_x\sqrt{L_NT}} \pmb{\mathcal{B}}_{T\ell} \eqqcolon \overline{\pmb{\mathcal{B}}}_{\ell}$, and we have suppressed $N$ and $T$ in $\overline{\pmb{\mathcal{B}}}_{\ell}$ for notational simplicity.

\medskip

To proceed, denote by  $\psi_{NT}(u)$ the characteristic function of $\widetilde{S}_{NT}$. Thus,

\begin{eqnarray}\label{def_psiNT}
\psi_{NT}(u) &=& E\left[\exp\left(\mathsf{i}u  \sum_{\ell=-\infty}^T\mathbf{1}_N^\top\overline{\pmb{\mathcal{B}}}_{\ell} \pmb{\varepsilon}_{\ell}\right)\right]\notag \\
&=&\prod_{\ell=-\infty}^T \prod_{j=1}^NE [\exp (\mathsf{i}u  \mathbf{1}_N^\top\overline{\mathbf{b}}_{\ell ,j}^\dag \varepsilon_{j\ell} ) ]\notag\\
&=& \prod_{\ell=-\infty}^T \prod_{j=1}^N\psi(\widetilde{b}_{\ell j} u) ,
\end{eqnarray}
where $\overline{\mathbf{b}}_{\ell ,j}^\dag$ stands for the $j^{th}$ column of $\overline{\pmb{\mathcal{B}}}_{\ell}$, $\widetilde{b}_{\ell j}\coloneqq \mathbf{1}_N^\top \overline{\mathbf{b}}_{\ell ,j}^\dag$, and the second equality follows from $\{\varepsilon_{it}\}$  being i.i.d. over both dimensions.  

Using \eqref{def_psiNT}, we are able to calculate the $r^{th}$ cumulant $\beta_r$ of $\widetilde{S}_{NT}$ for $r\ge 1$. Obviously, we have $\beta_1 =0$ and $\beta_2 =1$. For $r\ge 3$, write

\begin{eqnarray}\label{def_betar}
\beta_r &=& (-\mathsf{i})^r \frac{\mathrm{d}^r}{\mathrm{d}u^r}\log  \psi_{NT}(u) |_{u=0}\notag \\
&=&\sum_{\ell=-\infty}^T\sum_{j=1}^N (-\mathsf{i})^r\frac{\mathrm{d}^r}{\mathrm{d}u^r}\log\psi(\widetilde{b}_{\ell j}u)|_{u=0} \notag \\
&=&\sum_{\ell=-\infty}^T\sum_{j=1}^N\widetilde{b}_{\ell j}^r (-\mathsf{i})^r\frac{\mathrm{d}^r}{\mathrm{d}u^r}\log\psi(u)|_{u=0}\notag \\
&=&\kappa_r\left(\sum_{\ell=1}^T\sum_{j=1}^N\widetilde{b}_{\ell j}^r  +\sum_{\ell=-\infty}^0\sum_{j=1}^N\widetilde{b}_{\ell j}^r \right),
\end{eqnarray}
where the second equality follows from \eqref{def_psiNT}, and the third equality follows from \eqref{prop.cum}. For the second term on the right hand side of \eqref{def_betar}, we note that 

\begin{eqnarray}\label{def_betar2}
\left|\sum_{\ell=-\infty}^0\sum_{j=1}^N\widetilde{b}_{\ell j}^r\right|&\le &\frac{1}{(\sigma_x^2 L_NT)^{r/2}}\sum_{\ell=-\infty}^0\sum_{j=1}^N|\mathbf{1}_N^\top \mathbf{b}_{T\ell, j}^\dag|^r\notag \\
&\le &\frac{1}{(\sigma_x^2 L_N T)^{r/2}}\left(\sum_{\ell=-\infty}^0\sum_{j=1}^N|\mathbf{1}_N^\top \mathbf{b}_{T\ell, j}^\dag|^2\right)^{r/2}\notag\\
&= &\frac{1}{(\sigma_x^2 T)^{r/2}}\left(\frac{1}{L_N}\sum_{\ell=-\infty}^0  \mathbf{1}_N^\top \pmb{\mathcal{B}}_{T\ell}\pmb{\mathcal{B}}_{T\ell}^\top \mathbf{1}_N\right)^{r/2}\notag \\
&=&O\left(\frac{1}{T^{r/2}}\right),
\end{eqnarray}
where $\mathbf{b}_{T\ell, j}^\dag$ stands for the $j^{th}$ column of $\pmb{\mathcal{B}}_{T\ell}$, the second inequality follows from the fact that for a vector $\mathbf{x}$, $|\mathbf{x}|_{p_1}\le |\mathbf{x} |_{p_2}$ for any $p_1> p_2 \ge 1$, and the last equality follows from the proof of Lemma \ref{LM.A5}.2. Thus, \eqref{def_betar} and \eqref{def_betar2} together infer that

\begin{eqnarray}\label{def_betar3}
\beta_r =\kappa_r \sum_{\ell=1}^T\sum_{j=1}^N\widetilde{b}_{\ell j}^r +O\left(\frac{1}{T^{r/2}}\right).
\end{eqnarray}

Recall the definitions of \eqref{def_bdd}, and write further that

\begin{eqnarray}\label{def_betar4}
&&\frac{1}{(\sigma_x^2 L_NT)^{r/2}}\sum_{\ell=1}^T\sum_{j=1}^N|\mathbf{1}_N^\top \mathbf{b}_{T\ell, j}^\dag|^r \notag \\
&\le& \frac{2^{r-1}}{(\sigma_x^2 L_NT)^{r/2}}\sum_{\ell=1}^T\sum_{j=1}^N[|\mathbf{1}_N^\top \mathbf{b}_j^\dag|^r+|\mathbf{1}_N^\top \widetilde{\mathbf{b}}_{T-\ell, j}^\dag|^r]\notag\\
&=& \frac{2^{r-1}}{(\sigma_x^2 L_N)^{r/2}T^{r/2-1}} \sum_{j=1}^N |\mathbf{1}_N^\top \mathbf{b}_j^\dag|^r+O\left(\frac{1}{T^{r/2}}\right),
\end{eqnarray}
where the first inequality follows from $(a+b)^r\le 2^{r-1}(a^r+b^r)$ for $a,b\ge 0$ and $r>1$, and the last step follows from a development similar to \eqref{def_betar2}. Thus, by \eqref{def_betar3} and \eqref{def_betar4}, we can obtain that 

\begin{eqnarray}\label{def_betar5}
|\beta_r|& \le &\frac{2^{r-1}}{(\sigma_x^2 L_N)^{r/2}T^{r/2-1}} \sum_{j=1}^N(\mathbf{1}_N^\top \mathbf{b}_j^\dag )^r +O\left(\frac{1}{T^{r/2}}\right)\notag\\
&\le  &\frac{2^{r-1}}{(\sigma_x^2 L_N)^{r/2}T^{r/2-1}}  \cdot N\|\mathbf{B} \|_1^r +O\left(\frac{1}{T^{r/2}}\right)\notag\\
&=& O\left(\Delta_{NT}(r) \vee \frac{1}{T^{r/2}} \right),
\end{eqnarray}
where $\Delta_{NT}(r)\coloneqq NT\left(\frac{\|\mathbf{B} \|_1}{\sqrt{L_NT}}\right)^r$. 

Note that the right hand side of \eqref{def_betar5} reduces to $O\left(\frac{1}{ (NT)^{r/2-1}} \vee \frac{1}{T^{r/2}} \right)$ for both Examples \ref{EX1} and \ref{EX2} by simple calculation. In addition, using \eqref{def_betar5}, we write for $r\ge 3$

\begin{eqnarray}\label{def_betar6}
|\beta_r|  &=&  O\left(\Delta_{NT}(r) \vee \frac{1}{T^{r/2}} \right) =O(1)  \frac{1}{\left(\frac{\sqrt{L_NT}}{(NT)^{1/r} \|\mathbf{B}\|_1} \wedge T^{1/2}\right)^r}\notag \\
&\le &O (1)  \frac{1}{\left(\frac{\sqrt{L_NT}}{(NT)^{1/3} \|\mathbf{B}\|_1} \wedge T^{1/2}\right)^r} = O\left( \frac{1}{\widetilde{\Delta}_{NT}^r}\right),
\end{eqnarray}
where $\widetilde{\Delta}_{NT} \coloneqq \frac{\sqrt{L_N}T^{1/6}}{N^{1/3} \|\mathbf{B}\|_1} \wedge T^{1/2}$.

By applying \eqref{def.chiu} to $\psi_{NT}(u)$ and the characteristic function of $N(0,1)$, we have

\begin{eqnarray}\label{def_psiNT2}
\psi_{NT}(u) &=&\exp\left\{ \sum_{r=3}^\infty \beta_r\cdot\frac{(\mathsf{i} u)^r}{r!}\right\} \cdot\widetilde{\phi}(u)\notag \\
&=&\left\{1 +\sum_{n=1}^\infty\frac{1}{n!}\left(\sum_{r=3}^\infty \beta_r\cdot\frac{(\mathsf{i} u)^r}{r!}\right)^n \right\}\cdot\widetilde{\phi}(u),
\end{eqnarray}
where $\beta_r$ is bounded by \eqref{def_betar5}, and the second equality follows from \eqref{def.expu}. Let $f_{NT}(w)$ be the density function of $\widetilde{S}_{NT}$, and Fourier inversion of \eqref{def_psiNT2} leads to the following expansion:

\begin{eqnarray*}
f_{NT}(w) &=&\phi(w) +\frac{1}{2\pi}\int_{\mathbb{R}} \exp(-\mathsf{i}uw)\cdot \beta_3\cdot \frac{(\mathsf{i}u)^3}{3!}\cdot\widetilde{\phi}(u)\mathrm{d}u \notag \\
&&+\frac{1}{2\pi}\int_{\mathbb{R}} \exp(-\mathsf{i}uw)\cdot\widetilde{\psi}_{NT}(u)\cdot\widetilde{\phi}(u)\mathrm{d}u\notag \\
&=&\phi(w) + \frac{\beta_3}{6} H_3(w)\phi(w)\notag \\
&&+\frac{1}{2\pi}\int_{\mathbb{R}} \exp(-\mathsf{i}uw)\cdot\widetilde{\psi}_{NT}(u)\cdot\widetilde{\phi}(u)\mathrm{d}u,
\end{eqnarray*}
where $\widetilde{\psi}_{NT}(u)\coloneqq \sum_{r=4}^\infty \beta_r\cdot\frac{(\mathsf{i} u)^r}{r!} + \sum_{n=2}^\infty\frac{1}{n!}\left(\sum_{r=3}^\infty \beta_r\cdot\frac{(\mathsf{i} u)^r}{r!}\right)^n$, and the second equality follows from Lemma \ref{LM.A4}. 

Next, we bound $\frac{1}{2\pi}\int_{\mathbb{R}} \exp(-\mathsf{i}uw)\cdot \widetilde{\psi}_{NT}(u)\cdot\widetilde{\phi}(u)\mathrm{d}u$. First, note  

\begin{eqnarray}\label{def_psiNT3}
\left|\sum_{r=4}^\infty \beta_r\cdot\frac{(\mathsf{i} u)^r}{r!} \right|\widetilde{\phi}(u) &\le & |\beta_4| \widetilde{\phi}(u) \sum_{r=4}^\infty \frac{|u|^r}{r!} \notag \\
&\le & |\beta_4| \widetilde{\phi}(u)\exp(|u|)\notag \\
&=&|\beta_4| \exp(-u^2/2+|u|),
\end{eqnarray}
where the second inequality follows from \eqref{def.expu}. Second, in connection with \eqref{def_betar6}, assuming that $|u|\le \widetilde{\Delta}_{NT}$ and using Taylor theorem twice as in \eqref{def.expu2} we write

\begin{eqnarray}\label{def_psiNT4}
&&\left|\sum_{n=2}^\infty\frac{1}{n!}\left(\sum_{r=3}^\infty \beta_r\cdot\frac{(\mathsf{i} u)^r}{r!}\right)^n \right|\widetilde{\phi}(u)\notag \\
&\le &O(1)\sum_{n=2}^\infty\frac{1}{n!} \left| \sum_{r=3}^\infty \cdot\frac{(\mathsf{i} u/\widetilde{\Delta}_{NT} )^r}{r! }\right|^n \widetilde{\phi}(u)\notag\\
&\le & O(1)\frac{1}{\widetilde{\Delta}_{NT}^6}\widetilde{\phi}(u) \widetilde{u}^6=O(1)\frac{1}{\widetilde{\Delta}_{NT}^6},
\end{eqnarray}
where $|\widetilde{u}|$ is in between 0 and $\widetilde{\Delta}_{NT}$, and the last equality follows from $\widetilde{\phi}(u) \widetilde{u}^6$ being uniformly bounded. Thus, by \eqref{def_psiNT3} and \eqref{def_psiNT4}, we can write

\begin{eqnarray}\label{def_psiNT4_2}
&&\sup_{|u|\le \widetilde{\Delta}_{NT}}\left|\frac{1}{2\pi}\int_{\mathbb{R}} \exp(-\mathsf{i}uw)\cdot \widetilde{\psi}_{NT}(u)\cdot\widetilde{\phi}(u)\mathrm{d}u\right|\notag \\
&=&O\left( \Delta_{NT}(4) \vee \frac{1}{T^2} + \frac{1}{\widetilde{\Delta}_{NT}^6}\right) = O\left( \Delta_{NT}(4) \vee \frac{1}{T^2}  \right),
\end{eqnarray}
where the second equality follows from Assumption \ref{AS1}.2.

Thus, according to \eqref{def_psiNT2} and \eqref{def_psiNT4_2}, we write further 

\begin{eqnarray}\label{def_psiNT5}
&&\sup_{w\in\mathbb{R}}\left|f_{NT}(w)-\phi(w)-\frac{\beta_3}{6} H_3(w)\phi(w) \right|\notag \\
&\le &\frac{1}{2\pi}\int_{\mathbb{R}}\left|\psi_{NT}(u) -\widetilde{\phi}(u) - \frac{\beta_3}{6} (\mathsf{i}u)^3 \widetilde{\phi}(u)\right|\mathrm{d}u\notag\\
&= &\frac{1}{2\pi}\int_{|u|\ge \widetilde{\Delta}_{NT}} |\widetilde{\psi}_{NT}(u) |\widetilde{\phi}(u)\mathrm{d}u +\frac{1}{2\pi}\int_{|u|< \widetilde{\Delta}_{NT}} |\widetilde{\psi}_{NT}(u) |\widetilde{\phi}(u)\mathrm{d}u \notag \\
&=&O\left( \Delta_{NT}(4) \vee \frac{1}{T^2}\right),
\end{eqnarray}
where, in the last step, $\int_{|u|\ge \widetilde{\Delta}_{NT}} |\widetilde{\psi}_{NT}(u) |\widetilde{\phi}(u)\mathrm{d}u$ can be arbitrarily small by standard operation, and $\int_{|u|< \widetilde{\Delta}_{NT}} |\widetilde{\psi}_{NT}(u) |\widetilde{\phi}(u)\mathrm{d}u $ is bounded by \eqref{def_psiNT4_2}.

To study the CDF, we let $G(w) =\Phi(w)+\frac{\beta_3}{6}(1-w^2)\phi(w)$, where simple algebra shows that

\begin{eqnarray*}
\frac{\mathrm{d}( (1-w^2)\phi(w) )}{\mathrm{d}w} =  H_3(w)\phi(w) .
\end{eqnarray*}
Thus, $G(w)$ has a characteristic function $\xi(w) =(1+\frac{\beta_3}{6}(\mathsf{i}w)^3)\widetilde{\phi}(w)$. We then invoke Esseen's smoothing Lemma. First, we let $a$ of Lemma \ref{LM.A2} be sufficiently large, so the second term on the right hand side of Lemma \ref{LM.A2} becomes negligible. Then, similar to \eqref{def_psiNT5}, we study the first term and can obtain that

\begin{eqnarray*}
|F_{NT}(w)-G(w)|_\infty=O\left( \Delta_{NT}(4) \vee \frac{1}{T^2}\right).
\end{eqnarray*}
The proof is now completed. 
\end{proof}

\medskip

\begin{proof}[Proof of Corollary \ref{COL.1}]
\item 

According to \eqref{def_betar3} and Lemma \ref{LM.A5}.1, we write

\begin{eqnarray*}
\beta_3&=&\kappa_3 \sum_{\ell=1}^T\sum_{j=1}^N\widetilde{b}_{\ell j}^3+O\left(\frac{1}{T^{3/2}}\right)\notag\\
&=&\frac{\kappa_3}{(L_NT\sigma_x)^{3/2}} \sum_{\ell=1}^T\sum_{j=1}^N (\mathbf{1}_N^\top \mathbf{b}_j^\dag )^3-\frac{3\kappa_3}{(L_NT\sigma_x)^{3/2}} \sum_{\ell=1}^T\sum_{j=1}^N(\mathbf{1}_N^\top \mathbf{b}_j^\dag )^2(\mathbf{1}_N^\top \widetilde{\mathbf{b}}_{T-\ell,j}^\dag ) \notag \\
&&+\frac{3\kappa_3}{(L_NT\sigma_x)^{3/2}} \sum_{\ell=1}^T\sum_{j=1}^N (\mathbf{1}_N^\top \mathbf{b}_j^\dag) (\mathbf{1}_N^\top \widetilde{\mathbf{b}}_{T-\ell,j}^\dag )^2-\frac{\kappa_3}{(L_NT\sigma_x)^{3/2}} \sum_{\ell=1}^T\sum_{j=1}^N( \mathbf{1}_N^\top \widetilde{\mathbf{b}}_{T-\ell,j}^\dag )^3 \notag \\
&&+O\left(\frac{1}{T^{3/2}}\right)\notag \\
&\eqqcolon &\frac{\kappa_3}{\sigma_x^{3/2}}( \beta_{3,1}-3\beta_{3,2}+3\beta_{3,3}-\beta_{3,4})+O\left(\frac{1}{T^{3/2}}\right),
\end{eqnarray*}
where $\mathbf{b}_j^\dag$ and $\widetilde{\mathbf{b}}_{T-\ell,j}^\dag$ are defined in \eqref{def_bdd}, and the definitions of $\beta_{3,j}$ for $j\in[4]$ are obvious.

For $\beta_{3,2}$, we write

\begin{eqnarray*}
|\beta_{3,2}|&\le&\frac{\|\mathbf{B}\|_1^2}{(L_NT)^{3/2}} \sum_{j=1}^N \sum_{\ell=1}^T |\mathbf{1}_N^\top \widetilde{\mathbf{b}}_{T-\ell,j}^\dag |\notag\\
&\le&\frac{\|\mathbf{B}\|_1^2}{(L_NT)^{3/2}}  \sqrt{NT}\left( \sum_{j=1}^N\sum_{\ell=1}^T |\mathbf{1}_N^\top \widetilde{\mathbf{b}}_{T-\ell,j}^\dag|^2\right)^{1/2}\notag\\
&=&\frac{\sqrt{NT}\|\mathbf{B}\|_1^2}{(L_NT)^{3/2}}   \left(\sum_{\ell=1}^T\mathbf{1}_N^\top \widetilde{\mathbf{B}}_{T-\ell}\widetilde{\mathbf{B}}_{T-\ell}^\top \mathbf{1}_N\right)^{1/2}\notag\\
&=&\frac{\sqrt{N}\|\mathbf{B}\|_1^2}{L_NT  } \left(\frac{N}{L_N}\sum_{\ell=1}^T \|\widetilde{\mathbf{B}}_{T-\ell}\|_2^2\right)^{1/2}\notag \\
&=&O_P\left(\frac{\sqrt{N}\|\mathbf{B}\|_1^2}{L_NT}\right),
\end{eqnarray*}
where the last line follows from the proof of Lemma \ref{LM.A5}.2.

For $\beta_{3,3}$, we write

\begin{eqnarray*}
|\beta_{3,3}|&\le&\frac{\|\mathbf{B}\|_1}{(L_NT)^{3/2}} \sum_{j=1}^N  \sum_{\ell=1}^T |\mathbf{1}_N^\top \widetilde{\mathbf{b}}_{T-\ell,j}^\dag |^2\notag\\
&=&\frac{\|\mathbf{B}\|_1}{(L_NT)^{3/2}} \sum_{\ell=1}^T\mathbf{1}_N^\top \widetilde{\mathbf{B}}_{T-\ell}\widetilde{\mathbf{B}}_{T-\ell}^\top \mathbf{1}_N\notag\\
&=&O\left(\frac{\|\mathbf{B}\|_1}{L_N^{1/2}T^{3/2}}\right) ,
\end{eqnarray*}
where the second equality follows from the proof of Lemma \ref{LM.A5}.2.

For $\beta_{3,4}$, we write

\begin{eqnarray*}
|\beta_{3,4}|&\le&\frac{1}{(L_NT)^{3/2}} \left(\sum_{j=1}^N  \sum_{\ell=1}^T |\mathbf{1}_N^\top \widetilde{\mathbf{b}}_{T-\ell,j}^\dag |^2\right)^{3/2}\notag\\ 
&=&\frac{1}{(L_NT)^{3/2}} \left( \sum_{\ell=1}^T\mathbf{1}_N^\top \widetilde{\mathbf{B}}_{T-\ell}\widetilde{\mathbf{B}}_{T-\ell}^\top \mathbf{1}_N\right)^{3/2}\notag \\
&=&O\left(\frac{1}{L_N T^{3/2}}\right),
\end{eqnarray*}
where the first inequality follows a development similar to \eqref{def_betar2}.

By Assumption \ref{AS1}.2, it is easy to know that among $\beta_{3,j}$ with $j=2,3,4$, the term $|\beta_{3,2}|$ offers the lowest rate.  In connection with the decomposition of $\beta_3$, the proof is now completed.
\end{proof}

\medskip

\begin{proof}[Proof of Theorem \ref{THM.2}]
\item  

(1). Let $\mathcal{S}_{NT}^*=  \sum_{t=1}^T\mathbf{1}_N^\top \mathbf{x}_t\zeta_t$ and $\sigma_{\mathcal{S}}^{\ast}=\text{Var}^\ast(\mathcal{S}_{NT}^*)^{\frac{1}{2}}$. We can observe that  $\widetilde{S}_{NT}^*/\widetilde{\sigma}^\ast=\mathcal{S}_{NT}^*/\sigma_{\mathcal{S}}^\ast$. We aim to establish the Edgeworth expansion for $\widetilde{\mathcal{S}}_{NT}^*\coloneqq \mathcal{S}_{NT}^*/\sigma_{\mathcal{S}}^\ast$ through investigating its characteristic function: 

\begin{equation*}
\phi^\ast_{NT}(u)=E^\ast [\exp ( \mathsf{i} u\widetilde{\mathcal{S}}_{NT}^*  )].
\end{equation*}
Specifically, we adopt a proof strategy similar to those used by \cite{tik1981}. Let $\mathcal{T}_m \coloneqq \lfloor \frac{T}{m} \rfloor$, where we decompose $\widetilde{\mathcal{S}}_{NT}^*$ into a summation of 1-dependent series, conditional on the observations:

\begin{equation*}
\widetilde{\mathcal{S}}_{NT}^*=\sum_{s=1}^{\mathcal{T}_m} A_s+ A_{\mathcal{T}_m+1},
\end{equation*}
where $A_s \coloneqq \frac{1}{\sigma_{\mathcal{S}}^\ast}\sum_{t=(s-1)m+1}^{sm} \mathbf{1}_N^\top \mathbf{x}_t\zeta_t$ and $A_{\mathcal{T}_m+1} \coloneqq \frac{1}{\sigma_{\mathcal{S}}^\ast}\sum_{t=\mathcal{T}_m+1}^{T} \mathbf{1}_N^\top \mathbf{x}_t\zeta_t$. Additionally, define 

\begin{eqnarray*}
A_{s,l}^c=\sum_{|t-s|>l} A_t \quad \text{and} \quad A_{s,0}^c =  \widetilde{\mathcal{S}}_{NT}^*,
\end{eqnarray*}
for $l=1,\ldots, 4$ and $s=1,\ldots, \mathcal{T}_m+1$. By the nature of 1-dependent random variables, we know that $A_{s,l}^c$ and $A_{s}$ are independent conditional on the observations for $l=1,\ldots, 4$.  This decomposition enables us to apply established techniques to handle such series.

We are now ready to make the following expansion of $\exp ( \mathsf{i} u\widetilde{\mathcal{S}}_{NT}^*  )$ for each given $s$:

\begin{eqnarray}\label{Edex7}
\exp ( \mathsf{i} u\widetilde{\mathcal{S}}_{NT}^*  )&=& \exp ( \mathsf{i} u A_{s,1}^c )+ \{\exp ( \mathsf{i} u (\widetilde{\mathcal{S}}_{NT}^* - A_{s,1}^c ) )-1 \}\exp ( \mathsf{i} u A_{s,1}^c ),
\end{eqnarray}
and this process can be iterated for the second term involving $\exp ( \mathsf{i} u A_{s,1}^c )$ and the subsequent terms. Finally, we obtain

\begin{eqnarray}\label{2Edx1}
\exp ( \mathsf{i} u\widetilde{\mathcal{S}}_{NT}^*  )&=&\exp ( \mathsf{i} u A_{s,1}^c )+ \{\exp ( \mathsf{i} u (\widetilde{\mathcal{S}}_{NT}^* - A_{s,1}^c ) )-1 \}\exp ( \mathsf{i} u A_{s,2}^c )
\notag\\
&&+\sum_{l=3}^{4}\prod_{k=1}^{l-1} \{\exp ( \mathsf{i} u (A_{s,k-1}^c - A_{s,k}^c ) )-1 \}\exp ( \mathsf{i} u A_{s,l}^c )\notag\\
&&+\prod_{k=1}^{4} \{\exp ( \mathsf{i} u (A_{s,k-1}^c - A_{s,k}^c ) )-1 \}\exp ( \mathsf{i} u A_{s,4}^c ).
\end{eqnarray} 
Accordingly, we can write the first-order derivative of $\phi^\ast_{NT}(u)$ as

\begin{eqnarray}\label{Edex11}
\frac{\mathrm{d}\phi^\ast_{NT}(u)}{\mathrm{d}u}&=&\mathsf{i}E^\ast [\widetilde{\mathcal{S}}_{NT}^* \exp ( \mathsf{i} u \widetilde{\mathcal{S}}_{NT}^*  ) ]
\notag\\
&=& \mathsf{i}\sum_{s=1}^{\mathcal{T}_m+1} E^\ast [A_s\exp ( \mathsf{i} u  A_{s,1}^c ) ]\notag \\
&&+\mathsf{i}\sum_{s=1}^{\mathcal{T}_m+1} E^\ast [A_s  \{\exp ( \mathsf{i} u (\widetilde{\mathcal{S}}_{NT}^* - A_{s,1}^c ) )-1 \}\exp ( \mathsf{i} u A_{s,2}^c ) ] \notag\\
&&+\mathsf{i}\sum_{s=1}^{\mathcal{T}_m+1}\sum_{l=3}^{4} E^\ast \Big[A_s\prod_{k=1}^{l-1} \{\exp ( \mathsf{i} u (A_{s,k-1}^c - A_{s,k}^c ) )-1 \}\exp ( \mathsf{i} u A_{s,l}^c ) \Big] \notag\\
&&+\mathsf{i}\sum_{s=1}^{\mathcal{T}_m+1} E^\ast\Big[A_s \prod_{k=1}^{4} \{\exp ( \mathsf{i} u (A_{s,k-1}^c - A_{s,k}^c ) )-1 \}\exp ( \mathsf{i} u A_{s,4}^c )\Big]\notag\\
&\eqqcolon &\mathsf{J}_{NT,1}(u)+\cdots+\mathsf{J}_{NT,4}(u).
\end{eqnarray}

We then investigate the four terms on the right hand side one by one. 

For the first term, it is clear to see that 

\begin{eqnarray}\label{Edex12}
\mathsf{J}_{NT,1}(u) =\mathsf{i}\sum_{s=1}^{\mathcal{T}_m+1} E^\ast[A_s]E^\ast [\exp ( \mathsf{i} u  A_{s,1}^c)]=0.
\end{eqnarray}
To investigate $\mathsf{J}_{NT,2}(u)$, we invoke Lemma \ref{LM.A3} by letting $r=2$, 

\begin{eqnarray}\label{Edex1}
\mathsf{J}_{NT,2}(u)&=&\mathsf{i}\sum_{s=1}^{\mathcal{T}_m+1} E^\ast\Big[A_s\Big\{\mathsf{i}u (\widetilde{\mathcal{S}}_{NT}^* - A_{s,1}^c )-\frac{1}{2}u^2 (\widetilde{\mathcal{S}}_{NT}^* - A_{s,1}^c )^2+R^\ast_{s,1}(u)\Big\}\exp ( \mathsf{i} u A_{s,2}^c )\Big]
\notag\\
&=&-u\sum_{s=1}^{\mathcal{T}_m+1} E^\ast [A_s (\widetilde{\mathcal{S}}_{NT}^* - A_{s,1}^c )\exp ( \mathsf{i} u A_{s,2}^c ) ]\notag \\
&&-\frac{1}{2}\mathsf{i}u^2\sum_{s=1}^{\mathcal{T}_m+1} E^\ast [A_s (\widetilde{\mathcal{S}}_{NT}^* - A_{s,1}^c )^2\exp ( \mathsf{i} u A_{s,2}^c ) ]
\notag\\
&&+\mathsf{i}\sum_{s=1}^{\mathcal{T}_m+1}E^\ast [A_s R^\ast_{s,1} (u)\exp ( \mathsf{i} u A_{s,2}^c ) ]
\notag\\
&\eqqcolon &\mathsf{J}_{NT,2,1}(u)+\mathsf{J}_{NT,2,2}(u)+\mathsf{J}_{NT,2,3}(u),
\end{eqnarray}
where $R^\ast_{s,1}(u)=\sum_{l=3}^{\infty}\frac{\mathsf{(iu)}^l}{l!} (\widetilde{\mathcal{S}}_{NT}^* - A_{s,1}^c )^l$ satisfies $|R^\ast_{s,1}(u)|\leq \frac{\mathsf{u}^3}{3!} |\widetilde{\mathcal{S}}_{NT}^* - A_{s,1}^c |^3$.

Since $A_s$ and $\widetilde{\mathcal{S}}_{NT}^* - A_{s,1}^c$ are conditionally independent with $A_{s,2}^c$, we can further write

\begin{eqnarray*}
\mathsf{J}_{NT,2,1}(u)&=&-u\sum_{s=1}^{\mathcal{T}_m+1} E^\ast [A_s (\widetilde{\mathcal{S}}_{NT}^* - A_{s,1}^c ) ]E^\ast [\exp ( \mathsf{i} u A_{s,2}^c ) ] \notag\\
&=&u\sum_{s=1}^{\mathcal{T}_m+1} E^\ast [A_s (\widetilde{\mathcal{S}}_{NT}^* - A_{s,1}^c ) ]E^\ast [\exp ( \mathsf{i} u \widetilde{\mathcal{S}}_{NT}^* )-\exp ( \mathsf{i} u A_{s,2}^c ) ]\notag\\
&&-u\sum_{s=1}^{\mathcal{T}_m+1} E^\ast [A_s (\widetilde{\mathcal{S}}_{NT}^* - A_{s,1}^c ) ] \phi^\ast_{NT}(u),
\end{eqnarray*} 
where the second equality follows from the fact that $\phi^\ast_{NT}(u)=E^\ast [\exp ( \mathsf{i} u \widetilde{\mathcal{S}}_{NT}^*)]$.
  
For the first term on the right hand side of $\mathsf{J}_{NT,2,1}(u)$, we adopt a similar expansion procedure for $\exp ( \mathsf{i} u \widetilde{\mathcal{S}}_{NT}^* )$ as in \eqref{Edex7}, and obtain

\begin{eqnarray}\label{Edex8}
&&\exp ( \mathsf{i} u \widetilde{\mathcal{S}}_{NT}^* )-\exp ( \mathsf{i} u A_{s,2}^c ) \notag \\
&=& \{\exp ( \mathsf{i} u (\widetilde{\mathcal{S}}_{NT}^* - A_{s,2}^c ) )-1 \}\exp ( \mathsf{i} u A_{s,2}^c )\notag\\
&=& \{\exp ( \mathsf{i} u (\widetilde{\mathcal{S}}_{NT}^* - A_{s,2}^c ) )-1 \}\exp ( \mathsf{i} u A_{s,3}^c )\notag\\
&&+ \{\exp ( \mathsf{i} u (\widetilde{\mathcal{S}}_{NT}^* - A_{s,2}^c ) )-1 \} \{\exp (\mathsf{i} u (A_{s,2}^c - A_{s,3}^c ) )-1 \}\exp ( \mathsf{i} u A_{s,3}^c ). 
\end{eqnarray}
For the first term in \eqref{Edex8}, 

\begin{eqnarray}\label{Edex2_1}
&&|E^\ast [ \{\exp (\mathsf{i} u (\widetilde{\mathcal{S}}_{NT}^* - A_{s,2}^c ) )-1 \}\exp ( \mathsf{i} u A_{s,3}^c ) ] | \notag \\
&\leq& |E^\ast [ \{\exp\big(\mathsf{i} u\big(\widetilde{\mathcal{S}}_{NT}^* - A_{s,2}^c ) )-1 \} |\notag\\
&\leq&  u |E^\ast [\widetilde{\mathcal{S}}_{NT}^* - A_{s,2}^c ] |+\frac{u^2}{2}E^\ast [ (\widetilde{\mathcal{S}}_{NT}^* - A_{s,2}^c )^2 ]\notag\\
&=&\frac{u^2}{2}E^\ast [ (\widetilde{\mathcal{S}}_{NT}^* - A_{s,2}^c )^2 ],
\end{eqnarray}
for any given $u$, where the first inequality holds by the conditional independence between $\widetilde{\mathcal{S}}_{NT}^* - A_{s,2}^c$ and $A_{s,3}^c$ and $|E[\exp( \mathsf{i} u A_{s,3}^c)]|\leq1$,  and  the second inequality is an application of Lemma \ref{LM.A3} with $r=1$.

Additionally, we can use the inequality $|e^{iu_1}-e^{iu_2}|\leq |u_1-u_2|$  and then Cauchy-Schwarz inequality sequentially to obtain

\begin{eqnarray}\label{Edex2_2}
&& |E^\ast [ \{\exp ( \mathsf{i} u(\widetilde{\mathcal{S}}_{NT}^* - A_{s,2}^c))-1\}\{\exp(\mathsf{i} u\big(A_{s,2}^c - A_{s,3}^c))-1\}\exp ( \mathsf{i} u A_{s,3}^c ) ] |\notag\\
&\leq&E^\ast [ |\exp ( \mathsf{i} u (\widetilde{\mathcal{S}}_{NT}^* - A_{s,2}^c ) )-1 |\cdot |\exp (\mathsf{i} u (A_{s,2}^c - A_{s,3}^c ) )-1 | ]\notag\\
&\leq&u^2 E^\ast [ (\widetilde{\mathcal{S}}_{NT}^* - A_{s,2}^c )^2 ]^{\frac{1}{2}}E^\ast [ (A_{s,2}^c - A_{s,3}^c )^2 ]^{\frac{1}{2}}.
\end{eqnarray} 
By \eqref{Edex8}, \eqref{Edex2_1}, and \eqref{Edex2_2}, 

\begin{eqnarray*}
&& |E^\ast [\exp ( \mathsf{i} u \widetilde{\mathcal{S}}_{NT}^* )-\exp ( \mathsf{i} u A_{s,2}^c ) ] |\notag \\
&\le &\frac{u^2}{2}E^\ast [ (\widetilde{\mathcal{S}}_{NT}^* - A_{s,2}^c )^2 ] +u^2 E^\ast [ (\widetilde{\mathcal{S}}_{NT}^* - A_{s,2}^c )^2 ]^{\frac{1}{2}}E^\ast [ (A_{s,2}^c - A_{s,3}^c )^2 ]^{\frac{1}{2}}.
\end{eqnarray*}
Substituting this result into the first term in $\mathsf{J}_{NT,2,1}(u)$, we obtain

\begin{eqnarray}\label{Edex2}
&&\sum_{s=1}^{\mathcal{T}_m+1} |E^\ast [A_s (\widetilde{\mathcal{S}}_{NT}^* - A_{s,1}^c ) ] |\cdot |E^\ast [\exp ( \mathsf{i} u \widetilde{\mathcal{S}}_{NT}^* )-\exp ( \mathsf{i} u A_{s,2}^c ) ] |\notag\\
&\leq&\frac{u^2}{2}\sum_{s=1}^{\mathcal{T}_m+1} |E^\ast [A_s (\widetilde{\mathcal{S}}_{NT}^* - A_{s,1}^c ) ] |\cdot E^\ast [ (\widetilde{\mathcal{S}}_{NT}^* - A_{s,2}^c )^2 ]\notag\\
&&+u^2\sum_{s=1}^{\mathcal{T}_m+1} |E^\ast [A_s (\widetilde{\mathcal{S}}_{NT}^* - A_{s,1}^c ) ] |\cdot E^\ast [ (\widetilde{\mathcal{S}}_{NT}^* - A_{s,2}^c )^2 ]^{\frac{1}{2}}E^\ast [ (A_{s,2}^c - A_{s,3}^c )^2 ]^{\frac{1}{2}}\notag\\
&=&O_P\left(\left(\frac{m}{T}\right)^2\mathcal{T}_m \right) =O_P\left(\frac{m}{T}\right) 
\end{eqnarray}
for any given $u$, where the first equality holds by Lemma \ref{LM.A7} and the second equality holds because $m\mathcal{T}_m=O(T)$.  

For the second term in  $\mathsf{J}_{NT,2,1}(u)$,

\begin{eqnarray}\label{Edex3}
\sum_{s=1}^{\mathcal{T}_m+1} E^\ast [A_s (\widetilde{\mathcal{S}}_{NT}^* - A_{s,1}^c ) ] &=& \sum_{s=1}^{\mathcal{T}_m+1} E^\ast [A_s (\widetilde{\mathcal{S}}_{NT}^* - A_{s,1}^c ) ]\notag\\
&=&\sum_{s=1}^{\mathcal{T}_m+1} E^\ast [A_s ]\widetilde{\mathcal{S}}_{NT}^* 
\notag\\
&=&E^\ast [\widetilde{\mathcal{S}}_{NT}^{*2} ]=1.
\end{eqnarray}
By \eqref{Edex2} and \eqref{Edex3}, we have 

\begin{eqnarray}\label{Edex4}
\mathsf{J}_{NT,2,1}(u)&=&-u\phi^\ast_{NT}(u)+O_P\left(\frac{m}{T}\right),
\end{eqnarray}
for any given $u\in(-\infty,\infty)$.

For $\mathsf{J}_{NT,2,2}(u)$, using arguments  similar to those in \eqref{Edex2}, we obtain 

\begin{eqnarray}\label{Edex5}
\mathsf{J}_{NT,2,2}(u)&=&-\frac{1}{2}\mathsf{i}u^2\sum_{s=1}^{\mathcal{T}_m+1} E^\ast [A_s (\widetilde{\mathcal{S}}_{NT}^* - A_{s,1}^c )^2 ]\phi^\ast_{NT}(u)+O_P\left(\frac{m}{T}\right).
\end{eqnarray}

For $\mathsf{J}_{NT,2,3}(u)$, since $|R^\ast_{s,1}(u)|\leq \frac{\mathsf{u}^3}{3!} |\widetilde{\mathcal{S}}_{NT}^* - A_{s,1}^c |^3$, it is straightforward to show 

\begin{eqnarray}\label{Edex6}
    \mathsf{J}_{NT,2,3}(u)=O_P\left(\frac{m}{T}\right).
\end{eqnarray}
Combining \eqref{Edex4}, \eqref{Edex5}, and \eqref{Edex6} gives 

\begin{eqnarray}\label{Edex13}
\mathsf{J}_{NT,2}(u)=-u\phi^\ast_{NT}(u)-\frac{1}{2}\mathsf{i}u^2\sum_{s=1}^{\mathcal{T}_m+1} E^\ast [A_s (\widetilde{\mathcal{S}}_{NT}^* - A_{s,1}^c )^2 ]\phi^\ast_{NT}(u)+O_P\left(\frac{m}{T}\right).
\end{eqnarray}

After finishing the investigation of $\mathsf{J}_{NT,2}(u)$, we proceed to study $\mathsf{J}_{NT,3}(u)$. In a similar way to the development for $\mathsf{J}_{NT,2,1}(u)$, we can easily show that

\begin{eqnarray*}
\mathsf{J}_{NT,3}(u)&=&\mathsf{i}\sum_{s=1}^{\mathcal{T}_m+1}\sum_{l=3}^{4} E^\ast\Big[A_s\prod_{k=1}^{l-1} \{\exp ( \mathsf{i} u (A_{s,k-1}^c - A_{s,k}^c ) )-1 \}\Big] \phi^\ast_{NT}(u)+o_P\left(\frac{m}{T}\right) \notag\\
&\eqqcolon& \mathsf{J}_{NT,3,1}(u)+ \mathsf{J}_{NT,3,2}(u),
\end{eqnarray*}
where $\mathsf{J}_{NT,3,1}(u)$ and $\mathsf{J}_{NT,3,2}(u)$ contain the terms with $l=3$ and $l=4$, respectively.

For $\mathsf{J}_{NT,3,1}(u)$, we invoke the Taylor expansion and the inequality in Lemma \ref{LM.A3} (with $r=1$) to write 

\begin{eqnarray*}
&&\mathsf{J}_{NT,3,1}(u)\notag \\
&=&\mathsf{i}\sum_{s=1}^{\mathcal{T}_m+1}E^\ast\Big[A_s\big(\mathsf{i}u\big(\widetilde{\mathcal{S}}_{NT}^* - A_{s,1}^c\big)+ R^\ast_{s,2} (u)\big)\big(\mathsf{i}u\big(A_{s,1}^c-  A_{s,2}^c\big)+ R^\ast_{s,3} (u)\big)\Big]\phi^\ast_{NT}(u),
\end{eqnarray*}
where $R^\ast_{s,2}(u)=\sum_{l=2}^{\infty}\frac{\mathsf{(iu)}^l}{l!} \big(\widetilde{\mathcal{S}}_{NT}^* - A_{s,1}^c\big)^l$ and $R^\ast_{s,3}(u)=\sum_{l=2}^{\infty}\frac{\mathsf{(iu)}^l}{l!} \big(A_{s,1}^c-  A_{s,2}^c\big)^l$ satisfy $|R^\ast_{s,2}(u)|\leq \frac{\mathsf{u}^2}{2}\big(\widetilde{\mathcal{S}}_{NT}^* - A_{s,1}^c\big)^2$ and $|R^\ast_{s,3}(u)|\leq \frac{\mathsf{u}^2}{2}\big(A_{s,1}^c-  A_{s,2}^c\big)^2$.

By Lemma \ref{LM.A7}, we have

\begin{eqnarray}\label{Edex10}
\mathsf{J}_{NT,3,1}(u)&=&-\mathsf{i}u^2\sum_{s=1}^{\mathcal{T}_m+1}E^\ast [A_s (\widetilde{\mathcal{S}}_{NT}^* - A_{s,1}^c ) (A_{s,1}^c-  A_{s,2}^c ) ]\phi^\ast_{NT}(u)+O_P\left(\frac{m}{T}\right).
\end{eqnarray}

Analogously to \eqref{Edex2_2}, the inequality $|e^{iu_1}-e^{iu_2}|\leq |u_1-u_2|$ can be applied here to show that

\begin{eqnarray*}
    \mathsf{J}_{NT,3,2}(u)&=&O_P\left(\frac{m}{T}\right).
\end{eqnarray*}
Together with \eqref{Edex10}, it yields

\begin{eqnarray}\label{Edex14}
    \mathsf{J}_{NT,3}(u)&=&-\mathsf{i}u^2\sum_{s=1}^{\mathcal{T}_m+1}E^\ast [A_s (\widetilde{\mathcal{S}}_{NT}^* - A_{s,1}^c ) (A_{s,1}^c-  A_{s,2}^c ) ]\phi^\ast_{NT}(u)+O_P\left(\frac{m}{T}\right).
\end{eqnarray}

For $\mathsf{J}_{NT,3}(u)$, we can use similar arguments to obtain

\begin{eqnarray}\label{Edex15}
    \mathsf{J}_{NT,4}(u)&=&o_P\left(\frac{m}{T}\right).
\end{eqnarray}

In summary of \eqref{Edex11}, \eqref{Edex12}, \eqref{Edex13}, \eqref{Edex14} and \eqref{Edex15},

\begin{eqnarray*}
\frac{\mathrm{d}\phi^\ast_{NT}(u)}{\mathrm{d}u}&=&-u\phi^\ast_{NT}(u)-\frac{1}{2}\mathsf{i}u^2\sum_{s=1}^{\mathcal{T}_m+1} E^\ast [A_s (\widetilde{\mathcal{S}}_{NT}^* - A_{s,1}^c )^2 ]\phi^\ast_{NT}(u)\notag\\
&&-\mathsf{i}u^2\sum_{s=1}^{\mathcal{T}_m+1}E^\ast [A_s (\widetilde{\mathcal{S}}_{NT}^* - A_{s,1}^c ) (A_{s,1}^c-  A_{s,2}^c ) ]\phi^\ast_{NT}(u)+O_P\left(\frac{m}{T}\right).
\end{eqnarray*}

It is worth noting that 

\begin{eqnarray*}
&&\sum_{s=1}^{\mathcal{T}_m+1} E^\ast [A_s ( (\widetilde{\mathcal{S}}_{NT}^* - A_{s,1}^c )^2+2\big(\widetilde{\mathcal{S}}_{NT}^* - A_{s,1}^c ) (A_{s,1}^c-  A_{s,2}^c ) ) ]\notag\\
&=&\sum_{s=1}^{\mathcal{T}_m+1} E^\ast [A_s (\widetilde{\mathcal{S}}_{NT}^* - A_{s,1}^c ) (\widetilde{\mathcal{S}}_{NT}^* + A_{s,1}^c-2A_{s,2}^c ) ]\notag\\
&=&\sum_{s=1}^{\mathcal{T}_m+1}E^\ast [A_s (\widetilde{\mathcal{S}}_{NT}^{*2} - A_{s,1}^{c2} )]-2\sum_{s=1}^{\mathcal{T}_m+1}E^\ast [A_s (\widetilde{\mathcal{S}}_{NT}^* - A_{s,1}^c )A_{s,2}^c ]\notag\\
&=&E^\ast [\widetilde{\mathcal{S}}_{NT}^{*3} ].
\end{eqnarray*}

Therefore, it follows that 
\begin{eqnarray*}
    \frac{\mathrm{d}\phi^\ast_{NT}(u)}{\mathrm{d}u}&=&-u\phi^\ast_{NT}(u)-\frac{1}{2}\mathsf{i}u^2 E^\ast [\widetilde{\mathcal{S}}_{NT}^{*3} ]\phi^\ast_{NT}(u)
    +O_P\left(\frac{m}{T}\right)
\end{eqnarray*}
for any given $u\in(-\infty, \infty)$. Integration of this relation gives

\begin{eqnarray*}
    \phi^\ast_{NT}(u)&=&\exp\big(-\frac{1}{2}u^2-\frac{1}{6}\mathsf{i}u^3 E^\ast [\widetilde{\mathcal{S}}_{NT}^{*3} ]\big) +O_P\left(\frac{m}{T}\right)\notag\\
    &=&\exp\big(-\frac{1}{2}u^2\big)\Big\{1-\frac{1}{6}\mathsf{i}u^3 E^\ast [\widetilde{\mathcal{S}}_{NT}^{*3} ]\Big\} +O_P\left(\frac{m}{T}\right),
\end{eqnarray*}
where the second equality holds due to Taylor expansion of $\exp\left(-\frac{1}{6}\mathsf{i}u^3 E^\ast\left[\widetilde{\mathcal{S}}_{NT}^{*3}\right]\right)$, and the higher-order terms in this expansion, beyond the second order, are all bounded by the probability order $O_P\left(\frac{m}{T}\right)$. Moreover, by Esseen smoothing inequality in Lemma \ref{LM.A2}, we can  finally establish the desired Edgeworth expansion  for the CDF of $\widetilde{\mathcal{S}}_{NT}^\ast$: 
\begin{eqnarray*}
        \sup_{u\in \mathbb{R}}\left|\text{\normalfont Pr}^\ast(\widetilde{\mathcal{S}}_{NT}^{*}\le u) -\Phi(u) -\frac{1-u^2}{6}E^\ast\big[\widetilde{\mathcal{S}}_{NT}^{*3}\big]\phi(x)\right|=O_P\left(\frac{m}{T}\right).
\end{eqnarray*}

This completes the proof of Theorem \ref{THM.2}.1.

\medskip

(2). Since $E^\ast\big[\widetilde{S}_{NT}^{*3}\big]=O_P\big(\sqrt{\frac{m}{T}}\big)$ and $\widetilde{\sigma}^\ast=1+o_P(1)$, it follows from  Theorem \ref{THM.2}.1 that
\begin{eqnarray}\label{Edex16}
    \sup_{u\in \mathbb{R}}\left|\text{\normalfont Pr}^*(\widetilde{S}_{NT}^* \le u) -\Phi(u)\right|=O_P\left(\sqrt{\frac{m}{T}}\right).
\end{eqnarray}

Combining the results in Theorem \ref{THM.1} and \eqref{Edex16}, we can readily obtain
\begin{eqnarray*}
    &&\sup_{u\in \mathbb{R}}\left|\text{\normalfont Pr}^*(\widetilde{S}_{NT}^* \le u) -\Pr(\widetilde{S}_{NT}\le u)\right|\notag \\
    &\leq&\sup_{u\in \mathbb{R}}\left|\text{\normalfont Pr}^*(\widetilde{S}_{NT}^* \le u) -\Phi(u)\right|
    +\sup_{u\in \mathbb{R}}\left|\Pr(\widetilde{S}_{NT}\le u)-\Phi(u)\right|
    \notag\\
    &=&O_P\left(\sqrt{\frac{m}{T}}\right).
\end{eqnarray*}

It leads to the desired result in  Theorem \ref{THM.2}.2.

\medskip

(3). To prove Theorem \ref{THM.2}.3, it suffices to show the following  results:
\begin{eqnarray}\label{msel_0}
    &&E[E^\ast[\widetilde{S}_{NT}^{*2}]]-E[\widetilde{S}_{NT}^{2} ]=-\frac{C_{q_a}}{ \sigma_x^2m^{q_a}}\Delta_{q_a}+o_P(m^{-q_a}),
    \notag\\
    &&Var( E^\ast[\widetilde{S}_{NT}^{*2}])=\frac{2m}{T}\int_{-1}^1 a^2(u)du+o\left(\frac{m}{T}\right).
\end{eqnarray}
where $\Delta_{q_a}=L_N^{-1}\sum_{s=-\infty}^{\infty}\sum_{\ell=0}^{\infty} |s|^{q_a}\mathbf{1}_N^\top \mathbf{B}_\ell  \mathbf{B}_{\ell+|s|}^\top \mathbf{1}_N $ and 
$C_{q_a}$ is defined in Theorem \ref{THM.2}.

It is clear to see that
\begin{eqnarray*}
    \sigma_x^2\big(E[E^\ast[\widetilde{S}_{NT}^{*2}]]-E[\widetilde{S}_{NT}^{2} ]\big)&=& 
    \frac{1}{L_NT}\sum_{t=1}^T\sum_{s=1}^T \Big\{ a\Big(\frac{t-s}{m}\Big)-1\Big\} \mathbf{1}_N^\top E\big[ \mathbf{x}_{t}\mathbf{x}_{s}^\top\big]\mathbf{1}_N
    \notag\\
    &=& \frac{1}{L_N}\sum_{s=-m}^{m} \Big\{ a\Big(\frac{s}{m}\Big)-1\Big\} \mathbf{1}_N^\top E\big[ \mathbf{x}_{1}\mathbf{x}_{1+s}^\top\big]\mathbf{1}_N
    \notag\\
    &&-\frac{2}{L_N}\sum_{s=1}^{m}\frac{s}{T} \Big\{ a\Big(\frac{s}{m}\Big)-1\Big\} \mathbf{1}_N^\top E\big[ \mathbf{x}_{1}\mathbf{x}_{1+s}^\top\big]\mathbf{1}_N
    \notag\\
    &&-\frac{2}{L_N}\sum_{s=m+1}^{T-1}\frac{T-s}{T} \mathbf{1}_N^\top E\big[ \mathbf{x}_{1}\mathbf{x}_{1+s}^\top\big]\mathbf{1}_N
    \notag\\
    &:=&\mathcal{I}_1+\mathcal{I}_2+\mathcal{I}_3,
\end{eqnarray*}
where the definitions of $\mathcal{I}_1$, $\mathcal{I}_2$ and  $\mathcal{I}_3$ are obvious. Here, the expression for $\mathcal{I}_3$ is derived from the fact that $a(s/m)=0$ for $s \geq m+1$.  
    
We now proceed to investigate each term individually. For the first term, we employ the properties of the kernel function as specified in Assumption \ref{AS2} to establish its convergence. For $\forall \epsilon>0$, let $\epsilon^\ast = \frac{1}{2}\epsilon|\Delta_{q_a}|^{-1}$. 
By Assumption \ref{AS2}, there exists a positive constant $ \varsigma_\epsilon $ such that $\Big|\frac{1-a(u)}{|u|^{q_a}}-C_{q_a}\Big|<\epsilon^\ast$ for any $|u|\leq \varsigma_\epsilon$. Let $m^\ast = \lfloor m \varsigma_\epsilon\rfloor $. Without loss of generality, we assume   $m^\ast\leq m$, as $ \varsigma_\epsilon $ can be chosen sufficiently small. Then, we can write
\begin{eqnarray*}
    \mathcal{I}_1&=&\frac{1}{L_N}\sum_{s=-m^\ast}^{m^\ast} \Big\{ a\Big(\frac{s}{m}\Big)-1\Big\} \mathbf{1}_N^\top E\big[ \mathbf{x}_{1}\mathbf{x}_{1+s}^\top\big]\mathbf{1}_N\notag \\
    &&+\frac{2}{L_N}\sum_{s=m^\ast+1}^{m} \Big\{a\Big(\frac{s}{m}\Big)-1\Big\} \mathbf{1}_N^\top E\big[ \mathbf{x}_{1}\mathbf{x}_{1+s}^\top\big]\mathbf{1}_N
    \notag\\
    &:=&\mathcal{I}_{1,1}+\mathcal{I}_{1,2}.
\end{eqnarray*}

For $ \mathcal{I}_{1,1}$, it is clear to see that
\begin{eqnarray*}
    \mathbf{1}_N^\top E\big[ \mathbf{x}_{1}\mathbf{x}_{1+|s|}^\top\big]\mathbf{1}_N&=&\sum_{\ell=0}^{\infty} [(\mathbf{1}_N^\top \mathbf{B}_{\ell+|s|})\otimes (\mathbf{1}_N^\top \mathbf{B}_\ell)]E[\text{vec}(\pmb{\varepsilon}_{1-\ell}  \pmb{\varepsilon}_{1-\ell}^\top)]
    \notag\\
    &&+\sum_{\ell=0}^{\infty}\sum_{v=0}^{|s|+\ell-1}[(\mathbf{1}_N^\top \mathbf{B}_v)\otimes (\mathbf{1}_N^\top \mathbf{B}_{\ell})] E[\text{vec}(\pmb{\varepsilon}_{1-\ell} \pmb{\varepsilon}_{1+|s|-v}^\top )]
    \notag\\
    &&+\sum_{\ell=0}^{\infty}\sum_{v>|s|+\ell}^{\infty} [(\mathbf{1}_N^\top \mathbf{B}_v)\otimes (\mathbf{1}_N^\top \mathbf{B}_{\ell})]E[ \text{vec}(\pmb{\varepsilon}_{1-\ell} \pmb{\varepsilon}_{1+|s|-v}^\top )]
    \notag\\
    &=&\sum_{\ell=0}^{\infty} \mathbf{1}_N^\top \mathbf{B}_\ell  \mathbf{B}_{\ell+|s|}^\top \mathbf{1}_N.
\end{eqnarray*}

Note that \begin{eqnarray*}
    &&L_N^{-1}\sum_{k=1}^{\infty}\sum_{\ell=0}^{\infty}k^{q_a} \big|\mathbf{1}_N^\top \mathbf{B}_\ell  \mathbf{B}_{\ell+k}^\top \mathbf{1}_N\big|<\infty\notag \\
    &\le & \frac{N}{L_N} \sum_{\ell=0}^{\infty}\|\mathbf{B}_\ell \|_2 \sum_{k=1}^{\infty}k^{q_a}\| \mathbf{B}_{\ell+k}\|_2\notag \\
    &\le & \sqrt{\frac{N}{L_N}} \sum_{\ell=0}^{\infty}\|\mathbf{B}_\ell \|_2 \sqrt{\frac{N}{L_N}}\sum_{k=1}^{\infty}k^{q_a}\| \mathbf{B}_{k}\|_2<\infty,
\end{eqnarray*}
where the last line follows from Assumption \ref{AS2}.1. Additionally, since $|s/m|\leq \varsigma_\epsilon$ for $|s|\leq m^\ast$, we have
\begin{eqnarray}\label{msel_1}
    \big|m^{q_a}\mathcal{I}_{1,1}+C_{q_a}\Delta_{q_a}\big|<|\Delta_{q_a}|\epsilon^\ast=\frac{1}{2}\epsilon,
\end{eqnarray}
for sufficiently large $N$ and $T$. For $\mathcal{I}_{1,2}$, it is clear to see that $|s/m|\geq \varsigma_\epsilon$  and for sufficiently large $N$ and $T$,
\begin{eqnarray}\label{msel_2}
    \big|m^{q_a}\mathcal{I}_{1,2}\big|\leq \frac{2(\max_u(a(u))+1)}{\varsigma_\epsilon^{q_a}L_N}\sum_{s=m^\ast+1}^{m}s^{q_a}\big|\mathbf{1}_N^\top E\big[ \mathbf{x}_{1}\mathbf{x}_{1+s}^\top\big]\mathbf{1}_N\big|
    &<&\frac{1}{2}\epsilon,
\end{eqnarray}
where the inequality holds due to the fact that  $L_N^{-1}\sum_{s=m^\ast+1}^{m}|s|^{q_a}\big|\mathbf{1}_N^\top E\big[ \mathbf{x}_{1}\mathbf{x}_{1+s}^\top\big]\mathbf{1}_N\big|$ converges to zero as both $m,m^\ast\rightarrow\infty$. By \eqref{msel_1} and \eqref{msel_2}, we have 
\begin{eqnarray*}
    \big|m^{q_a}\mathcal{I}_{1}+C_{q_a}\Delta_{q_a}\big|<\epsilon,
\end{eqnarray*}
for sufficiently large $N$ and $T$. Since $\epsilon$ is an arbitrarily small positive number, we obtain
\begin{eqnarray}\label{msel_3}
    \mathcal{I}_{1}=-C_{q_a}m^{-q_a}\Delta_{q_a}+o_P(m^{-q_a}).
\end{eqnarray}

We then proceed to study $\mathcal{I}_{2}$. Specifically, we have
\begin{eqnarray}\label{msel_4}
|\mathcal{I}_{2}|&\leq&\frac{2(\max_u(a(u))+1)}{L_NT}\sum_{s=1}^{m} s \big|\mathbf{1}_N^\top E\big[ \mathbf{x}_{1}\mathbf{x}_{1+s}^\top\big]\mathbf{1}_N\big|
\notag\\
&=&O\left(\frac{1}{T}\right)=o(m^{-q_a}),
\end{eqnarray}
where the second equality holds because $m^2/T \rightarrow 0$, as required in Assumption \ref{AS2}.

For $\mathcal{I}_{3}$, 
\begin{eqnarray}\label{msel_5}
    |\mathcal{I}_{3}|&\leq& \frac{2}{L_N}\sum_{s=m+1}^{T-1} \big|\mathbf{1}_N^\top E\big[ \mathbf{x}_{1}\mathbf{x}_{1+s}^\top\big]\mathbf{1}_N\big|
    \notag\\
    &\leq& \frac{2}{L_Nm^2}\sum_{s=m+1}^{T-1}s^2 \big|\mathbf{1}_N^\top E\big[ \mathbf{x}_{1}\mathbf{x}_{1+s}^\top\big]\mathbf{1}_N\big|
    \notag\\
    &=&o\left(\frac{1}{m^2}\right)=o_P(m^{-q_a}),
\end{eqnarray}
where the first equality follows from
\begin{eqnarray*}
   \frac{1}{L_N}\sum_{s=m+1}^{T-1}s^2 \big|\mathbf{1}_N^\top E\big[ \mathbf{x}_{1}\mathbf{x}_{1+s}^\top\big]\mathbf{1}_N\big|&\leq&\frac{1}{L_N}\sum_{s=m+1}^{T-1}\sum_{\ell=0}^{\infty}s^2 \big| \mathbf{1}_N^\top \mathbf{B}_\ell  \mathbf{B}_{\ell+s}^\top \mathbf{1}_N\big| = o(1),  
\end{eqnarray*}
with the equality holding by Assumption \ref{AS2}.

Combining \eqref{msel_3}, \eqref{msel_4} and \eqref{msel_5}, we obtain the first desired result in \eqref{msel_0}. Next, we investigate the variance of $E^\ast[\widetilde{S}_{NT}^{*2}]$.

For notational simplicity, let $\mathcal{U}_{ts}=L_N^{-1}\mathbf{1}_N^\top\mathbf{x}_{t}\mathbf{x}_{s}^\top\mathbf{1}_N$. Additionally, define 
\begin{eqnarray*}
    \varsigma (s)=E[\mathcal{U}_{1,1+s}],\,\upsilon (s_1,s_2,s_3)= E[\mathcal{U}_{1,1+s_1}\mathcal{U}_{1+s_2,1+s_3}],
\end{eqnarray*}
and $\varrho (s_1,s_2,s_3)=\upsilon (s_1,s_2,s_3)- \varsigma(s_1) \varsigma(s_2-s_3)- \varsigma(s_2) \varsigma(s_1-s_3)- \varsigma(s_3) \varsigma(s_1-s_2)$.

To derive the variance of $E^\ast[\widetilde{S}_{NT}^{*2}]$, we write
\begin{eqnarray*}
    && \sigma_x^4 Var( E^\ast[\widetilde{S}_{NT}^{*2}])
    \notag\\ 
    &=& \frac{1}{T^2}\sum_{t_1,t_2=1}^T\sum_{s_1=1-t_1}^{T-t_1}\sum_{s_2=1-t_2}^{T-t_2}  a\Big(\frac{s_1}{m}\Big)a\Big(\frac{s_2}{m}\Big) Cov\big(\mathcal{U}_{t_1,t_1+s_1}, \mathcal{U}_{t_2,t_2+s_2}\big)
     \notag\\
     &=& \frac{1}{T^2}\sum_{s_1=-m}^{m}\sum_{s_2=-m}^{m}  a\Big(\frac{s_1}{m}\Big)a\Big(\frac{s_2}{m}\Big) Cov\big(\sum_{t_1=\psi_{1,s_1}}^{\psi_{2,s_1}}\mathcal{U}_{t_1,t_1+s_1}, \sum_{t_2=\psi_{1,s_2}}^{\psi_{2,s_2}}\mathcal{U}_{t_2,t_2+s_2}\big)
     ,
 \end{eqnarray*}
where  $\psi_{1,s}=\max\{1,1-s\}$, $\psi_{2,s}=\min\{T,T-s\}$.

We initially focus on the case of $0\leq s_1\leq s_2$. Using standard results for the fourth moment of stationary time series  \citep[see, for example, equation (30) in Section 8.3 of][]{Anderson2011}, we obtain 
\begin{eqnarray*}
    \frac{1}{T-s_2}Cov\big(\sum_{t_1=1}^{T-s_1}\mathcal{U}_{t_1,t_1+s_1}, \sum_{t_2=1}^{T-s_2}\mathcal{U}_{t_2,t_2+s_2}\big)&:=&\varpi_{s_1,s_2,1}+ \varpi_{s_1,s_2,2}+\varpi_{s_1,s_2,3},
\end{eqnarray*}
where 
\begin{eqnarray*}
    \varpi_{s_1,s_2,1}&=& \sum_{r=0}^{s_2-s_1}\{\varsigma(r)\varsigma(r+s_1-s_2)+\varsigma(r-s_2)\varsigma(r+s_1)\}
    \notag\\
    &&+\sum_{r=s_2-s_1+1}^{T-s_1-1}\Big(1-\frac{r-(s_2-s_1)}{T-s_2}\Big)\{\varsigma(r)\varsigma(r+s_1-s_2)+\varsigma(r-s_2)\varsigma(r+s_1)\}
    \notag\\
    &&+\sum_{r=-(T-s_2-1)}^{-1}\Big(1-\frac{|r|}{T-s_2}\Big)\{\varsigma(r)\varsigma(r+s_1-s_2)+\varsigma(r-s_2)\varsigma(r+s_1)\}
    \notag\\
    \varpi_{s_1,s_2,2}&=& \sum_{r=0}^{s_2-s_1}\varrho (s_1,-r,s_2-r)
    +\sum_{r=s_2-s_1+1}^{T-s_1-1}\Big(1-\frac{r-(s_2-s_1)}{T-s_2}\Big)\varrho (s_1,-r,s_2-r)
    \notag\\
    &&+\sum_{r=-(T-s_2-1)}^{-1}\Big(1-\frac{|r|}{T-s_2}\Big)\varrho (s_1,-r,s_2-r).
\end{eqnarray*}

Regarding terms with $\varpi_{s_1,s_2,1}$,  applying the same arguments as in the proof of Corollary 8.3.1 of \cite{Anderson2011}, we find that
\begin{eqnarray}\label{msel_6}
    &&\frac{1}{T^2}\sum_{s_1= 0}^{m}\sum_{s_2=s_1+1}^{m} a\Big(\frac{s_1}{m}\Big)a\Big(\frac{s_2}{m}\Big) (T-s_2)\varpi_{s_1,s_2,1}
    \notag\\
    &=&\frac{1}{T}\sum_{s_1= 0}^{m}\sum_{s_2=s_1+1}^{m}\sum_{r=-\infty}^{\infty}  a\Big(\frac{s_1}{m}\Big)a\Big(\frac{s_2}{m}\Big) \{\varsigma(r+s_2)\varsigma(r+s_1)+\varsigma(r-s_2)\varsigma(r+s_1)\}(1+o(1)).
    \notag\\
\end{eqnarray}

Using equation (46) in Section 8.3 of \cite{Anderson2011}, we can demonstrate that the terms with $\varpi_{s_1,s_2,2}$ are asymptotically negligible:
\begin{eqnarray}\label{msel_7}
    &&\frac{1}{T^2}\sum_{s_1= 0}^{m}\sum_{s_2=s_1+1}^{m}  a\Big(\frac{s_1}{m}\Big)a\Big(\frac{s_2}{m}\Big) (T-s_2)|\varpi_{s_1,s_2,2}|
    \notag\\
    &=&\frac{\kappa_4}{T^2}\sum_{s_1= 0}^{m}\sum_{s_2=s_1+1}^{m}  a\Big(\frac{s_1}{m}\Big)a\Big(\frac{s_2}{m}\Big) (T-s_2)|\varsigma(s_1)||\varsigma(s_2)|(1+o(1))
    \notag\\
    &\leq&\frac{\kappa_4\max_u a(u)^2}{T}\sum_{s_1= 0}^{m}\sum_{s_2=s_1+1}^{m}|\varsigma(s_1)||\varsigma(s_2)|(1+o(1)) = O\left(\frac{1}{T}\right)=o\left(\frac{l}{T}\right),
\end{eqnarray}
where $\kappa_4=\sigma_x^{-4}E[|L_N^{-1/2}\mathbf{1}_N^\top\mathbf{x}_{t}|^4]-3$ and the second equality follows from $\sum_{s= 0}^{\infty}|\varsigma(s)|=O(1)$. Combining \eqref{msel_6} and \eqref{msel_7}, we obtain
\begin{eqnarray*}
    &&\frac{1}{T}\sum_{s_1= 0}^{m}\sum_{s_2=s_1+1}^{m}  a\Big(\frac{s_1}{m}\Big)a\Big(\frac{s_2}{m}\Big) Cov\big(\sum_{t_1=1}^{T-s_1}\mathcal{U}_{t_1,t_1+s_1}, \sum_{t_2=1}^{T-s_2}\mathcal{U}_{t_2,t_2+s_2}\big)
    \notag\\
    &=&\frac{1}{T}\sum_{s_1= 0}^{m}\sum_{s_2=s_1+1}^{m}  \sum_{r=-\infty}^\infty a\Big(\frac{s_1}{m}\Big)a\Big(\frac{s_2}{m}\Big) \{\varsigma(r+s_2)\varsigma(r+s_1)+\varsigma(r-s_2)\varsigma(r+s_1)\}+o\left(\frac{l}{T}\right).
\end{eqnarray*}

In the case of other pairings of  $s_1$ and $s_2$, similar arguments yield the leading-order terms. Ultimately, we establish that
\begin{eqnarray*}
    &&\sigma_x^4 Var( E^\ast[\widetilde{S}_{NT}^{*2}]) \notag \\
    &=&\frac{1}{T}\sum_{s_1= -m}^{m}\sum_{s_2=-m}^{m}  \sum_{r=-\infty}^\infty a\Big(\frac{s_1}{m}\Big)a\Big(\frac{s_2}{m}\Big) \{\varsigma(r+s_2)\varsigma(r+s_1)+\varsigma(r-s_2)\varsigma(r+s_1)\}
    +o\left(\frac{l}{T}\right)
    \notag\\
    &:=&\mathcal{I}_4+\mathcal{I}_5+o\left(\frac{l}{T}\right),
\end{eqnarray*}
where the definitions of $\mathcal{I}_4$ and  $\mathcal{I}_5$ are obvious. 

For $\mathcal{I}_4$, we can further write
\begin{eqnarray}\label{msel_8}
    \frac{T}{m}\mathcal{I}_4&=&\frac{1}{m} \sum_{r=-\infty}^\infty\sum_{s_1= r-m}^{r+m}\sum_{s_2=r-m}^{r+m}  a\Big(\frac{s_1-r}{m}\Big)a\Big(\frac{s_2-r}{m}\Big) \varsigma(s_1)\varsigma(s_2)
    \notag\\
    &=&\frac{1}{m} \sum_{r=-\infty}^\infty\sum_{s_1= r-m}^{r+m}\sum_{s_2=r-m}^{r+m}a^2\Big(\frac{r}{m}\Big) \varsigma(s_1)\varsigma(s_2)+o(1)
    \notag\\
    &=&\sigma_x^4\int_{-1}^1 a^2(u)du +o(1),
\end{eqnarray}
where the second equality follows from Lipschitz continuity of $a(u)$, such that $|a(\frac{s-r}{m})-a(\frac{-r}{m})|=O(\frac{s}{m})$ and from the fact  $\sum_{s=-\infty}^{\infty}|s||E[\mathcal{U}_{1,1+s}]|=O(1)$, both  implied by Assumption \ref{AS2}. By similar arguments, we can obtain 
\begin{eqnarray}\label{msel_9}
    \frac{T}{m}\mathcal{I}_5=\sigma_x^4\int_{-1}^1 a^2(u)du +o(1). 
\end{eqnarray}
\eqref{msel_8} and \eqref{msel_9} together yield the second desired result in \eqref{msel_0}. Therefore,  Theorem \ref{THM.2}.3 is established.
\end{proof}

\medskip

\begin{proof}[Proof of Theorem \ref{THM.3}]
\item 

Before proceeding, we introduce some notation that will be repeatedly used in this proof. Note that by Lemma \ref{LM.A5} $\widetilde{p}_i$ admits the following decomposition:

\begin{eqnarray*}
p_i &=& \mathbf{e}_i^\top \sum_{t=1}^T\mathbf{x}_t = \mathbf{e}_i^\top \mathbf{B} \sum_{s=1}^T\pmb{\varepsilon}_{s}-\mathbf{e}_i^\top\widetilde{\mathbf{B}}(L)\pmb{\varepsilon}_{T}+\mathbf{e}_i^\top\widetilde{\mathbf{B}}(L)\pmb{\varepsilon}_{0} = \mathbf{e}_i^\top\sum_{\ell =-\infty}^T \pmb{\mathcal{B}}_{T\ell}\pmb{\varepsilon}_\ell.
\end{eqnarray*}

Let $\mathbf{B}_\ell=(\mathbf{b}_{\ell 1}^\sharp,\ldots, \mathbf{b}_{\ell N}^\sharp)^\top$ and $\widetilde{\mathbf{B}}_{\ell}=(\widetilde{\mathbf{b}}_{\ell 1}^\sharp,\ldots, \widetilde{\mathbf{b}}_{\ell N}^\sharp)^\top$,
where $\mathbf{B}_\ell$ and $\widetilde{\mathbf{B}}_{\ell}$ have been defined in \eqref{def.xt_1} and \eqref{def.BN1} respectively. We can then write 

\begin{eqnarray}\label{decom_bi_tilde}
\mathbf{e}_i^\top\widetilde{\mathbf{B}}(1)&=&\mathbf{e}_i^\top\sum_{\ell=0}^{\infty}\widetilde{\mathbf{B}}_\ell =\sum_{\ell=0}^{\infty}\widetilde{\mathbf{b}}_{\ell i}^{\sharp \top} = \sum_{\ell=0}^{\infty}\sum_{k=\ell+1}^\infty\mathbf{b}_{k i}^{\sharp\top} = \sum_{\ell=1}^{\infty}\ell \mathbf{b}_{\ell i}^{\sharp \top},
\end{eqnarray}
which yields that $E |\mathbf{e}_i^\top\widetilde{\mathbf{B}}(L)\pmb{\varepsilon}_{T} |^2 = \sum_{\ell=1}^{\infty}\ell^2 \|\mathbf{b}_{\ell i}^{\sharp}\|_2^2$. Therefore, $|\mathbf{e}_i^\top\widetilde{\mathbf{B}}(L)\pmb{\varepsilon}_{T}| =O_P(1)$. Similarly, we obtain that $|\mathbf{e}_i^\top\widetilde{\mathbf{B}}(L)\pmb{\varepsilon}_{0}| =O_P(1)$. We can then conclude that
\begin{eqnarray*}
\text{Var}\left(\frac{1}{\sqrt{T}}p_i\right) =  \mathbf{e}_i^\top \mathbf{B}\mathbf{B}^\top \mathbf{e}_i +o(1)\eqqcolon \sigma_{p,i}^2+o(1).
\end{eqnarray*}

In what follows, we study $\frac{1}{\sigma_{p,i}\sqrt{T}}p_i$, and write
\begin{eqnarray*}
\widetilde{p}_i\coloneqq \frac{1}{\sigma_{p,i}\sqrt{T}}p_i =\frac{1}{\sigma_{p,i}\sqrt{T}}  \mathbf{e}_i^\top\sum_{\ell =-\infty}^T \pmb{\mathcal{B}}_{T\ell}\pmb{\varepsilon}_\ell\eqqcolon    \mathbf{e}_i^\top\sum_{\ell =-\infty}^T \dot{\pmb{\mathcal{B}}}_{\ell}\pmb{\varepsilon}_\ell,
\end{eqnarray*}
where $\frac{1}{\sigma_{p,i}\sqrt{T}}\pmb{\mathcal{B}}_{T\ell}\eqqcolon \dot{\pmb{\mathcal{B}}}_{\ell}= \{\dot{\mathtt{b}}_{\ell, ij}\}_{N\times N}$, and we have suppressed $T$ in $\dot{\pmb{\mathcal{B}}}_{\ell}$ for notational simplicity.

To proceed, denote by $\psi_{i}(u)$ the characteristic function of $\widetilde{p}_i$. Thus,
\begin{eqnarray}\label{def_psi_i}
\psi_i(u) &=& E\left[\exp\left(\mathsf{i}u  \sum_{\ell=-\infty}^T\mathbf{e}_i^\top\dot{\pmb{\mathcal{B}}}_{\ell} \pmb{\varepsilon}_{\ell}\right)\right]\notag \\
&=&\prod_{\ell=-\infty}^T \prod_{j=1}^NE [\exp (\mathsf{i}u \dot{\mathtt{b}}_{\ell ,ij}\varepsilon_{j\ell} )] =  \prod_{\ell=-\infty}^T \prod_{j=1}^N\psi(\dot{\mathtt{b}}_{\ell ,ij} u),
\end{eqnarray}
where the second equality follows from $\{\varepsilon_{it}\}$  being i.i.d. over both dimensions. We have suppressed $T$ in $\dot{\mathtt{b}}_{\ell ,ij}$ for notational simplicity.

Using \eqref{def_psi_i}, we are able to calculate the $r^{th}$ cumulant $\beta_{ir}$ of $\widetilde{p}_i$ for $r\ge 1$. Obviously, we have $\beta_{i1} =0$ and $\beta_{i2} =1$. For $r\ge 3$, write

\begin{eqnarray} \label{def_psi_i2}
\beta_{ir} &=& (-\mathsf{i})^r \frac{\mathrm{d}^r}{\mathrm{d}u^r}\log  \psi_i(u) |_{u=0}\notag \\
&=&\sum_{\ell=-\infty}^T\sum_{j=1}^N (-\mathsf{i})^r\frac{\mathrm{d}^r}{\mathrm{d}u^r}\log\psi(\dot{\mathtt{b}}_{\ell ,ij} u) |_{u=0} \notag \\
&=&\sum_{\ell=-\infty}^T\sum_{j=1}^N \dot{\mathtt{b}}_{\ell ,ij}^r (-\mathsf{i})^r\frac{\mathrm{d}^r}{\mathrm{d}u^r}\log\psi(u)|_{u=0}\notag \\
&=&\kappa_r\left(\sum_{\ell=1}^T\sum_{j=1}^N \dot{\mathtt{b}}_{\ell ,ij}^r  +\sum_{\ell=-\infty}^0\sum_{j=1}^N \dot{\mathtt{b}}_{\ell ,ij}^r \right),
\end{eqnarray}
where the second equality follows from \eqref{def_psi_i}, and the third equality follows from \eqref{prop.cum}.  

For the second term on the right hand side of \eqref{def_psi_i2}, we note that 
\begin{eqnarray} \label{def_psi_i3}
\max_i\left|\sum_{\ell=-\infty}^0\sum_{j=1}^N \dot{\mathtt{b}}_{\ell ,ij}^r \right|&\le & \max_i \frac{1}{(\sigma_{p,i}^2 T)^{r/2}}\sum_{\ell=-\infty}^0\sum_{j=1}^N|b_{T\ell, ij}|^r\notag \\
&\le &\max_i\frac{1}{(\sigma_{p,i}^2 T)^{r/2}}\left(\sum_{\ell=-\infty}^0\sum_{j=1}^N |b_{T\ell, ij}|^2\right)^{r/2}\notag\\
&\le &\max_i O(1)\frac{1}{T^{r/2}}\left(\sum_{\ell=-\infty}^0 \| \widetilde{\mathbf{b}}_{\ell i}^\sharp\|^2  \right)^{r/2} = O\left(\frac{1}{T^{r/2}}\right),
\end{eqnarray}
where $\pmb{\mathcal{B}}_{T\ell}\eqqcolon \{b_{T\ell, ij}\}_{N\times N}$, the second inequality follows from the fact that for a vector $\mathbf{x}$, $|\mathbf{x}|_{p_1}\le |\mathbf{x} |_{p_2}$ for any $p_1> p_2 \ge 1$, the third inequality follows from Lemma \ref{LM.A5}.1, and the last equality follows from Assumption \ref{AS3} and \eqref{decom_bi_tilde}. 

Thus, \eqref{def_psi_i2} and \eqref{def_psi_i3} together infer that

\begin{eqnarray*}  
\max_i\left|\beta_{ir} -\kappa_r \sum_{\ell=1}^T\sum_{j=1}^N \dot{\mathtt{b}}_{\ell ,ij}^r\right|=O\left(\frac{1}{T^{r/2}}\right).
\end{eqnarray*}

Recall the definitions of $\dot{\mathtt{b}}_{\ell ,ij}^r$ for $\ell\ge 1$ and $j\ge 1$, and write further that

\begin{eqnarray*} 
\max_i\left| \sum_{\ell=1}^T\sum_{j=1}^N \dot{\mathtt{b}}_{\ell ,ij}^r \right|&\le& \max_i\frac{2^{r-1}}{(\sigma_{p,i}^2 T)^{r/2}}\sum_{\ell=1}^T\sum_{j=1}^N[|b_{ij}|^r+| \widetilde{b}_{T-\ell, ij}|^r]\notag\\
&=& \max_i\frac{2^{r-1}}{\sigma_{p,i}^rT^{r/2-1}} \sum_{j=1}^N | b_{ij} |^r+O\left(\frac{1}{T^{r/2}}\right)\notag \\
&\le & \max_i\frac{2^{r-1}}{\sigma_{p,i}^rT^{r/2-1}} \left(\sum_{j=1}^N | b_{ij} |^2\right)^{r/2}+O\left(\frac{1}{T^{r/2}}\right)\notag \\
&=&\max_i\frac{2^{r-1}}{\sigma_{p,i}^rT^{r/2-1}} \|\mathbf{b}_{i}^{\sharp}\|_2^r+O\left(\frac{1}{T^{r/2}}\right),
\end{eqnarray*}
where the first inequality follows from $(a+b)^r\le 2^{r-1}(a^r+b^r)$ for $a,b\ge 0$ and $r>1$, the first equality follows from a development similar to \eqref{def_betar2}, and the second inequality follows from the fact that for a vector $\mathbf{x}$, $|\mathbf{x}|_{p_1}\le |\mathbf{x} |_{p_2}$ for any $p_1> p_2 \ge 1$. Thus, by \eqref{def_betar3} and \eqref{def_betar4}, we can obtain that for $r\ge 3$

\begin{eqnarray} \label{def_psi_i6}
|\beta_{ir}|& \le &\max_i\frac{2^{r-1}}{\sigma_{p,i}^rT^{r/2-1}} \|\mathbf{b}_{i}^{\sharp}\|_2^r+O\left(\frac{1}{T^{r/2}}\right) \notag \\
&=& O\left( \frac{1}{T^{r/2 -1}} \right) \le  O\left( \frac{1}{(T^{1/6})^r} \right),
\end{eqnarray}
where the equality follows from Assumption \ref{AS3}, and the second inequality follows from $r\ge 3$.

By applying \eqref{def.chiu} to $\psi_i(u)$ and the characteristic function of $N(0,1)$, we have

\begin{eqnarray} \label{def_psi_i7}
\psi_i(u) &=&\exp\left\{ \sum_{r=3}^\infty \beta_{ir}\cdot\frac{(\mathsf{i} u)^r}{r!}\right\} \cdot\widetilde{\phi}(u)\notag \\
&=&\left\{1 +\sum_{n=1}^\infty\frac{1}{n!}\left(\sum_{r=3}^\infty \beta_{ir}\cdot\frac{(\mathsf{i} u)^r}{r!}\right)^n \right\}\cdot\widetilde{\phi}(u),
\end{eqnarray}
where $\beta_{ir}$ is bounded by \eqref{def_psi_i6}, and the second equality follows from \eqref{def.expu}. Let $f_i(w)$ be the density function of $\widetilde{p}_i$, and the Fourier inversion of \eqref{def_psiNT2} leads to the following expansion:

\begin{eqnarray*}
f_i(w) &=&\phi(w) +\frac{1}{2\pi}\int_{\mathbb{R}} \exp(-\mathsf{i}uw)\cdot \beta_{i3}\cdot \frac{(\mathsf{i}u)^3}{3!}\cdot\widetilde{\phi}(u)\mathrm{d}u \notag \\
&&+\frac{1}{2\pi}\int_{\mathbb{R}} \exp(-\mathsf{i}uw)\cdot\widetilde{\psi}_i(u)\cdot\widetilde{\phi}(u)\mathrm{d}u\notag \\
&=&\phi(w) + \frac{\beta_{i3}}{6} H_3(w)\phi(w)\notag \\
&&+\frac{1}{2\pi}\int_{\mathbb{R}} \exp(-\mathsf{i}uw)\cdot\widetilde{\psi}_i(u)\cdot\widetilde{\phi}(u)\mathrm{d}u
\end{eqnarray*}
where $\widetilde{\psi}_i(u)\coloneqq \sum_{r=4}^\infty \beta_{ir}\cdot\frac{(\mathsf{i} u)^r}{r!} + \sum_{n=2}^\infty\frac{1}{n!}\left(\sum_{r=3}^\infty \beta_{ir}\cdot\frac{(\mathsf{i} u)^r}{r!}\right)^n$, and the second equality follows from Lemma \ref{LM.A4}. 

We then bound $\frac{1}{2\pi}\int_{\mathbb{R}} \exp(-\mathsf{i}uw)\cdot \widetilde{\psi}_i(u)\cdot\widetilde{\phi}(u)\mathrm{d}u$. First, note  

\begin{eqnarray} \label{def_psi_i8}
\left|\sum_{r=4}^\infty \beta_{ir}\cdot\frac{(\mathsf{i} u)^r}{r!} \right|\widetilde{\phi}(u) &\le & |\beta_{i4}| \widetilde{\phi}(u) \sum_{r=4}^\infty \frac{|u|^r}{r!} \notag \\
&\le & |\beta_{i4}| \widetilde{\phi}(u)\exp(|u|) =|\beta_{i4}| \exp(-u^2/2+|u|),
\end{eqnarray}
where the second inequality follows from \eqref{def.expu}. Second, in connection with \eqref{def_psi_i6}, assuming that $|u|\le T^{1/6}$ and using Taylor theorem twice as in \eqref{def.expu2} we write

\begin{eqnarray} \label{def_psi_i9}
&&\left|\sum_{n=2}^\infty\frac{1}{n!}\left(\sum_{r=3}^\infty \beta_{ir}\cdot\frac{(\mathsf{i} u)^r}{r!}\right)^n \right|\widetilde{\phi}(u)\notag \\
&\le &O(1)\sum_{n=2}^\infty\frac{1}{n!} \left| \sum_{r=3}^\infty \cdot\frac{(\mathsf{i} u/T^{1/6} )^r}{r! }\right|^n \widetilde{\phi}(u)\notag\\
&\le & O(1)\frac{1}{T}\widetilde{\phi}(u) \widetilde{u}^6=O(1)\frac{1}{T},
\end{eqnarray}
where $|\widetilde{u}|$ is in between 0 and $T^{1/6}$, and the last equality follows from $\widetilde{\phi}(u) \widetilde{u}^6$ being uniformly bounded. Thus, by \eqref{def_psi_i6}, \eqref{def_psi_i8} and \eqref{def_psi_i9}, we can write

\begin{eqnarray} \label{def_psi_i10}
\sup_{|u|\le T^{1/6}}\left|\frac{1}{2\pi}\int_{\mathbb{R}} \exp(-\mathsf{i}uw)\cdot \widetilde{\psi}_i(u)\cdot\widetilde{\phi}(u)\mathrm{d}u\right| &=&O\left(  \frac{1}{T}  \right).
\end{eqnarray}
    
Thus, according to \eqref{def_psi_i7} and \eqref{def_psi_i10}, we write further 

\begin{eqnarray} \label{def_psi_i11}
&&\sup_{w\in\mathbb{R}}\left|f_i(w)-\phi(w)-\frac{\beta_{i3}}{6} H_3(w)\phi(w) \right|\notag \\
&\le &\frac{1}{2\pi}\int_{\mathbb{R}}\left|\psi_i(u) -\widetilde{\phi}(u) - \frac{\beta_{i3}}{6} (\mathsf{i}u)^3 \widetilde{\phi}(u)\right|\mathrm{d}u\notag\\
&= &\frac{1}{2\pi}\int_{|u|\ge T^{1/6}} |\widetilde{\psi}_i(u) |\widetilde{\phi}(u)\mathrm{d}u +\frac{1}{2\pi}\int_{|u|< T^{1/6}} |\widetilde{\psi}_i(u) |\widetilde{\phi}(u)\mathrm{d}u = O\left( \frac{1}{T}\right),
\end{eqnarray}
where, in the last step, $\int_{|u|\ge T^{1/6}} |\widetilde{\psi}_i(u) |\widetilde{\phi}(u)\mathrm{d}u$ can be arbitrarily small by standard operation, and $\frac{1}{2\pi}\int_{|u|< T^{1/6}} |\widetilde{\psi}_i(u) |\widetilde{\phi}(u)\mathrm{d}u$ is bounded by \eqref{def_psi_i10}.

To study the CDF, we let $G_i(w) =\Phi(w)+\frac{\beta_{i3}}{6}(1-w^2)\phi(w)$. Simple algebra shows that

\begin{eqnarray*}
\frac{\mathrm{d}( (1-w^2)\phi(w) )}{\mathrm{d}w} =  H_3(w)\phi(w) .
\end{eqnarray*}
Thus, $G_i(w)$ has a characteristic function $\xi_i(w) =(1+\frac{\beta_{i3}}{6}(\mathsf{i}w)^3)\widetilde{\phi}(w)$. We then invoke Esseen's smoothing Lemma. First, we let $a$ involved in Lemma \ref{LM.A2} be sufficiently large, so the second term on the right hand side of Lemma \ref{LM.A2} becomes negligible. Then, similar to \eqref{def_psi_i11}, we study the first term and can obtain that

\begin{eqnarray*}
|F_i(w)-G_i(w)|_\infty=O\left( \Delta_{NT}(4) \vee \frac{1}{T^2}\right).
\end{eqnarray*}
The proof of is now completed. 
\end{proof}

\medskip

\begin{proof}[Proof of Theorem \ref{THM.4}]
\item  
    
(1). For the heterogeneous bootstrap statistics, we can apply a similar approach to what has been used in the proof of Theorem \ref{THM.2}.1 to establish its Edgeworth expansion.
    
Let $\mathcal{P}_{i}^*=  \sum_{t=1}^T\mathbf{e}_i^\top \mathbf{x}_t\zeta_{t}$ and $\sigma_{\mathcal{P},i}^{\ast}=\text{Var}^\ast(\mathcal{P}_{i}^*)^{\frac{1}{2}}$. We need to establish the Edgeworth expansion for $\widetilde{\mathcal{P}}_{i}^*\coloneqq \mathcal{P}_{i}^*/\sigma^\ast_{\mathcal{P},i}$.   
First, we make the following decomposition for $\widetilde{\mathcal{P}}_{i}^*$:
\begin{equation*}
    \widetilde{\mathcal{P}}_{i}^*=\sum_{s=1}^{\mathcal{T}_m} B_{i,s}+ B_{i,\mathcal{T}_m+1},
\end{equation*}
where $\mathcal{T}_m \coloneqq \lfloor \frac{T}{m} \rfloor$,  $B_{i,s} \coloneqq \sigma_{\mathcal{P},i}^{\ast-1}\sum_{t=(s-1)m+1}^{sm} \mathbf{e}_i^\top \mathbf{x}_t\zeta_t$ and $B_{i,\mathcal{T}_m+1} \coloneqq \sigma_{\mathcal{P},i}^{\ast-1}\sum_{t=\mathcal{T}_m+1}^{T} \mathbf{e}_i^\top \mathbf{x}_t\zeta_t$.  Additionally, define $B_{i,s,l}^c=\sum_{|t-s|>l} B_{i,t}$ and $B_{i,s,0}^c =  \widetilde{\mathcal{P}}_{i}^*$
for $l=1,\ldots, 4$ and $s=1,\ldots, \mathcal{T}_m+1$. 

To study the characteristic function $\phi^\ast_{i}(u)=E^\ast [\exp ( \mathsf{i} u\widetilde{\mathcal{P}}_{i}^*  ) ]$, we require the following expansion for $\exp ( \mathsf{i} u\widetilde{\mathcal{P}}_{i}^*  ) $, which can be derived analogously to \eqref{2Edx1}:

\begin{eqnarray*}
\exp ( \mathsf{i} u\widetilde{\mathcal{P}}_{i}^* )&=&\exp ( \mathsf{i} u B_{i,s,1}^c )+ \{\exp ( \mathsf{i} u (\widetilde{\mathcal{P}}_{i}^* - B_{i,s,1}^c ) )-1 \}\exp ( \mathsf{i} u B_{i,s,2}^c)\notag\\
&&+\sum_{l=3}^{4}\prod_{k=1}^{l-1} \{\exp ( \mathsf{i} u (B_{i,s,k-1}^c - B_{i,s,k}^c ) )-1 \}\exp ( \mathsf{i} u B_{i,s,l}^c )\notag\\
&&+\prod_{k=1}^{4} \{\exp ( \mathsf{i} u (B_{i,s,k-1}^c - B_{i,s,k}^c ) )-1 \}\exp ( \mathsf{i} u B_{i,s,4}^c ).
\end{eqnarray*} 

It further yields the following expansion for the first derivative of $\phi^\ast_{i}(u)$:
\begin{eqnarray}\label{2Edex1}
\frac{\mathrm{d}\phi^\ast_{i}(u)}{\mathrm{d}u}
&=& \mathsf{i}\sum_{s=1}^{\mathcal{T}_m+1} E^\ast [B_{i,s}\exp ( \mathsf{i} u  B_{i,s,1}^c ) ]\notag \\
&&+\mathsf{i}\sum_{s=1}^{\mathcal{T}_m+1} E^\ast [B_{i,s}  \{\exp ( \mathsf{i} u (\widetilde{\mathcal{P}}_{i}^* - B_{i,s,1}^c ) )-1 \}\exp ( \mathsf{i} u B_{i,s,2}^c ) ]\notag\\
&&+\mathsf{i}\sum_{s=1}^{\mathcal{T}_m+1}\sum_{l=3}^{4} E^\ast\Big[B_{i,s}\prod_{k=1}^{l-1} \{\exp ( \mathsf{i} u (B_{i,s,k-1}^c - B_{i,s,k}^c ) )-1 \}\exp ( \mathsf{i} u B_{i,s,l}^c )\Big] \notag\\
&&+\mathsf{i}\sum_{s=1}^{\mathcal{T}_m+1} E^\ast\Big[B_{i,s} \prod_{k=1}^{4} \{\exp ( \mathsf{i} u (B_{i,s,k-1}^c - B_{i,s,k}^c ) )-1 \}\exp ( \mathsf{i} u B_{i,s,4}^c )\Big]\notag\\
&\eqqcolon &\mathsf{J}_{i,1}(u)+\cdots+\mathsf{J}_{i,4}(u),
\end{eqnarray}
where $\mathsf{J}_{i,1}(u),\ldots,\mathsf{J}_{i,4}(u)$ are defined obviously. 

For $\mathsf{J}_{i,1}(u)$, it is straightforward to show that $\mathsf{J}_{i,1}(u)=0$. For $\mathsf{J}_{i,2}(u)$, 
\begin{eqnarray*}
\mathsf{J}_{i,2}(u)&=&-u\sum_{s=1}^{\mathcal{T}_m+1} E^\ast [B_{i,s} (\widetilde{\mathcal{P}}_{i}^* - B_{i,s,1}^c )\exp ( \mathsf{i} u B_{i,s,2}^c ) ]\notag \\
&&-\frac{1}{2}\mathsf{i}u^2\sum_{s=1}^{\mathcal{T}_m+1} E^\ast [B_{i,s} (\widetilde{\mathcal{P}}_{i}^* - B_{i,s,1}^c )^2\exp ( \mathsf{i} u B_{i,s,2}^c ) ] \notag\\
&&+\mathsf{i}\sum_{s=1}^{\mathcal{T}_m+1}E^\ast [B_{i,s} R^\ast_{i,s,1} (u)\exp ( \mathsf{i} u B_{i,s,2}^c ) ]\notag\\
&\eqqcolon &\mathsf{J}_{i,2,1}(u)+\mathsf{J}_{i,2,2}(u)+\mathsf{J}_{i,2,3}(u),
\end{eqnarray*}
where $R^\ast_{i,s,1}(u)=\sum_{l=3}^{\infty}\frac{\mathsf{(iu)}^l}{l!} (\widetilde{\mathcal{P}}_{i}^* - B_{i,s,1}^c )^l$ satisfies $|R^\ast_{i,s,1}(u)|\leq \frac{\mathsf{u}^3}{3!} |\widetilde{\mathcal{P}}_{i}^* - B_{i,s,1}^c |^3$, which is a direct application of Lemma \ref{LM.A3}.

By adopting similar arguments to those in the proof of \eqref{Edex4} and involking Lemma \ref{LM.A8}, for each $i$, we obtain 
\begin{eqnarray*}
    \big|\mathsf{J}_{i,2,1}(u)+u\phi^\ast_{i}(u)\big|=O_P\left(\frac{m}{T}\right),
\end{eqnarray*}
and $\Big|\mathsf{J}_{i,2,2}(u)+\frac{1}{2}\mathsf{i}u^2\sum_{s=1}^{\mathcal{T}_m+1} E^\ast [B_{i,s} (\widetilde{\mathcal{P}}_{i}^* - B_{i,s,1}^c )^2 ]\phi^\ast_{i}(u)\Big|=O_P\left(\frac{m}{T}\right)$ for any $u\in(-\infty,\infty)$.

In summary of these results, we can readily obtain
\begin{eqnarray}\label{2Edex2}
   |\mathsf{J}_{i,2}(u)-\mathsf{J}^\ast_{i,2}(u) |=O_P\left(\frac{m}{T}\right),
\end{eqnarray}
where $\mathsf{J}^\ast_{i,2}(u)=-u\phi^\ast_{i}(u)-\frac{1}{2}\mathsf{i}u^2\sum_{s=1}^{\mathcal{T}_m+1} E^\ast [B_{i,s} (\widetilde{\mathcal{P}}_{i}^* - B_{i,s,1}^c )^2 ]\phi^\ast_{i}(u)$.

Using arguments analogous to those for  \eqref{Edex14}, we can further show 
\begin{eqnarray}\label{2Edex3}
 |\mathsf{J}_{i,3}(u)-\mathsf{J}^\ast_{i,3}(u) |=O_P\left(\frac{m}{T}\right)\quad\text{and} \quad |\mathsf{J}_{i,4}(u) |=o_P\left(\frac{m}{T}\right),
\end{eqnarray}
where 

\begin{eqnarray*}
\mathsf{J}^\ast_{i,3}(u)&=&-\mathsf{i}u^2\sum_{s=1}^{\mathcal{T}_m+1}E^\ast [B_{i,s} (\widetilde{\mathcal{P}}_{i}^* - B_{i,s,1}^c ) (B_{i,s,1}^c-  B_{i,s,2}^c ) ]\phi^\ast_{i}(u).
\end{eqnarray*}

By combining \eqref{2Edex1}, \eqref{2Edex2}, and \eqref{2Edex3}, we obtain

\begin{eqnarray}\label{2Edex4}
    \big|\frac{\mathrm{d}\phi^\ast_{i}(u)}{\mathrm{d}u}-\mathsf{J}^\ast_{i,2}(u)-\mathsf{J}^\ast_{i,3}(u)\big|=O_P\left(\frac{m}{T}\right),
\end{eqnarray}
where $\mathsf{J}^\ast_{i,2}(u)$ and $\mathsf{J}^\ast_{i,3}(u)$ satisfy
\begin{eqnarray*}
&&\mathsf{J}^\ast_{i,2}(u)+\mathsf{J}^\ast_{i,3}(u) = -u\phi^\ast_{i}(u)-\frac{1}{2}\mathsf{i}u^2\sum_{s=1}^{\mathcal{T}_m+1} E^\ast [B_{i,s} (\widetilde{\mathcal{P}}_{i}^* - B_{i,s,1}^c )^2 ]\phi^\ast_{i}(u)\notag\\
&&-\mathsf{i}u^2\sum_{s=1}^{\mathcal{T}_m+1}E^\ast [B_{i,s} (\widetilde{\mathcal{P}}_{i}^* - B_{i,s,1}^c ) (B_{i,s,1}^c-  B_{i,s,2}^c ) ]\phi^\ast_{NT}(u) \notag\\
&=&-u\phi^\ast_{i}(u)-\frac{1}{2}\mathsf{i}u^2 \sum_{s=1}^{\mathcal{T}_m+1} E^\ast [B_{i,s} ( (\widetilde{\mathcal{P}}_{i}^* - B_{i,s,1}^c )^2+2 (\widetilde{\mathcal{P}}_{i}^* - B_{i,s,1}^c ) (B_{i,s,1}^c-  B_{i,s,2}^c ) ) ]\phi^\ast_{i}(u)
\notag\\
&=&-u\phi^\ast_{i}(u)-\frac{1}{2}\mathsf{i}u^2 E^\ast [\widetilde{\mathcal{P}}_{i}^{*3} ]\phi^\ast_{i}(u).
\end{eqnarray*}

Finally, in connection with Taylor expansion of $\exp (-\frac{1}{6}\mathsf{i}u^3 E^\ast[\widetilde{\mathcal{S}}_{i}^{*3} ] )$ and Lemma \ref{LM.A8}, integration of \eqref{2Edex4} leads to the following result for the characteristic function $\phi^\ast_{i}(u)$:

\begin{eqnarray*}
    \Big|\phi^\ast_{i}(u)-\exp\big(-\frac{1}{2}u^2\big) \big\{1-\frac{1}{6}\mathsf{i}u^3 E^\ast [\widetilde{\mathcal{P}}_{i}^{*3} ] \big\}\Big| =O_P\left(\frac{m}{T}\right).
\end{eqnarray*}

Using Esseen smoothing inequality, this condition is sufficient to establish Theorem \ref{THM.4}.1.

(2). By Theorem \ref{THM.3} and Theorem \ref{THM.4}.1, it follows for each $i$ that 
\begin{eqnarray*}
    &&\sup_{u\in \mathbb{R}}\left|\text{\normalfont Pr}^*(\widetilde{p}_{i}^* \le u) -\Pr(\widetilde{p}_{i}\le u)\right|\notag \\
    &\leq&\sup_{u\in \mathbb{R}}\left|\text{\normalfont Pr}^*(\widetilde{p}_{i}^* \le u) -\Phi(u)\right|
    +\sup_{u\in \mathbb{R}}\left|\Pr(\widetilde{p}_{i}\le u)-\Phi(u)\right| = O_P\left(\sqrt{\frac{m}{T}}\right).
\end{eqnarray*}

It completes the proof of Theorem \ref{THM.4}.2.

(3) Similar arguments used in the proof of Theorem \ref{THM.2}.3 can be directly applied here to establish the desired result in Theorem \ref{THM.4}.3. Consequently, the details are omitted for brevity.
\end{proof}

\medskip

\begin{proof}[Proof of Theorem \ref{THM.5}]
\item 

By using the BN decomposition in Lemma \ref{LM.A5}, we have 

\begin{eqnarray*}
\frac{1}{\sqrt{T}}\sum_{t=1}^T(\mathbf{x}_{t} -\bm{\mu}) &=& \frac{1}{\sqrt{T}} \sum_{t=1}^T [\mathbf{B} -(1-L)\widetilde{\mathbf{B}}(L)] \pmb{\varepsilon}_{t} \notag \\
&=&\frac{1}{\sqrt{T}}\sum_{t=1}^T\mathbf{B} \pmb{\varepsilon}_{t} - \frac{1}{\sqrt{T}}\sum_{t=1}^T\widetilde{\mathbf{B}}(L)\pmb{\varepsilon}_{t}+ \frac{1}{\sqrt{T}}\sum_{t=1}^T \widetilde{\mathbf{B}}(L)\pmb{\varepsilon}_{t-1}\notag \\
&=&\frac{1}{\sqrt{T}}\sum_{t=1}^T\mathbf{B} \pmb{\varepsilon}_{t}-\frac{1}{\sqrt{T}}\widetilde{\mathbf{B}}(L)\pmb{\varepsilon}_{T}+\frac{1}{\sqrt{T}}\widetilde{\mathbf{B}}(L)\pmb{\varepsilon}_{0}.
\end{eqnarray*}

Note that for any $\epsilon>0$, we have

\begin{eqnarray*}
&&\Pr\left(\left|\frac{1}{\sqrt{T}}\sum_{t=1}^{T}(\mathbf{x}_{t}-\bm{\mu})\right|_{\infty} \leq u \right) \notag \\
&\leq& \Pr\left(\left|\frac{1}{\sqrt{T}}\widetilde{\mathbf{B}}(L)\pmb{\varepsilon}_{T}-\frac{1}{\sqrt{T}}\widetilde{\mathbf{B}}(L)\pmb{\varepsilon}_{0}\right|_{\infty}\geq \epsilon \right)+\Pr\left(\left|\frac{1}{\sqrt{T}}\sum_{t=1}^{T}\mathbf{B} \pmb{\varepsilon}_{t}\right|_{\infty} \leq u + \epsilon\right)
\end{eqnarray*}
and

\begin{eqnarray*}
&&\Pr\left(\left|\frac{1}{\sqrt{T}}\sum_{t=1}^{T}\mathbf{z} _t\right|_{\infty} \leq u\right) \notag \\
&=&\Pr\left(\left|\frac{1}{\sqrt{T}}\sum_{t=1}^{T}\mathbf{z} _t\right|_{\infty} \leq u+\epsilon\right)-\Pr\left(u<\left|\frac{1}{\sqrt{T}}\sum_{t=1}^{T}\mathbf{z} _t\right|_{\infty} \leq u + \epsilon\right).
\end{eqnarray*}

Hence, we have
\begin{eqnarray*}
&&\sup_{u\in \mathbb{R}}\left|\Pr\left(\left|\frac{1}{\sqrt{T}}\sum_{t=1}^{T}(\mathbf{x} _t-\bm{\mu})\right|_{\infty} \leq u\right) - \Pr\left(\left|\frac{1}{\sqrt{T}}\sum_{t=1}^{T}\mathbf{z} _t\right|_{\infty} \leq u \right)\right| \notag \\
&\leq&  \Pr\left(\left|\frac{1}{\sqrt{T}}\widetilde{\mathbf{B}}(L)\pmb{\varepsilon}_{T}-\frac{1}{\sqrt{T}}\widetilde{\mathbf{B}}(L)\pmb{\varepsilon}_{0}\right|_{\infty}\geq \epsilon \right) \notag \\
&&+ \sup_{u\in \mathbb{R}}\left|\Pr\left(\left|\frac{1}{\sqrt{T}}\sum_{t=1}^{T}\mathbf{B} \pmb{\varepsilon}_{t}\right|_{\infty} \leq u\right) - \Pr\left(\left|\frac{1}{\sqrt{T}}\sum_{t=1}^{T}\mathbf{z} _t\right|_{\infty} \leq u \right)\right| \notag\\
&& + \sup_{u\in \mathbb{R}}\left| \Pr\left(\left|\left|\frac{1}{\sqrt{T}}\sum_{t=1}^{T}\mathbf{z} _t\right|_{\infty}-u\right| \leq \epsilon \right)\right| \eqqcolon I_1+I_2+I_3.
\end{eqnarray*}

Consider $I_1$ first. Given $E[\varepsilon_{it}^J] < \infty$ and $\max_{i}\sum_{\ell=1}^{\infty}\ell \|\mathbf{b}_{\ell i}^{\sharp}\|_2<\infty$, we next show that each element in $\widetilde{\mathbf{B}}(L)\pmb{\varepsilon}_{t}$ has the finite $J^{th}$ moment to imply

\[ 
E [ |\widetilde{\mathbf{B}}(L)\pmb{\varepsilon}_{t} |_\infty^J ] = O(N).
\] 

Let $\widetilde{\mathbf{B}}_{\ell} =  (\widetilde{b}_{\ell,ij} )_{1\leq i,j\leq N}$. Then, for the $i^{th}$ element in $\widetilde{\mathbf{B}}(L)\pmb{\varepsilon}_{t}$, we have

\begin{eqnarray*}
&&\left(E(\sum_{\ell=0}^{\infty}\sum_{j=1}^{N}\widetilde{b}_{\ell,ij}\varepsilon_{j,t-\ell})^{J}\right)^{1/J} \leq \sum_{\ell=0}^{\infty}\left(E(\sum_{j=1}^{N}\widetilde{b}_{\ell,ij}\varepsilon_{j,t-\ell})^{J}\right)^{1/J} \notag\\
&\leq& \sum_{\ell=0}^{\infty}\frac{14.5J}{\log J}\left[\left(\sum_{j=1}^{N}E(\widetilde{b}_{\ell,ij}\varepsilon_{j,t-\ell})^J \right)^{1/J} + \left(\sum_{j=1}^{N}E(\widetilde{b}_{\ell,ij}\varepsilon_{j,t-\ell})^2 \right)^{1/2}  \right] \notag\\
&\leq& \frac{29J(E[\varepsilon_{it}^J])^{1/J}}{\log J} \sum_{\ell=0}^{\infty}\ell\|\mathbf{b}_{\ell i}^{\sharp}\|_2 =O(1),
\end{eqnarray*}
where the first inequality follows from the triangle inequality, and the second inequality follows from the Rosenthal inequality for independent variables (e.g., \citealp{Johnson}).

Then choose $\epsilon = \sqrt{\frac{N^{2/J}}{T^{1-2/J}}}$ and by using the Markov inequality, we have

$$
I_1 = O\left(\frac{N/T^{J/2}}{N/T^{J/2-1}} \right) =O(T^{-1}).
$$

Consider $I_2$. Similarly, we can show that each element in $\mathbf{B}\pmb{\varepsilon}_{t}$ has bounded the $J^{th}$ moment and thus $E[|\mathbf{B}\pmb{\varepsilon}_{t}|_\infty^J] = O(N)$. Then by using high-dimensional Gaussian approximations for independent random vectors (cf., Theorem 2.5 of \citealp{chernozhuokov2022improved}), we have  

\begin{eqnarray*}
&&\sup_{u\in \mathbb{R}}\left|\Pr\left(\left|\frac{1}{\sqrt{T}}\sum_{t=1}^{T}\mathbf{B}\pmb{\varepsilon}_{t}\right|_{\infty} \leq u\right) - \Pr\left(\left|\frac{1}{\sqrt{T}}\sum_{t=1}^{T}\mathbf{z} _t\right|_{\infty} \leq u\right)\right| \notag \\
&=&O\left( \left(\frac{ N^{2/J}(\log N)^5}{T}\right)^{1/4}+\sqrt{\frac{N^{2/J}(\log N)^{3-2/J}}{T^{1-2/J}}}\right).
\end{eqnarray*}

Consider $I_3$. Using Lemma A.1 in \cite{chernozhukov2017central}, we have
$$
\sup_{u\in \mathbb{R}}\left| \Pr\left(\left|\left|\frac{1}{\sqrt{T}}\sum_{t=1}^{T}\mathbf{z} _t\right|_{\infty}-u\right| \leq \epsilon \right)\right| = O(\epsilon \sqrt{\log N}) = o \left(\sqrt{\frac{N^{2/J}(\log N)^{3-2/J}}{T^{1-2/J}}}\right)
$$
when choosing $\epsilon = \sqrt{\frac{N^{2/J}}{T^{1-2/J}}}$. The proof is now completed.
\end{proof}

\medskip

\begin{proof}[Proof of Theorem \ref{THM.6}]
\item 

To complete this theorem, it is sufficient to prove 

\[
|\widehat{\bm{\Omega}}-\bm{\Omega}|_{\max} = O_P(\sqrt{\widetilde{m}\log N/T}).
\] 
Then Theorem \ref{THM.6} follows from Theorem \ref{THM.5} and Proposition 2.1 in \cite{chernozhuokov2022improved}.

Define $\mathbf{y}_t = (y_{1t},\ldots,y_{Nt})^\top = \sum_{\ell=0}^{\infty}\mathbf{B}_{\ell}\bm{\varepsilon}_{t-\ell}$. Let $y_{it}^*$ be the coupled version of $y_{it}$ with $\bm{\varepsilon}_0^*$ replacing $\bm{\varepsilon}_0$ and $\mathbf{B}_\ell = \left(b_{\ell,ij} \right)_{1\leq i,j\leq N}$. By using Rosenthal inequality for independent variables, we have

\begin{eqnarray*}
&&\left(E(y_{it}y_{jt}-y_{it}^*y_{jt}^*)^J\right)^{1/J} \leq  O(1)\max_{i} \left(E(y_{it}-y_{it}^*)^J\right)^{1/J}\notag\\
&\leq&O(1)\max_i\left(E(\sum_{j=1}^{N}b_{t,ij}\varepsilon_{j,0})^{J}\right)^{1/J} \notag\\
&\leq& O(1)\max_i\left[\left(\sum_{j=1}^{N}E(b_{t,ij}\varepsilon_{j0})^J \right)^{1/J} + \left(\sum_{j=1}^{N}E(b_{t,ij}\varepsilon_{j0})^2 \right)^{1/2}  \right] \notag\\
&\leq& O(1)\max_i\|\mathbf{b}_{t i}^{\sharp}\|_2.
\end{eqnarray*}

Then by using Lemma A.8 (1) of \cite{gao2024robust}, if $\frac{N^2T\log T}{(T\widetilde{m}\log N)^{J/4}}\to 0$, we have
$$
\max_{1\leq i,j\leq N}\left|\frac{1}{T}\sum_{t,s=1}^{T}a\left(\frac{t-s}{\widetilde{m}}\right)(y_{it}y_{js}-E(y_{it}y_{js})) \right| = O_P\left(\sqrt{\widetilde{m}\log N/T}\right).
$$
Let $\bm{\Omega} = \{\omega_{ij}\}_{1\leq i,j\leq N}$. By using standard arguments for the bias term (e.g., the proof of Theorem 2.2 in \citealp{gao2024robust}), we have 
\[
\left|\frac{1}{T}\sum_{t,s=1}^{T}a\left(\frac{t-s}{\widetilde{m}}\right)E(y_{it}y_{js}) - \omega_{ij} \right| = O(\widetilde{m}^{-q_{\alpha}}).
\]

Then, to complete the proof, it is sufficient to show 
\begin{eqnarray*}
\max_{1\leq i,j\leq N}\left|\frac{1}{T}\sum_{t,s=1}^{T}a\left(\frac{t-s}{\widetilde{m}}\right)((x_{it}-\mu_i)(x_{js}-\mu_j)-(x_{it}-\overline{x}_i)(x_{js}-\overline{x}_j)) \right| = O_P\left(\sqrt{\widetilde{m}\log N/T}\right).
\end{eqnarray*}

By Lemma 1 (1) in \cite{gao2024robust} we have
\begin{eqnarray*}
&&\max_{1\leq i,j\leq N}\left|\frac{1}{T}\sum_{t,s=1}^{T}a\left(\frac{t-s}{\widetilde{m}}\right)((x_{it}-\mu_i)(x_{js}-\mu_j)-(x_{it}-\overline{x}_i)(x_{js}-\overline{x}_j))\right|\notag\\
&\leq&2\max_{1\leq j\leq N}\left|\frac{1}{T}\sum_{t,s=1}^{T}a\left(\frac{t-s}{\widetilde{m}}\right)y_{js}\right|\max_{1\leq i\leq N}\left|\mu_i-\overline{x}_i\right| + \max_{1\leq i\leq N}(\mu_i-\overline{x}_i)^2\left|\frac{1}{T}\sum_{t,s=1}^{T}a\left(\frac{t-s}{\widetilde{m}}\right) \right| \notag\\
&=& O_P\left(\widetilde{m}\log N/T\right). \notag
\end{eqnarray*}

The proof is now completed.
\end{proof}

\medskip

\begin{proof}[Proof of Proposition \ref{LM.group}]
\item     

(1). Using BN decomposition, we obtain that

\begin{eqnarray*}
    \max_i|\overline{x}_i -\mu_i|&=& \frac{1}{T}\left|\mathbf{e}_i^\top \mathbf{B} \sum_{s=1}^T\pmb{\varepsilon}_{s}-\mathbf{e}_i^\top\widetilde{\mathbf{B}}(L)\pmb{\varepsilon}_{T}+\mathbf{e}_i^\top\widetilde{\mathbf{B}}(L)\pmb{\varepsilon}_{0} \right|.
\end{eqnarray*}
By Corollary 2.1 of \cite{JSW_2023}, we can obtain that

\begin{eqnarray*}
    \Pr\left(\max_i\left|\frac{1}{\sqrt{\sum_{s=1}^T E[\zeta_{is}^2]}} \sum_{s=1}^T\zeta_{is}\right| \ge \epsilon\right) &\le &\sum_{i=1}^N \Pr\left(\left|\frac{1}{\sqrt{\sum_{s=1}^T E[\zeta_{is}^2]}} \sum_{s=1}^T\zeta_{is}\right|\ge \epsilon\right)\notag \\
    &\le & 2N \exp(-\epsilon^2/4)= 2N \exp(\log(1/N^c))\notag \\
    &=&2N/N^c\to 0,
\end{eqnarray*}
where $\zeta_{is}\coloneqq \mathbf{e}_i^\top \mathbf{B} \pmb{\varepsilon}_{s}$, and $\epsilon=2\sqrt{\log(N^c)}$ with $c>1$. After neglecting the residuals from the BN decomposition and using Assumption \ref{AS.group}.1, the first result is proved.

\medskip

(2). Let $\overline{\pmb{\mu}} =(\mu_1,\ldots, \mu_{J_0})$. Since $J_0$ is known, we immediately obtain that 

\begin{eqnarray*}
     0\le S(\widehat{\mathcal{G}}, \widehat{\pmb{\nu}})\le S(\mathscr{G}, \overline{\pmb{\mu}})= O_P( \log(N)/T  +c_{N}^2),
\end{eqnarray*}
where the last step is due to Assumption \ref{AS.group}.1 and the first result of this lemma. If one individual is allocated in a wrong group, then we can always show that $S(\mathscr{G}, \overline{\pmb{\mu}})$ is a better option. Therefore, we conclude that $\Pr(\widehat{\mathcal{G}}_j=\mathscr{G}_j)\to 1$ for all $j\in [J_0]$.

\medskip

(3). Without loss of generality, we consider two cases: (i) $J=J_0-1$ and (ii) $J=J_0+1$. For case (i) (i.e.,$J= J_0-1$), we write  

\begin{eqnarray*}
    &&S(\widehat{\mathcal{G}}_{\mid J}, \widehat{\pmb{\nu}}_{\mid J})+\rho_{NT}J \notag \\
    &= &\frac{1}{N}\sum_{j=1}^J\sum_{i\in \widehat{\mathcal{G}}_{j\mid J}}| \overline{x}_i- \sum_{j}\overline{\mu}_j I(i\in \mathscr{G}_j)+\sum_{j}\overline{\mu}_j I(i\in \mathscr{G}_j) - \widehat{\nu}_{j\mid J}|^2+\rho_{NT}J  \notag \\
    &\ge & \frac{1}{N}\sum_{j=1}^J\sum_{i\in \widehat{\mathcal{G}}_{j\mid J}}| \sum_{j}\overline{\mu}_j I(i\in \mathscr{G}_j) - \widehat{\nu}_{j\mid J}|^2 -o_P (1)\notag \\
    &\ge &c^*-o_P (1)\notag \\
    &> &S(\mathscr{G}, \overline{\pmb{\mu}})+\rho_{NT}J_0,
\end{eqnarray*}
where the first inequality follows from the first result of this lemma, and we must be able to find a $c^*$ using Assumption \ref{AS.group}.2 and $J=J_0-1$.

For case (ii) (i.e.,$J= J_0+1$), write

\begin{eqnarray*}
    &&S(\widehat{\mathcal{G}}_{\mid J}, \widehat{\pmb{\nu}}_{\mid J})+\rho_{NT}(J_0+1)\notag \\
    &=& \frac{1}{N}\sum_{j=1}^J\sum_{i\in \widehat{\mathcal{G}}_{j\mid J}}| \sum_{j}\overline{\mu}_j I(i\in \mathscr{G}_j)-\widehat{\nu}_{j\mid J} |^2+\rho_{NT} (J_0+1)+O_P(\sqrt{\log(N)/T} +c_{N})\notag \\
    &> &S(\mathscr{G}, \overline{\pmb{\mu}})+\rho_{NT}J_0,
\end{eqnarray*}
where the last step follows in view of the fact that $\rho_{NT}/(\sqrt{\log(N)/T} +c_{N})\to \infty$.

The proof is now completed.
\end{proof}

\medskip

\begin{proof}[Proof of Proposition \ref{P2}]
\item    

Write

\begin{eqnarray*}
&&\frac{1}{NT^2}\sum_{t=1}^T \mathbf{y}_t ^{\top} \mathbf{y}_t \notag \\
&=&\frac{1}{NT^2}\sum_{t=1}^T \left( \mathbf{B}  \sum_{s=1}^t\pmb{\varepsilon}_{s} -\widetilde{\mathbf{B}} (L)\pmb{\varepsilon}_{t} +\widetilde{\mathbf{B}} (L)\pmb{\varepsilon}_{0} \right)^\top \left( \mathbf{B}  \sum_{s=1}^t\pmb{\varepsilon}_{s} -\widetilde{\mathbf{B}} (L)\pmb{\varepsilon}_{t} +\widetilde{\mathbf{B}} (L)\pmb{\varepsilon}_{0} \right) ,
\end{eqnarray*}
where the leading term is obviously $\frac{1}{NT^2}\sum_{t=1}^T \sum_{s_1=1}^t \sum_{s_2=1}^t \pmb{\varepsilon}_{s_1} ^{\top}\mathbf{B} ^{\top}\mathbf{B} \pmb{\varepsilon}_{s_2}  $.

\medskip

We now focus on the leading term $\frac{1}{NT^2}\sum_{t=1}^T\sum_{s_1=1}^t\sum_{s_2=1}^t\pmb{\varepsilon}_{s_1} ^{\top}\mathbf{B} ^{\top}\mathbf{B} \pmb{\varepsilon}_{s_2} $ in what follows. Firstly, write

\begin{eqnarray*}
&&\frac{1}{NT^2}\sum_{t=1}^TE  \left[\sum_{s_1=1}^t\pmb{\varepsilon}_{s_1} ^{\top}\mathbf{B} ^{\top}\mathbf{B} \sum_{s_2=1}^t\pmb{\varepsilon}_{s_2} \right] = \frac{1}{NT^2}\sum_{t=1}^T\sum_{s=1}^tE  \left[\pmb{\varepsilon}_{s} ^{\top}\mathbf{B} ^{\top}\mathbf{B}  \pmb{\varepsilon}_{s} \right]\notag \\
&=&\frac{1}{NT^2}\sum_{t=1}^T\sum_{s=1}^t \text{vec}(\mathbf{B} ^{\top}\mathbf{B}  )^\top E [\pmb{\varepsilon}_{s} \otimes\pmb{\varepsilon}_{s} ] = \frac{1}{NT^2}\sum_{t=1}^T\sum_{s=1}^t \text{vec}(\mathbf{B} ^{\top}\mathbf{B}  )^\top (\mathbf{e}_1^\top,\ldots, \mathbf{e}_N^\top)^\top \notag  \\
&= &\frac{1}{NT^2}\sum_{t=1}^T\sum_{s=1}^t \|\mathbf{B} \|^2 = \left(\int_0^1x \mathrm{d}x+ O\left(\frac{1}{T}\right) \right) \frac{1}{N}\|\mathbf{B} \|^2 \to \frac{b }{2},
\end{eqnarray*}
where the second equality follows from the vectorization operation, the fifth equality follows from the definition of Riemann integral, and the last step follows  from the condition $\lim_{N}\frac{1}{N}\|\mathbf{B} \|^2\to b$.

We further note a few facts:

\begin{enumerate}[leftmargin=24pt, parsep=2pt, topsep=2pt]
\item $E\big[(\pmb{\varepsilon}_{s_1} ^{\top}\mathbf{B} ^{\top}\mathbf{B} \pmb{\varepsilon}_{s_2} - E[\pmb{\varepsilon}_{s_1} ^{\top}\mathbf{B} ^{\top}\mathbf{B} \pmb{\varepsilon}_{s_2} ]) (\pmb{\varepsilon}_{s_3} ^{\top}\mathbf{B} ^{\top}\mathbf{B} \pmb{\varepsilon}_{s_4} - E[\pmb{\varepsilon}_{s_3} ^{\top}\mathbf{B} ^{\top}\mathbf{B} \pmb{\varepsilon}_{s_4} ])\big]=0$ when 

\begin{enumerate}[leftmargin=24pt, parsep=2pt, topsep=2pt]
  \item three or four of $s_1, s_2, s_3, s_4$ are mutually different;
  \item $s_1=s_2$, and $s_3\ne s_1$ (or $s_4\ne s_1$);
  \item $s_1\ne s_2$, and $s_4=s_3=s_1$.
\end{enumerate}

\item for $s_1\ne s_2$, $E[\pmb{\varepsilon}_{s_1} ^{\top}\mathbf{B} ^{\top}\mathbf{B} \pmb{\varepsilon}_{s_2} ]^2 =E[\pmb{\varepsilon}_{s_1} ^{\top} \mathbf{B} ^{\top}\mathbf{B} \mathbf{B} ^{\top}\mathbf{B} \pmb{\varepsilon}_{s_1}^{u }] \le  \|\mathbf{B} \|_2^4 N =O(N)$.

\item Note that 

\begin{eqnarray*}
&&E|\pmb{\varepsilon}_{s} ^{\top}\mathbf{B} ^{\top}\mathbf{B} \pmb{\varepsilon}_{s} |^2 = \text{vec}(\mathbf{B} ^{\top}\mathbf{B}  )^\top E [(\pmb{\varepsilon}_{s} \otimes\pmb{\varepsilon}_{s}) (\pmb{\varepsilon}_{s}^{\top}\otimes\pmb{\varepsilon}_{s}^{\top})] \text{vec}(\mathbf{B} ^{\top}\mathbf{B}  )\notag \\
&\le & O(N)\text{vec}(\mathbf{B} ^{\top}\mathbf{B}  )^\top \text{vec}(\mathbf{B} ^{\top}\mathbf{B}) = O(N)\tr(\mathbf{B} ^{\top}\mathbf{B} \mathbf{B} ^{\top}\mathbf{B} )\notag \\
&\le &O(N) \text{vec}(\mathbf{B} ^{\top})^\top (\mathbf{B} \otimes \mathbf{B} ^{\top}) \text{vec}(\mathbf{B} ^{\top}) \le O(N)\|\text{vec}(\mathbf{B} ^{\top})\|^2 = O(N^2),
\end{eqnarray*}
where the first inequality follows from Lemma \ref{LM.A6}, the second equality follows from the fact that $\tr(\mathbf{C}^\top\mathbf{D}) = \text{vec}(\mathbf{C})^\top \text{vec}(\mathbf{D})$ for $\forall \mathbf{C},\mathbf{D}\in \mathbb{R}^{N\times N}$, the second inequality follows from the fact that $\tr(\mathbf{A}_1\mathbf{A}_2\mathbf{A}_3\mathbf{A}_4) =\text{vec}(\mathbf{A}_1)^\top(\mathbf{A}_2\otimes \mathbf{A}_4^\top )\text{vec}(\mathbf{A}_3^\top)$ for any conformable matrices $\mathbf{A}_1,\mathbf{A}_2,\mathbf{A}_3,\mathbf{A}_4$ (\citealp[p. 253]{Bernstein}), and the last step follows from the condition $\lim_{N}\frac{1}{N}\|\mathbf{B} \|^2\to b$.
\end{enumerate}

We then write
\begin{eqnarray*}
&&E\left[\frac{1}{NT^2}\sum_{t=1}^T\sum_{s_1=1}^t\sum_{s_2=1}^t(\pmb{\varepsilon}_{s_1} ^{\top}\mathbf{B} ^{\top}\mathbf{B} \pmb{\varepsilon}_{s_2} - E[\pmb{\varepsilon}_{s_1} ^{\top}\mathbf{B} ^{\top}\mathbf{B} \pmb{\varepsilon}_{s_2} ])\right]^2\notag \\
&=&\frac{1}{N^2T^4}\sum_{t_1=1}^T\sum_{t_2=1}^T\sum_{s_1=1}^{t_1}\sum_{s_2=1}^{t_1}\sum_{s_3=1}^{t_2}\sum_{s_4=1}^{t_2}E\big[(\pmb{\varepsilon}_{s_1} ^{\top}\mathbf{B} ^{\top}\mathbf{B} \pmb{\varepsilon}_{s_2} - E[\pmb{\varepsilon}_{s_1} ^{\top}\mathbf{B} ^{\top}\mathbf{B} \pmb{\varepsilon}_{s_2} ]) \notag \\
&&\cdot (\pmb{\varepsilon}_{s_3} ^{\top}\mathbf{B} ^{\top}\mathbf{B} \pmb{\varepsilon}_{s_4} - E[\pmb{\varepsilon}_{s_3} ^{\top}\mathbf{B} ^{\top}\mathbf{B} \pmb{\varepsilon}_{s_4} ])\big] \notag \\
&=&\frac{1}{N^2T^4}\sum_{t_1=1}^T\sum_{t_2=1}^T\sum_{s_1=1}^{t_1} E|\pmb{\varepsilon}_{s_1} ^{\top}\mathbf{B} ^{\top}\mathbf{B} \pmb{\varepsilon}_{s_1} - E[\pmb{\varepsilon}_{s_1} ^{\top}\mathbf{B} ^{\top}\mathbf{B} \pmb{\varepsilon}_{s_1} ]|^2 \notag \\
&&+\frac{1}{N^2T^4}\sum_{t_1=1}^T\sum_{t_2=1}^T\sum_{ s_1,s_2=1, s_1\ne s_2}^{t_1} E|\pmb{\varepsilon}_{s_1} ^{\top}\mathbf{B} ^{\top}\mathbf{B} \pmb{\varepsilon}_{s_2} - E[\pmb{\varepsilon}_{s_1} ^{\top}\mathbf{B} ^{\top}\mathbf{B} \pmb{\varepsilon}_{s_2} ]|^2\notag \\
&\le &O(1)\frac{1}{N^2T}(E|\pmb{\varepsilon}_{s_1} ^{\top}\mathbf{B} ^{\top}\mathbf{B} \pmb{\varepsilon}_{s_1} |^2+E|\pmb{\varepsilon}_{s_1} ^{\top}\mathbf{B} ^{\top}\mathbf{B} \pmb{\varepsilon}_{s_2} |^2) = O(1)\frac{1}{T},
\end{eqnarray*}
where the last step follows from the above facts.

Putting everything together, we immediately obtain that $\frac{1}{NT^2}\sum_{t=1}^T \mathbf{y}_t ^{\top} \mathbf{y}_t \to_P\frac{b }{2}.$
\end{proof}

\medskip

\begin{proof}[Proof of Proposition \ref{P1}]

\item 
    
In what follows, we focus on $\widehat{\pmb{\theta}}_i$, and write

\begin{eqnarray*}
\widehat{\pmb{\theta}}_i-\pmb{\theta}_i &=& (\mathbf{W}_i^\top \mathbf{M}_{\overline{\mathbf{Z}}} \mathbf{W}_i)^{-1} \mathbf{W}_i^\top \mathbf{M}_{\overline{\mathbf{Z}}}\mathbf{F}\pmb{\gamma}_i  +  (\mathbf{W}_i^\top \mathbf{M}_{\overline{\mathbf{Z}}} \mathbf{W}_i)^{-1} \mathbf{W}_i^\top \mathbf{M}_{\overline{\mathbf{Z}}} \pmb{\epsilon}_i.
\end{eqnarray*}

We study $\mathbf{M}_{\overline{\mathbf{Z}}}$ first.

\begin{eqnarray*}
\mathbf{M}_{\overline{\mathbf{Z}}} &=&\mathbf{I}_T-(\mathbf{F}\overline{\mathbf{C}}+\overline{\mathbf{U}})[(\mathbf{F}\overline{\mathbf{C}}+\overline{\mathbf{U}})^\top (\mathbf{F}\overline{\mathbf{C}}+\overline{\mathbf{U}})]^+ (\mathbf{F}\overline{\mathbf{C}}+\overline{\mathbf{U}})^\top.
\end{eqnarray*}
We now investigate $(\mathbf{F}\overline{\mathbf{C}}+\overline{\mathbf{U}})^\top (\mathbf{F}\overline{\mathbf{C}}+\overline{\mathbf{U}})$, and write

\begin{eqnarray*}
(\mathbf{F}\overline{\mathbf{C}}+\overline{\mathbf{U}})^\top (\mathbf{F}\overline{\mathbf{C}}+\overline{\mathbf{U}})=\overline{\mathbf{C}}^\top \mathbf{F}^\top \mathbf{F} \overline{\mathbf{C}}+\overline{\mathbf{U}}^\top \overline{\mathbf{U}} + \overline{\mathbf{C}}^\top \mathbf{F}^\top \overline{\mathbf{U}}+ \overline{\mathbf{U}}^\top  \mathbf{F} \overline{\mathbf{C}}.
\end{eqnarray*}
By Assumption \ref{AS5}.1,

\begin{eqnarray*}
\frac{1}{T} (\mathbf{F}\overline{\mathbf{C}}+\overline{\mathbf{U}})^\top (\mathbf{F}\overline{\mathbf{C}}+\overline{\mathbf{U}}) = \frac{1}{T}\overline{\mathbf{C}}^\top \mathbf{F}^\top \mathbf{F} \overline{\mathbf{C}}+O_P\left(\frac{1}{N}+\frac{1}{\sqrt{NT}}\right).
\end{eqnarray*}

Also, note that by Assumptions \ref{AS5}.1

\begin{eqnarray*}
\left( \overline{\mathbf{C}}^\top \frac{\mathbf{F}^\top \mathbf{F} }{T}\overline{\mathbf{C}}\right)^+ = \overline{\mathbf{C}}^+ \left(\frac{\mathbf{F}^\top \mathbf{F} }{T}\right)^+ (\overline{\mathbf{C}}^\top)^+ ,
\end{eqnarray*}
where $\overline{\mathbf{C}}^+ =\overline{\mathbf{C}}^\top(\overline{\mathbf{C}}\overline{\mathbf{C}}^\top)^{-1}$ and $(\overline{\mathbf{C}}^\top)^+ =(\overline{\mathbf{C}}\overline{\mathbf{C}}^\top)^{-1}\overline{\mathbf{C}}$.
Thus, similar to (40) of \cite{Pesaran2006}, it is straightforward to obtain that

\begin{eqnarray*}
\frac{1}{T} \mathbf{W}_i^\top \mathbf{M}_{\overline{\mathbf{Z}}}\mathbf{F}\pmb{\gamma}_i =O_P\left(\frac{1}{N}+\frac{1}{\sqrt{NT}}\right).
\end{eqnarray*}


We then focus on $\frac{1}{T}\mathbf{W}_i^\top \mathbf{M}_{\overline{\mathbf{Z}}} \pmb{\epsilon}_i$, and write

\begin{eqnarray*}
\frac{1}{T}\mathbf{W}_i^\top \mathbf{M}_{\overline{\mathbf{Z}}} \pmb{\epsilon}_i &=&\frac{1}{T}\mathbf{W}_i^\top \mathbf{M}_{\mathbf{F}} \pmb{\epsilon}_i + O_P\left(\frac{1}{N}+\frac{1}{\sqrt{NT}}\right)\notag \\
&=&\frac{1}{T}\mathbf{V}_i^\top  \pmb{\epsilon}_i + O_P\left(\frac{1}{N}+\frac{1}{T}+\frac{1}{\sqrt{NT}}\right),
\end{eqnarray*}
where the first equality follows from (45) of \cite{Pesaran2006}, and the second equality follows from Assumption \ref{AS5}.1. 

Note further that Assumption \ref{AS5}.2 ensures that $\mathbf{v}_{it}\epsilon_{it}$ admits an MA($\infty$) process via the second order BN decomposition. See \eqref{def.BN2} and \citet[p. 978]{PS1992} for example. Thus, by the proofs of Lemma \ref{LM.A6}, it is easy to know that $\frac{1}{T}\mathbf{V}_i^\top  \pmb{\epsilon}_i$ can be further decomposed via the second order BN decomposition. Also, we note that $\widehat{\pmb{\theta}}-\pmb{\theta} =O_P(\frac{1}{\sqrt{NT}})$ under the null and the condition $N\asymp T$ (see \citealp{Westerlund2018} for detailed development). Finally, putting everything together, we invoke Theorem \ref{THM.6} and Assumption \ref{AS5}.2, and the result follows immediately.
\end{proof}

\section{Proofs of the Preliminary Lemmas}\label{SecA4}

\begin{proof}[Proof of Lemma \ref{LM.A4}]
\item 

First, we note that the Probabilist's Hermite polynomials 

\[
\left\{H_n(x)\mid (-1)^n \exp\left(\frac{x^2}{2}\right) \frac{\mathrm{d}^n}{\mathrm{d}x^n}\widetilde{\phi}(x)\ \text{ for }\ n \ge 0\right\}
\]
have the following generating function:

\[
\exp\left( xu-\frac{u^2}{2}\right)  = \sum_{n=0}^\infty H_n(x)\frac{u^n}{n!}.
\]
Thus,

\begin{eqnarray}\label{def.HMP}
\sum_{n=0}^\infty\phi(x) H_n(x)\frac{u^n}{n!}&=&\frac{1}{\sqrt{2\pi}}\exp\left(-\frac{x^2}{2}+xu-\frac{u^2}{2}\right) \notag \\
&=& \frac{1}{\sqrt{2\pi}}\exp\left(-\frac{(x-u)^2}{2}\right) .
\end{eqnarray}

We now apply the Fourier transformation to both sides of \eqref{def.HMP}, and note that the right hand side can be further written as follows:
\begin{eqnarray*}
&& \int_{\mathbb{R}}\exp(\mathsf{i}w)\cdot\frac{1}{\sqrt{2\pi}}\exp\left(-\frac{(x-u)^2}{2}\right) \mathrm{d}x = \exp\left(\mathsf{i}w u-\frac{1}{2}w^2\right)\notag \\
&& = \exp(\mathsf{i}w u)\widetilde{\phi}(w) = \sum_{n=0}^\infty\widetilde{\phi}(w) \frac{(\mathsf{i}wu)^n}{n!} =  \sum_{n=0}^\infty\widetilde{\phi}(w)(\mathsf{i}w)^n\cdot \frac{u^n}{n!},
\end{eqnarray*}
where the first equality is obvious in view of the characteristic function of the normal distribution, and the third equality follows from \eqref{def.expu}. 

By comparing the Fourier transformation of both sides of \eqref{def.HMP}, the result follows immediately.
\end{proof}

\medskip

\begin{proof}[Proof of Lemma \ref{LM.A5}]
\item

(1).a The expression $\mathbf{B}(L) = \mathbf{B} -(1-L)\widetilde{\mathbf{B}}(L)$ of \eqref{def.BN1} is the so-called BN decomposition from \cite{PS1992}. We then write

\begin{eqnarray*}
\sum_{s=1}^t \mathbf{x}_{s} &=&  \sum_{s=1}^t [ \mathbf{B} -(1-L)\widetilde{\mathbf{B}}(L)] \pmb{\varepsilon}_{s} \notag \\
&=&\mathbf{B} \sum_{s=1}^t\pmb{\varepsilon}_{s} - \sum_{s=1}^t\widetilde{\mathbf{B}}(L)\pmb{\varepsilon}_{s}+ \sum_{s=1}^t \widetilde{\mathbf{B}}(L)\pmb{\varepsilon}_{s-1}\notag \\
&=&\mathbf{B} \sum_{s=1}^t\pmb{\varepsilon}_{s}-\widetilde{\mathbf{B}}(L)\pmb{\varepsilon}_{t}+\widetilde{\mathbf{B}}(L)\pmb{\varepsilon}_{0}.
\end{eqnarray*}

By Assumption \ref{AS1}.2,

\begin{eqnarray*} 
\sum_{\ell=0}^{\infty}\sqrt{\frac{N}{L_N}}\|\widetilde{\mathbf{B}}_\ell\|_2\le \sum_{\ell=0}^{\infty}\sqrt{\frac{N}{L_N}}\sum_{k=\ell+1}^{\infty}\| \mathbf{B}_k\|_2=\sum_{\ell=1}^{\infty}\ell \sqrt{\frac{N}{L_N}}\| \mathbf{B}_\ell\|_2=\sum_{\ell=1}^{\infty}\ell\cdot C_{N\ell}<\infty.
\end{eqnarray*}

\medskip

\noindent (1).b According to (1).a, write

\begin{eqnarray*}
\sum_{s=1}^t \mathbf{x}_{s}&=&\mathbf{B} \sum_{s=1}^t\pmb{\varepsilon}_{s}-\sum_{\ell=0}^{t-1}\widetilde{\mathbf{B}}_\ell \pmb{\varepsilon}_{t-\ell}-\sum_{\ell=t}^{\infty}\widetilde{\mathbf{B}}_\ell \pmb{\varepsilon}_{t-\ell}+\sum_{\ell=0}^{\infty}\widetilde{\mathbf{B}}_\ell \pmb{\varepsilon}_{-\ell} \notag \\
&=&\mathbf{B} \sum_{\ell=1}^t\pmb{\varepsilon}_{\ell}-\sum_{\ell=1}^{t}\widetilde{\mathbf{B}}_{t-\ell} \pmb{\varepsilon}_{\ell}-\sum_{\ell=0}^{\infty}\widetilde{\mathbf{B}}_{t+\ell} \pmb{\varepsilon}_{-\ell}+\sum_{\ell=0}^{\infty}\widetilde{\mathbf{B}}_\ell \pmb{\varepsilon}_{-\ell} \notag \\
&=& \sum_{\ell=1}^t (\mathbf{B}-\widetilde{\mathbf{B}}_{t-\ell})\pmb{\varepsilon}_{\ell}-\sum_{\ell=0}^{\infty}(\widetilde{\mathbf{B}}_{t+\ell} -\widetilde{\mathbf{B}}_\ell)\pmb{\varepsilon}_{-\ell}\notag\\
&=& \sum_{\ell=1}^t (\mathbf{B}-\widetilde{\mathbf{B}}_{t-\ell})\pmb{\varepsilon}_{\ell}-\sum_{\ell=-\infty}^{0}(\widetilde{\mathbf{B}}_{t-\ell} -\widetilde{\mathbf{B}}_{-\ell})\pmb{\varepsilon}_{\ell} \eqqcolon  \sum_{\ell=-\infty}^t \pmb{\mathcal{B}}_{t\ell} \pmb{\varepsilon}_\ell.
\end{eqnarray*}

\medskip

(2).  Note that for a given vector $\mathbf{v}$, $|\mathbf{v}|_2\le |\mathbf{v}|_1$ in which $|\mathbf{v}|_j$ with $j=1,2$ defines its $L^j$ norm. In connection with the fact that $\sum_{\ell=0}^{\infty}\sqrt{\frac{N}{L_N}}\|\widetilde{\mathbf{B}}_\ell\|_2<\infty $ of the first result, we obtain that $\sum_{\ell=0}^{\infty} \frac{N}{L_N}\|\widetilde{\mathbf{B}}_\ell\|_2^2< \infty.$ We are now able to write

\begin{eqnarray*}
E\left|\frac{1}{\sqrt{L_N}} \sum_{\ell=-\infty}^0 \mathbf{1}_N^\top\pmb{\mathcal{B}}_{T\ell} \pmb{\varepsilon}_{\ell}\right|^2 &=&\frac{1}{L_N}\sum_{\ell=-\infty}^0\mathbf{1}_N^\top\pmb{\mathcal{B}}_{T\ell}\pmb{\mathcal{B}}_{T\ell}^\top \mathbf{1}_N\notag \\
&=&\frac{1}{L_N}\sum_{\ell=-\infty}^0\mathbf{1}_N^\top (\widetilde{\mathbf{B}}_{T-\ell} -\widetilde{\mathbf{B}}_{-\ell})(\widetilde{\mathbf{B}}_{T-\ell} -\widetilde{\mathbf{B}}_{-\ell})^\top \mathbf{1}_N\notag \\
&\le &\frac{2}{L_N}\sum_{\ell=T}^{\infty}\|\mathbf{1}_N\|_2^2 \|\widetilde{\mathbf{B}}_{\ell}\|_2^2+\frac{2}{L_N}\sum_{\ell=0}^\infty\|\mathbf{1}_N\|_2^2 \|\widetilde{\mathbf{B}}_{\ell}\|_2^2\notag \\
&\le &\frac{4N}{L_N}\sum_{\ell=0}^\infty \|\widetilde{\mathbf{B}}_{\ell}\|_2^2 <\infty.
\end{eqnarray*}

\medskip

(3).  By the first result of this lemma, we have

\begin{eqnarray*}
E\left| \frac{1}{\sqrt{L_N}}\sum_{t=1}^T \mathbf{1}_N^\top\widetilde{\mathbf{B}}_{T-\ell} \pmb{\varepsilon}_{\ell} \right|^2 &=&\frac{1}{L_N}\mathbf{1}_N^\top \left(\sum_{t=1}^T \widetilde{\mathbf{B}}_{T-\ell}\right) \left(\sum_{t=1}^T \widetilde{\mathbf{B}}_{T-\ell}\right)^\top\mathbf{1}_N\notag\\
&\le &\left(\sqrt{\frac{N}{L_N}}\sum_{\ell=0}^{T-1}\|\widetilde{\mathbf{B}}_{\ell}\|_2\right)^2 \le \left(\sqrt{\frac{N}{L_N}}\sum_{\ell=0}^{\infty}\|\widetilde{\mathbf{B}}_{\ell}\|_2\right)^2<\infty.
\end{eqnarray*}
The proof is now completed.
\end{proof}

\medskip

\begin{proof}[Proof of Lemma \ref{LM.A6}]

\item 

(1). For simplicity, we drop index $t$, so write $\pmb{\varepsilon}=(\varepsilon_1,\ldots, \varepsilon_N)^\top$. It is obvious that

\begin{eqnarray*}
&&\| E[(\pmb{\varepsilon} \otimes \pmb{\varepsilon} )(\pmb{\varepsilon} ^\top \otimes \pmb{\varepsilon} ^\top)]\|_2\le \sqrt{\|E[(\pmb{\varepsilon} \otimes \pmb{\varepsilon} )(\pmb{\varepsilon} ^\top \otimes \pmb{\varepsilon} ^\top)]\|_1\|E[(\pmb{\varepsilon} \otimes \pmb{\varepsilon} )(\pmb{\varepsilon} ^\top \otimes \pmb{\varepsilon} ^\top)]\|_\infty}=O(N).
\end{eqnarray*}

\medskip

(2). Write

\begin{eqnarray*}
\frac{1}{L_NT}\sum_{t=1}^T \mathbf{B}_0^*(L)\text{vec}(\pmb{\varepsilon}_{t}  \pmb{\varepsilon}_{t}^\top)&=& \frac{1}{L_NT}\sum_{t=1}^T\mathbf{B}_0^* \text{vec}(\pmb{\varepsilon}_{t}  \pmb{\varepsilon}_{t}^\top)-\frac{1}{L_NT} \widetilde{\mathbf{B}}_0^*(L) \text{vec}(\pmb{\varepsilon}_{T}  \pmb{\varepsilon}_{T}^\top)\notag \\
&&+ \frac{1}{L_NT} \widetilde{\mathbf{B}}_0^*(L) \text{vec}(\pmb{\varepsilon}_{0} \pmb{\varepsilon}_{0}^\top).
\end{eqnarray*}

Note that
\begin{eqnarray*}
&&E[\widetilde{\mathbf{B}}_0(L) \text{vec}(\pmb{\varepsilon}_{T}  \pmb{\varepsilon}_{T}^\top) \text{vec}(\pmb{\varepsilon}_{T}  \pmb{\varepsilon}_{T}^\top)^\top \widetilde{\mathbf{B}}_0(L)^\top ]\notag \\
&\le &  \lambda_{\max}(E[(\pmb{\varepsilon} \otimes \pmb{\varepsilon} )(\pmb{\varepsilon} ^\top \otimes \pmb{\varepsilon} ^\top)])\sum_{\ell=0}^{\infty}\|\widetilde{\mathbf{B}}_{0\ell}^* \|_2 \notag \\
&\le & O(N)\sum_{\ell=0}^{\infty} \sum_{k=\ell+1}^{\infty}\|(\mathbf{1}_N^\top \mathbf{B}_k)\otimes (\mathbf{1}_N^\top \mathbf{B}_{k})\|_2 \le O(1)N^2 \sum_{\ell=0}^{\infty} \ell  \|\mathbf{B}_{\ell}\|_2^2,
\end{eqnarray*}
where the second inequality follows from the first result of this lemma. Thus,
\begin{eqnarray*}
\frac{1}{L_NT} |\widetilde{\mathbf{B}}_0^*(L) \text{vec}(\pmb{\varepsilon}_{T}  \pmb{\varepsilon}_{T}^\top)|&=&O_P(1)\frac{1}{L_NT} \sqrt{N^2 \sum_{\ell=0}^{\infty} \ell  \|\mathbf{B}_{\ell}\|_2^2}   \notag \\
&\le &O_P(1)\frac{1}{\sqrt{L_N}T} \sum_{\ell=0}^{\infty} \frac{\sqrt{N^2 \ell}}{\sqrt{L_N}} \|\mathbf{B}_{\ell}\|_2 = O_P(1)\frac{1}{ T}.
\end{eqnarray*}

Similarly, we have

\begin{eqnarray*}
\frac{1}{L_NT} |\widetilde{\mathbf{B}}_0^*(L) \text{vec}(\pmb{\varepsilon}_0  \pmb{\varepsilon}_0^\top)|= O_P(1)\frac{1}{ T}.
\end{eqnarray*}

Therefore, we need only to consider $\frac{1}{L_NT}\sum_{t=1}^T\mathbf{B}_0^* \text{vec}(\pmb{\varepsilon}_{t}  \pmb{\varepsilon}_{t}^\top)$ in what follows.  By the proof similar to Proposition \ref{P2}, we obtain that 

\begin{eqnarray*}
\left|\frac{1}{L_NT}\sum_{t=1}^T \mathbf{B}_0^*(L)\text{vec}(\pmb{\varepsilon}_{t}  \pmb{\varepsilon}_{t}^\top) -   \frac{1}{L_N} \sum_{\ell=0}^{\infty} \mathbf{1}_N^\top \mathbf{B}_\ell \mathbf{B}_\ell^\top \mathbf{1}_N \right|=O_P\left(\frac{1}{\sqrt{T}}\right).
\end{eqnarray*}

\medskip

(3). Note that

\begin{eqnarray*}
&&\frac{1}{L_NT}\sum_{t=1}^T\sum_{v=1}^\infty \mathbf{B}_v^*(L)\text{vec}(\pmb{\varepsilon}_{t-v} \pmb{\varepsilon}_{t}^\top ) \notag \\
&=&\sum_{v=1}^\infty \frac{1}{L_NT}\sum_{t=1}^T \mathbf{B}_v^* \text{vec}(\pmb{\varepsilon}_{t-v} \pmb{\varepsilon}_{t}^\top )+\sum_{v=1}^\infty \frac{1}{L_NT} \widetilde{\mathbf{B}}_v^*(L) \text{vec}(\pmb{\varepsilon}_{T-v}  \pmb{\varepsilon}_{T}^\top)\notag \\
&&+\sum_{v=1}^\infty \frac{1}{L_NT} \widetilde{\mathbf{B}}_v^*(L) \text{vec}(\pmb{\varepsilon}_{0-v}  \pmb{\varepsilon}_{0}^\top).
\end{eqnarray*}

In what follows, we consider the three terms on the right hand side one by one.

Write 
\begin{eqnarray*}
&&E\left|\sum_{t=1}^T\sum_{v=1}^\infty \mathbf{B}_v^* \text{vec}(\pmb{\varepsilon}_{t-v} \pmb{\varepsilon}_{t}^\top ) \right|^2 = \sum_{t,s=1}^T\sum_{v,k=1}^\infty\mathbf{B}_v^* E[\text{vec}(\pmb{\varepsilon}_{t-v} \pmb{\varepsilon}_{t}^\top )\text{vec}(\pmb{\varepsilon}_{s-k} \pmb{\varepsilon}_{s}^\top )^\top]\mathbf{B}_k^{*\top}\notag \\
&=&\sum_{t=1}^T\sum_{v,k=1}^\infty\mathbf{B}_v^* E[\text{vec}(\pmb{\varepsilon}_{t-v} \pmb{\varepsilon}_{t}^\top )\text{vec}(\pmb{\varepsilon}_{t-k} \pmb{\varepsilon}_{t}^\top )^\top]\mathbf{B}_k^{*\top}\notag \\
&\le &O(1)T\sum_{v=1}^\infty\|\mathbf{B}_v^*\|_2^2\leq O(1)NT\left(\sum_{v=1}^\infty\sum_{\ell=0}^{\infty}  \| \mathbf{B}_\ell\|_2 \|\mathbf{B}_{\ell+v}\|_2\right)^2   \notag \\
&\le &O(1)NT \left(\sum_{\ell=0}^{\infty}  \| \mathbf{B}_\ell\|_2 \right)^2.
\end{eqnarray*}

Thus, we have

\begin{eqnarray*}
\left|\frac{1}{L_NT}\sum_{t=1}^T  \sum_{t=1}^T\sum_{v=1}^\infty \mathbf{B}_v^* \text{vec}(\pmb{\varepsilon}_{t-v} \pmb{\varepsilon}_{t}^\top )\right|=O_P(1)\frac{\sqrt{N}}{\sqrt{L_N}T}\sum_{\ell=0}^{\infty}  \| \mathbf{B}_\ell\|_2 =O_P(1)\frac{1}{\sqrt{T}},
\end{eqnarray*}
where the second equality follows from Assumption \ref{AS1}.  Similar to the proof of the second result, it is easy to know that the term $\sum_{v=1}^\infty \frac{1}{L_NT}\sum_{t=1}^T \mathbf{B}_v^* \text{vec}(\pmb{\varepsilon}_{t-v} \pmb{\varepsilon}_{t}^\top )$ offers the lowest rate. Thus, we have 

\begin{eqnarray*}
\left|\frac{1}{L_NT}\sum_{t=1}^T\sum_{v=1}^\infty \mathbf{B}_v^*(L)\text{vec}(\pmb{\varepsilon}_{t-v} \pmb{\varepsilon}_{t}^\top ) \right|= O_P(1)\frac{1}{\sqrt{T}}.
\end{eqnarray*}

The proof is now completed.
\end{proof}

\medskip

\begin{proof}[Proof of Lemma \ref{LM.A7}]
\item  
(1). We now show that

\begin{eqnarray*}
E^*[S_{NT}^{*2}]= \sigma_x^2+o_P(1),
\end{eqnarray*}
which then infers the desired result. It suffices to show that

\begin{eqnarray}\label{var_Rate1}
\frac{1}{L_NT} \sum_{t=1}^T \mathbf{x}_{t}^\top \mathbf{1}_N\mathbf{1}_N^\top \mathbf{x}_{t}= \frac{1}{L_NT} \sum_{t=1}^T E[\mathbf{x}_{t}^\top \mathbf{1}_N\mathbf{1}_N^\top \mathbf{x}_{t}]+o_P(1),
\end{eqnarray}
and 

\begin{eqnarray}\label{var_Rate2}
\frac{1}{L_NT}\sum_{k=1}^{T-1}\sum_{t=1}^{T-k} \mathbf{x}_{t}^\top \mathbf{W}_{ts}\mathbf{x}_{t+k}=\frac{1}{L_NT}\sum_{k=1}^{T-1}\sum_{t=1}^{T-k} E[\mathbf{x}_{t}^\top \mathbf{1}_N  \mathbf{1}_N^\top \mathbf{x}_{t+k}]+o_P(1).
\end{eqnarray}

We start with \eqref{var_Rate1}. By \eqref{def.xx} and Lemma \ref{LM.A6}, we immediately obtain that

\begin{eqnarray} \label{rate.xx}
\frac{1}{L_NT} \sum_{t=1}^T \mathbf{x}_{t}^\top \mathbf{1}_N\mathbf{1}_N^\top \mathbf{x}_{t}=\frac{1}{L_N} \sum_{\ell=0}^{\infty} \mathbf{1}_N^\top \mathbf{B}_\ell \mathbf{B}_\ell^\top \mathbf{1}_N+O_P\left(\frac{1}{\sqrt{T}}\right).
\end{eqnarray}

Therefore, the proof of \eqref{var_Rate1} is completed.

We then consider \eqref{var_Rate2}. Let $a_{k/m}\coloneqq a\left(\frac{k}{m}\right)$ for notational simplicity. Firstly, we consider

\begin{eqnarray}\label{var_Rate4}
&&E\left|\frac{1}{L_NT}\sum_{k=1}^{T-1}\sum_{t=1}^{T-k} (\mathbf{x}_{t}^\top \mathbf{W}_{ts}\mathbf{x}_{t+k}-E[\mathbf{x}_{t}^\top \mathbf{W}_{ts}\mathbf{x}_{t+k}])\right|\notag \\
&=&E\left|\frac{1}{L_NT}\sum_{k=1}^{T-1} a_{k/m} \sum_{t=1}^{T-k} (\mathbf{x}_{t}^\top \mathbf{1}_N \mathbf{1}_N^\top\mathbf{x}_{t+k}-E[\mathbf{x}_{t}^\top \mathbf{1}_N \mathbf{1}_N^\top\mathbf{x}_{t+k}])\right|\notag \\
&=&\sum_{k=1}^{T-1} a_{k/m} E\left|\frac{1}{L_NT} \sum_{t=1}^{T-k} (\mathbf{x}_{t}^\top \mathbf{1}_N \mathbf{1}_N^\top\mathbf{x}_{t+k}-E[\mathbf{x}_{t}^\top \mathbf{1}_N \mathbf{1}_N^\top\mathbf{x}_{t+k}])\right|\notag \\
&=&O_P(1)\sum_{k=1}^{T-1} a_{k/m} \cdot \frac{1}{\sqrt{T}} = O_P\left(\frac{m}{\sqrt{T}}\right),
\end{eqnarray}
where the third equality follows from a development similar to that for \eqref{var_Rate1}, and the last equality follows from Assumption \ref{AS2}. Sequentially, we consider
\begin{eqnarray}\label{var_Rate5}
&&\frac{1}{L_NT}\sum_{k=1}^{T-1}\sum_{t=1}^{T-k} E[\mathbf{x}_{t}^\top \mathbf{W}_{ts}\mathbf{x}_{t+k}] = \frac{1}{L_NT}\sum_{k=1}^{T-1}\sum_{t=1}^{T-k} E[\mathbf{x}_{t}^\top \mathbf{1}_N \mathbf{1}_N^\top\mathbf{x}_{t+k}] \notag \\
&&+\frac{1}{L_NT}\sum_{k=1}^{T-1}\sum_{t=1}^{T-k} (a_{k/m}-1) E[\mathbf{x}_{t}^\top \mathbf{1}_N \mathbf{1}_N^\top\mathbf{x}_{t+k}].
\end{eqnarray}

Note that
\begin{eqnarray}\label{var_Rate6}
&&\left|\frac{1}{L_NT}\sum_{k=1}^{T-1}\sum_{t=1}^{T-k} (a_{k/m}-1) E[\mathbf{x}_{t}^\top \mathbf{1}_N \mathbf{1}_N^\top\mathbf{x}_{t+k}]  \right|\notag \\
&\le & \frac{1}{L_N}\sum_{k=1}^{d_T}  |a_{k/m}-1|\cdot |E[\mathbf{x}_{t}^\top \mathbf{1}_N \mathbf{1}_N^\top\mathbf{x}_{t+k}] |+O(1)\frac{1}{L_N}\sum_{k=d_T+1}^{\infty}   |E[\mathbf{x}_{t}^\top \mathbf{1}_N \mathbf{1}_N^\top\mathbf{x}_{t+k}] |\notag \\
&\le & \frac{1}{L_N} |E[\mathbf{x}_{t}^\top \mathbf{1}_N \mathbf{1}_N^\top\mathbf{x}_{t+1}] |\sum_{k=1}^{d_T}  \frac{k}{m} +O(1)\frac{1}{L_N}\sum_{k=d_T+1}^{\infty}   |E[\mathbf{x}_{1}^\top \mathbf{1}_N \mathbf{1}_N^\top\mathbf{x}_{1+k}] | \rightarrow 0
\end{eqnarray}
by letting $d_T^2/m+1/d_T\to 0$, where the second inequality follows from Assumption \ref{AS2}.

By \eqref{var_Rate4}, \eqref{var_Rate5} and \eqref{var_Rate6}, we have proved \eqref{var_Rate2}.
Up to this point, we have verified both \eqref{var_Rate1} and \eqref{var_Rate2}. Thus, the result in  Lemma \ref{LM.A7}.1 follows.

Using arguments similar to those employed in the proof of Lemma \ref{LM.A7}.1, we can also establish the results in Lemmas \ref{LM.A7}.2-3. The detailed steps are omitted here to avoid unnecessary repetition.
\end{proof}

\medskip

\begin{proof}[Proof of Lemma \ref{LM.A8}]
\item  

For the first result, it suffices to show 

\begin{eqnarray}\label{rate.xtxs_i}
\Big|\frac{1}{T}\sum_{s,t=1}^T\mathbf{e}_i^\top \mathbf{x}_t\mathbf{x}_s^\top\mathbf{e}_i  a\Big(\frac{t-s}{m}\Big)-\sigma_{p,i}^2\Big|=o_P(1),
\end{eqnarray}
for each $i$.

In a manner analogous to \eqref{rate.xx}, we can employ similar arguments as those used in \eqref{def.xx} along with Lemma \ref{LM.A6} and Assumption \ref{AS3} to derive the following result:

\begin{eqnarray}\label{rate.xtxt_i}
&&\Big|\frac{1}{T}\sum_{t=1}^T(\mathbf{e}_i^\top \mathbf{x}_t\mathbf{x}_t^\top\mathbf{e}_i-E[\mathbf{e}_i^\top \mathbf{x}_t\mathbf{x}_t^\top\mathbf{e}_i])\Big|\notag \\
&=&\Big|\frac{1}{T}\sum_{t=1}^T\mathbf{e}_i^\top \mathbf{B}(L)(\pmb{\varepsilon}_t\pmb{\varepsilon}_t^\top-E[\pmb{\varepsilon}_t\pmb{\varepsilon}_t^\top])\mathbf{B}(L)^\top\mathbf{e}_i\Big|
=O_P\left(\frac{1}{\sqrt{T}}\right).
\end{eqnarray} 
Recall that we have defined $a_{k/m}= a\left(\frac{k}{m}\right)$. We can write

\begin{eqnarray}\label{rate.xtxtk_i}
&&\left|\frac{1}{T}\sum_{k=1}^{T-1}\sum_{t=1}^{T-k} a_{k/m}\big(\mathbf{e}_i^\top \mathbf{x}_t\mathbf{x}_{t+k}^\top\mathbf{e}_i-E[\mathbf{e}_i^\top \mathbf{x}_t\mathbf{x}_{t+k}^\top\mathbf{e}_i] \big)\right|\notag \\
&=&\sum_{k=1}^{T-1} a_{k/m} \left|\frac{1}{T} \sum_{t=1}^{T-k} \big(\mathbf{e}_i^\top \mathbf{x}_t\mathbf{x}_{t+k}^\top\mathbf{e}_i-E[\mathbf{e}_i^\top \mathbf{x}_t\mathbf{x}_{t+k}^\top\mathbf{e}_i] \big)\right|\notag \\
&=&O_P\left(\sum_{k=1}^{T-1} a_{k/m} \cdot \frac{1}{\sqrt{T}}\right) = O_P\left(\frac{m}{\sqrt{T}}\right).
\end{eqnarray}

Additionally, using arguments similar to those in \eqref{var_Rate6}, we can further show that 

\begin{eqnarray}\label{rate.xtxtk_i2}
\left|\frac{1}{T}\sum_{k=1}^{T-1}\sum_{t=1}^{T-k}(a_{k/m}-1)E[\mathbf{e}_i^\top \mathbf{x}_t\mathbf{x}_{t+k}^\top\mathbf{e}_i] \right|=o_P(1).
\end{eqnarray}

In light of the results in \eqref{rate.xtxt_i}, \eqref{rate.xtxtk_i} and \eqref{rate.xtxtk_i2}, we are ready to establish \eqref{rate.xtxs_i} and conclude that Lemma \ref{LM.A8}.1 holds. Using similar arguments, we can obtain the desired results in Lemmas \ref{LM.A8}.2-3. The detailed steps are omitted here.
\end{proof}

\end{document}